\documentclass[lettersize,journal]{IEEEtran}
\usepackage{array}
\usepackage{stfloats}
\usepackage{url}
\usepackage{verbatim}
\usepackage{graphicx,subcaption}
\usepackage[T1]{fontenc}
\usepackage{amsmath,amssymb,amsfonts,dsfont}
\usepackage{graphicx}
\usepackage{algorithm,algorithmic,booktabs}
\usepackage{textcomp}

\newcommand{\HIDEFIG}[1]{}

\newtheorem{remark}{Remark}
\newtheorem{problem}{Problem}
\newtheorem{assumption}{Assumption}

\newtheorem{proposition}{Proposition}
\newtheorem{proof}{Proof}
\def\BibTeX{{\rm B\kern-.05em{\sc i\kern-.025em b}\kern-.08em
    T\kern-.1667em\lower.7ex\hbox{E}\kern-.125emX}}
\def\BibTeX{{\rm B\kern-.05em{\sc i\kern-.025em b}\kern-.08em
    T\kern-.1667em\lower.7ex\hbox{E}\kern-.125emX}}
\usepackage{balance}
\begin{document}
\title{Equivalent Circuit Model–based Electric Vehicle Evacuation with Mobile Charging Stations}
\author{Joseph Moyalan, \IEEEmembership{Member, IEEE}, Ricardo de Castro, \IEEEmembership{Senior Member, IEEE}, Shuang Feng, Xuchang Tang, Xinfan Lin, \IEEEmembership{Senior Member, IEEE}, Qijian Gan
\thanks{This work is supported by the California Climate Action Seed Grants of the University of California (Grant No. R02CP6996).
Joseph Moyalan, Shuang Feng, and Ricardo de Castro are with the Department of Mechanical and Aerospace Engineering, UC Merced. Xuchang Tang and Xinfan Lin are with the Department of Mechanical and Aerospace Engineering, UC Davis. Qijian Gan is with Partners for Advanced Transportation Technology (PATH), UC Berkeley.}}

\maketitle
\markboth{IEEE TRANSACTIONS ON INTELLIGENT TRANSPORTATION SYSTEMS}%
{Moyalan \MakeLowercase{\textit{et al.}}: ECM--Based EV Evacuation With Mobile Charging Stations}
\begin{abstract}
The increasing penetration of electric vehicles (EVs) introduces new challenges for emergency evacuation planning due to limited driving range, long charging times, and constrained charging infrastructure, particularly under disaster-induced disruptions. This paper proposes a novel optimization-based evacuation framework for EVs using Equivalent Circuit Models (ECMs) to jointly address routing, charging, and congestion management. By leveraging electrical analogies, traffic flow is modeled as electrical current, travel time as resistance, and driving range as voltage, enabling the use of Kirchhoff’s laws to enforce flow balance and energy feasibility constraints. The proposed controllable ECM incorporates binary switches to regulate route selection and explicitly models charging delays and range replenishment at both Fixed Charging Stations (FCSs) and Mobile Charging Stations (MCSs). The resulting formulation leads to an integer programming problem that determines optimal evacuation routes, charging durations, and the placement and number of MCSs to minimize evacuation time. The framework is extended to multiple origin–destination pairs using the principle of superposition and supports fairness-aware performance metrics, including worst-case, average, and variance-based evacuation times. Simulation studies on large-scale transportation networks in California demonstrate that the proposed approach significantly improves evacuation efficiency and robustness, particularly in scenarios with limited charging access, highlighting the critical role of MCSs in EV-based emergency evacuations.
\end{abstract}

\section{Introduction}
As transportation remains one of the largest contributors to global emissions, accelerating the deployment of Electric vehicles (EVs) is not only a technological imperative but also a critical societal pathway toward a cleaner, healthier, and more sustainable future \cite{dik2022electric,aldhanhani2024future,khaleel2024electric}. However, despite their environmental benefits, EVs still face several disadvantages when compared to conventional gasoline vehicles. EVs are constrained by limited driving range and longer refueling times, as charging, even with fast chargers, typically takes significantly longer than refueling a gasoline vehicle, contributing to range anxiety among users \cite{campana2023optimal}. As reported in \cite{EVIOT}, charging times vary significantly by charger level: Level 1 (120V) requires more than 8 hours to deliver roughly 75–80 miles of range, Level 2 (240V) reduces this to about 4 hours, and Level 3 (480V+) can provide up to 180 miles in less than an hour. Also, the growth of EV adoption continues to exceed the expansion of charging station infrastructure, resulting in limited charging access, particularly in disadvantaged communities \cite{nicholas2019quantifying}. Long charging times and a high EV-to-charger ratio lead to congestion, long waiting times, and uncertainty in trip planning, particularly during peak travel periods or in urban areas with high EV penetration. These challenges are further exacerbated in emergency evacuation scenarios.

Recent years have witnessed a significant increase in mass evacuation events triggered by natural disasters such as tornadoes, hurricanes, and wildfires. For instance, Hurricane Irma (2017) led to the evacuation of more than six million people across Florida, Georgia, and South Carolina in the USA \cite{wong2018understanding}. In 2025, California alone experienced several large wildfires, including the Gifford Fire in San Luis Obispo and Santa Barbara Counties \cite{Gifford_fires}, the Orleans Complex Fire in the Six Rivers National Forest \cite{Orleans_fires}, and the Pickett Fire in Napa County \cite{Pickett_fires}. Earlier urban–wildland interface fires, such as the Eaton Fire \cite{Eaton_fires} and the Palisades Fire \cite{Palisades_fires} in Los Angeles County, caused widespread destruction, forced tens of thousands of evacuations, and underscored the severe risks wildfires pose to densely populated areas, infrastructure, and communities.

Emergency evacuation is a critical component of disaster management, aimed at ensuring the rapid, safe, and orderly movement of populations away from threatened regions. Effective evacuation planning requires coordinated transportation network management \cite{zareafifi2024planning}, real-time traffic control \cite{feng2020can}, and clear communication strategies \cite{moyalan2026optimal}. In the literature, evacuation problems are typically formulated as network design problems \cite{bell1997transportation}. However, traditional evacuation planning approaches are inadequate when additional constraints are considered, such as battery state-of-charge, location dependency, and charging availability of Fixed Charging Stations (FCSs), and power grid limitations. This highlights the need for algorithms that jointly address routing and charging decisions.

Several studies have investigated EV evacuation planning \cite{feng2020can,donaldson2022integration,li2022optimal,yusuf2024evacuation,zheng2010modeling}. In \cite{donaldson2022integration}, EV evacuation charging is modeled as a mobile and spatially redistributed load, showing that synchronized demand spikes can significantly alter power flows and stress transmission networks. To address this, the authors propose an integrated resilience framework combining evacuation modeling with Unit Commitment (UC) and AC Optimal Power Flow (AC-OPF). In \cite{li2022optimal}, a three-stage approach is introduced that consolidates the network through pre-assigned charging stations, ranks nodes based on critical paths, and formulates evacuation planning as a mixed-integer linear program (MILP) to minimize total evacuation time. Similarly, \cite{zheng2010modeling} presents an optimization–simulation framework for no-notice evacuations that accounts for background traffic and congestion. However, these approaches generally assume a fully operational charging infrastructure and do not explicitly consider power grid disruptions that can disable charging stations during disasters and leave EVs stranded, thereby limiting their effectiveness. Moreover, these methodologies identify limited charging station availability and restricted location reachability as major bottlenecks to rapid evacuation. Consequently, the use of Mobile Charging Stations (MCSs) is explored due to their flexibility and location-independent charging capability.

Many studies suggest MCSs, such as battery-equipped trucks \cite{zhang2020mobile,afshar2021mobile,beyazit2024fairness}, can deliver power where it is most needed, including during emergency evacuations, in remote areas, or to support peak demand. A substantial body of literature addresses the optimal deployment of MCSs for various applications \cite{liu2022mobile,wang2019location}. In \cite{liu2022mobile}, a federated learning–based framework is proposed for the proactive placement of idle MCSs in the Internet of Electric Vehicles (IoEV), framing deployment as a prediction-driven decision problem to increase successful charging rates while reducing cost and waiting time. In contrast, \cite{wang2019location} develops a dynamic location–allocation framework that incorporates a stochastic user equilibrium traffic assignment to jointly capture the effects of congestion and charging availability on EV routing decisions. More recently, \cite{beyazit2024fairness} introduces a real-time, traffic-aware MILP that optimizes multiple quality-of-service (QoS) metrics, such as energy satisfaction, delay, and final state-of-charge, while explicitly enforcing fairness among EV users using Jain’s fairness index. Despite the versatility of MCSs across these applications, their use in emergency evacuation remains largely unexplored. A notable exception is \cite{tang2024enhancing}, which proposes a two-stage optimization framework for large-scale EV evacuations that integrates MCS deployment while accounting for heterogeneity in battery capacity and initial state-of-charge. In this approach, energy-feasible evacuation routes are first generated using an augmented network with virtual charging nodes, followed by an MILP that schedules vehicle departures, assigns routes, and deploys MCSs to minimize total evacuation time.

This paper introduces an equivalent circuit model (ECM) framework for EV evacuation planning with Mobile Charging Stations (MCSs). By leveraging analogies from electrical engineering, ECMs enable non-electrical systems to be analyzed using circuit laws, such as Ohm’s and Kirchhoff’s laws, and fundamental components, including resistors, sources, capacitors, inductors, and switches. A well-established example of this approach is in thermal systems modeling, where voltage, current, and resistance are used as analogs for temperature, heat flow, and thermal resistance, respectively \cite{wunsche1997electro}. ECMs have also gained traction in transportation system analysis over the past several decades. In these models, traffic flow is commonly represented as electrical current, while the impedance associated with travel time or distance is modeled as electrical resistance. This analogy has been used to accelerate traffic network simulations, estimate traffic flow distributions, and compute shortest travel paths. Early work, such as \cite{furber1973electrical}, demonstrated the use of electrical circuit analogs to build computational tools for faster traffic simulation. Furthermore, recognizing the uncertainty inherent in driver behavior, several studies \cite{furber1973electrical,wellin1975simulation,sinop2023robust} adopt the principle that traffic tends to follow routes of minimum travel cost, analogous to the flow of electrical current along paths of least resistance, to infer flow patterns in transportation networks.
\begin{figure}
    \centering
    \includegraphics[width=0.99\linewidth]{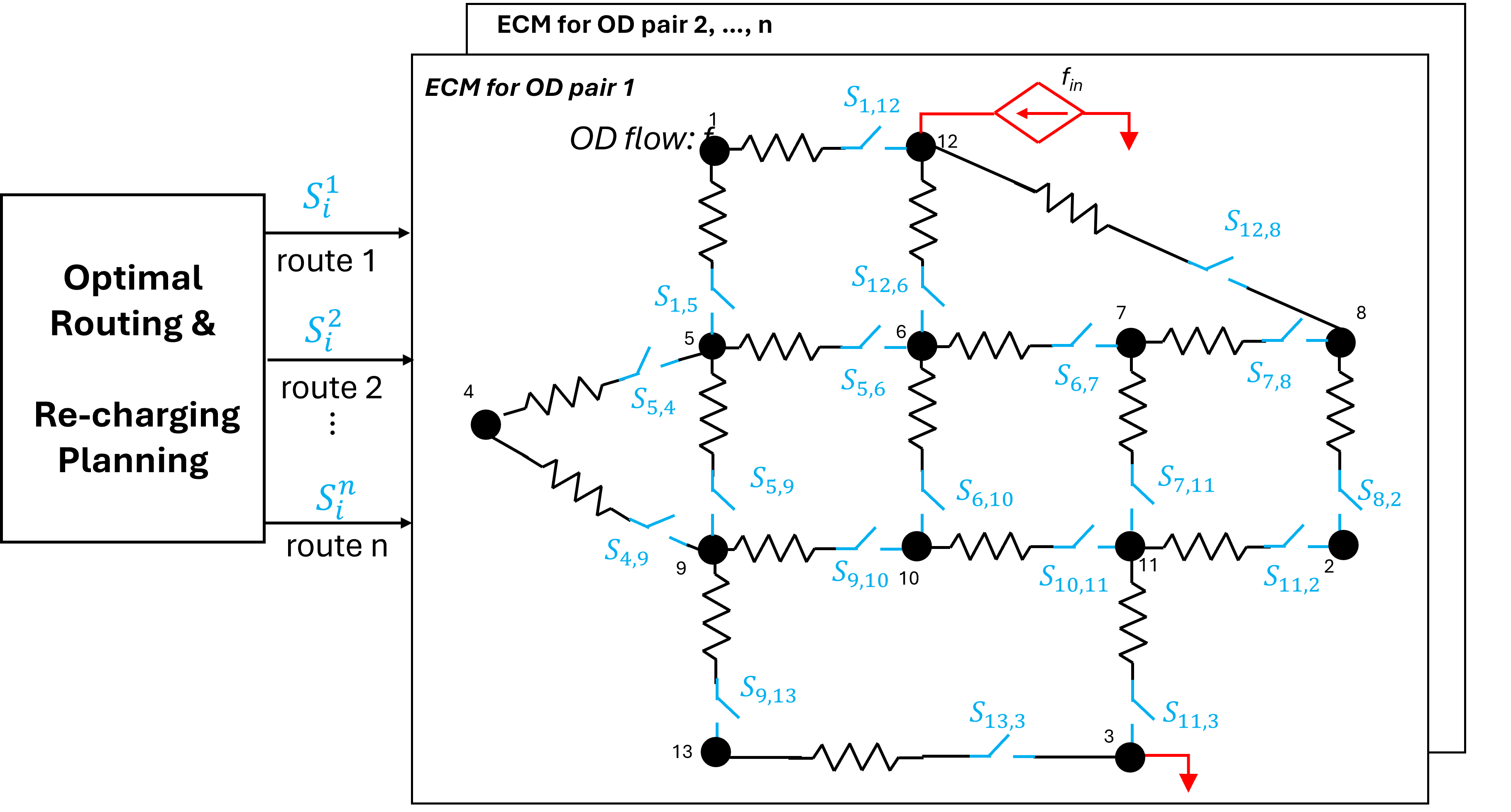}
    \caption{Overview of the evacuation planning with controllable ECMs.}
    \label{fig:controlableECM}
\end{figure}
The main contributions of this paper are as follows:
\begin{enumerate}
    \item We develop a controllable extension of ECM for traffic networks by incorporating additional circuit components, specifically switching elements, that enable active regulation of traffic flow (see Fig.~\ref{fig:controlableECM}). The enhanced circuit representation further integrates components that capture delays associated with vehicle recharging, an aspect that becomes especially important when evacuating EVs whose driving range is more limited than that of conventional internal combustion vehicles~\cite{zhang2022multi,purba2022evacuation}. In addition, we demonstrate how the superposition principle~\cite{alexander2007fundamentals} can be leveraged to decompose and evaluate routing scenarios involving multiple origins and destinations.
    \item We propose a new method to capture the loss of driving range along evacuation routes. This new analogy uses voltage as a proxy for driving range, which allows us to determine vehicles’ energy consumption along the transportation/charging network using Kirchhoff’s voltage law. To the best of the author’s knowledge, this has not been addressed in ECM-based traffic models previously proposed in the literature.
    \item We incorporate MCSs into the ECM-based evacuation framework to enable a flexible and cost-effective solution. The proposed algorithm jointly determines the optimal placement and required number of MCSs at each location to minimize the total evacuation time across all EVs.
    \item The proposed algorithm also determines the optimal charging duration for EVs at each charging infrastructure.
\end{enumerate}

This paper extends our earlier work in \cite{moyalan2026optimal}. In contrast to the prior study, this work introduces the use of MCSs to support evacuation and proposes a novel methodology for modeling the loss of driving range using Kirchhoff’s voltage law. Additionally, the computation of optimal charging durations for EVs at each charging infrastructure is newly developed in this paper.

The remainder of the paper is organized as follows. Section~\ref{section:prelim} presents the preliminaries, notation, and problem formulation. The main results, including the ECM-based formulation of the evacuation problem, are detailed in Section~\ref{section:main}. Simulation results are provided in Section~\ref{section:results}, followed by concluding remarks.

\section{Preliminaries and Problem Statement} \label{section:prelim}
This section introduces the notation and the problem statement of the paper.  
\begin{table}[hb]
\caption{List of Key Parameters of the Evacuation Problem} 
\label{table:symbols}
\centering 
\resizebox{\columnwidth}{!}{
\begin{tabular}{c c}
\hline
\textbf{Symbol} & \textbf{Description} \\ 
\hline
 & \textbf{Transportation Parameters} \\ 

$\mathbb{S}$ & Set of nodes of the transportation network. \\ 
$\mathbb{E}$ & Set of edges of the transportation network. \\
$d_{ij}$ & Distance between the nodes $i$ and $j$.\\
$T_{ij}$ & Free flow traveling time between the nodes $i$ and $j$.\\
$f^{free}_{ij}$ & Free flow capacity of the edge.\\
$k_{ij}$ & Density of the cars on the edge.\\
$u^{avg}_{ij}$ & Average speed of the cars in the edge.\\
\hline
& \textbf{Charging Parameters}\\
$\mathbb{E}_{ch}$ & Set of edge containing charging station. \\
$T_{ch}$ & Minimum charging duration for a charging station.\\
$d_{ch}$ & Minimum driving range gained from a charging station.\\
$f^{max}_{ch,(ij)}$ & Number of EVs per hour that the charger can serve.\\
\hline
& \textbf{Evacuation Requests}\\
$\mathcal{OD}$ & Set of origin-destination pairs.\\
$r^0_{od}$ & Initial driving range of vehicle.\\
$t_{evac,m}$ & Evacuation time for the $m^{th}$ od-pair.\\
\hline

\end{tabular}
}
\end{table}
\subsection{Transportation network}
\begin{enumerate}
    \item {\it Nodes}: $\mathbb{S}=\{s_1, ..., s_{N_S}\}$ denotes the set of vertices in the transportation network. At a local scale, nodes may represent intersections, while at a regional scale, they may correspond to towns or cities connected by the road network.\\
    \item {\it Edges} denote the connections between two nodes $s_i$ and $s_j$ and represent the road segments of the network. Formally, $(s_i, s_j) \in \mathbb{E} \subset \mathbb{S} \times \mathbb{S}$, where $\mathbb{E}$ is the set of all network edges. Similarly, a {\it path} between two nodes $s_1$ and $s_l$ is defined as a sequence of nodes ${\bf p} = [s_1, s_2, \dots, s_l]$ such that $(s_k, s_{k+1}) \in \mathbb{E}$ for all $k = 1, \dots, l-1$.\\
    \item {\it Weights}: The distance (in km) between two nodes is denoted by $d_{ij}$, and the vector $\mathbb{D} = [d_{ij}]$ contains the distances for all node pairs in the network. Similarly, the free-flow travel time (in hours) between two nodes is denoted by $T_{ij}$, and the vector $\mathbb{T} = [T_{ij}]$ captures the free-flow travel times for all node pairs. The maximum flow of vehicles (in vehicles/hour) that an edge can support without congestion is denoted by $f^{free}_{ij}$, and the vector $\mathbb{F}^{free} = [f^{free}_{ij}]$ stores the free-flow capacities of all network edges.\\
    \item {\it Transportation Network} is denoted as a combination of the transportation graph and weights $\Theta_T=\{\mathbb{S},\mathbb{E}, \mathbb{D},\mathbb{T},\mathbb{F}^{free}\}$.
\end{enumerate}
\begin{remark}\label{remark:free_flow_capacity}
    The free-flow capacity ($\mathbb{F}^{free}$) of a link is determined by the maximum achievable vehicle density and the allowable speed on that link. Vehicle density is influenced by factors such as the number of lanes, the minimum spacing between vehicles, and the composition of traffic, including light-, medium-, and heavy-duty vehicles. In contrast, speed limits are typically set based on the characteristics of the surrounding environment, such as residential streets, school areas, or highways. By establishing reasonable upper limits for both vehicle density and travel speed, the free-flow capacity of each link can be estimated \cite{knoop2017introduction}.
\end{remark}
\subsection{Charging network}
\begin{enumerate}
    \item {\it FCS Charging Edges}: $\mathbb{E}_{FCS}\subset \mathbb{E}$ represents a set of edges containing an FCS, where vehicles can recharge their batteries. We denote $N_{FCS}$ as the total number of FCS available in the network.\\
    \item {\it MCS Charging Edges}: $\mathbb{E}_{MCS} \subset \mathbb{E}$ denotes the set of edges where MCS can be deployed to support EV evacuation. Since MCS may also be positioned at FCS locations to accommodate additional charging demand, it follows that $\mathbb{E}_{FCS} \subset \mathbb{E}_{MCS}$. We denote by $N_{MCS}$ the total number of MCS available within the network.\\
    \item {\it Weights}: Vehicles that stop at charging stations remain stationary for a duration $w^{ch}_{ij}T_{ch,(i,j)}$ where $w^{ch}_{ij}$ is a positive integer representing the number of charging intervals and $T_{ch,(i,j)}$ denotes the fixed charging time. For simplicity, we assume this value is constant across all charging infrastructure in the network and denote it by $T_{ch}$. Similarly, the additional driving range gained after charging can be written as a multiple of a fixed range denoted by $d_{ch}$. This corresponds to a charging rate of $d_{ch}/T_{ch}$ [km/h]. In the above discussion and throughout the paper, the subscript $ch$ may refer to (and may be replaced by) either $FCS$, indicating parameters associated with a fixed charging station, or $MCS$, indicating parameters associated with a mobile charging station.\\
    \item Re-charging decisions are defined as
    \begin{equation}
    \begin{split}
        \mathbf{p}_{ch} &=[(s_{ch,1}, s_{ch,2},w_{12}^{ch},v_{12}^{ch}), \ldots, \\
          &\quad\quad\quad(s_{ch,l_p}, s_{ch,l_p+1},w_{l_p,l_{p+1}}^{ch},v_{l_p,l_{p+1}}^{ch})] \nonumber
    \end{split}
    \end{equation}
    such that $(s_{ch,k}, s_{ch,k+1}) \in \mathbb{E}_{FCS}\cup \mathbb{E}_{MCS}$ for all $k$. This vector specifies the charging stations visited by the vehicle during the evacuation. Here, $w_{k,k+1}^{ch}$ specifies the charging duration ($w_{k,k+1}^{ch}T_{ch}$). Also, $v_{k,k+1}^{ch}$ is a non-negative integer that denotes the number of MCS deployed to the charging edge. If $v_{k,k+1}^{ch}=0$, then the charging edge contains an FCS that does not require additional MCS for charging support. If the vehicle does not need to stop for charging, then $\mathbf{p}_{ch} = \emptyset$.\\   
    \item We also assume that each FCS has a limit on the number of EVs it can serve per hour. We use the vector $\mathbb{F}^{max}_{FCS} = [f^{max}_{FCS,(i,j)}]$ to store the maximum number of EVs per hour that each FCS in the network can accommodate. We assume that the service rate of all MCS units is the same for the sake of simplicity. Therefore, $f^{max}_{MCS,(i,j)}=f^{max}_{MCS}$. \\
    \item {\it Charging Network} information is captured by $\Theta_{ch}=\{ \mathbb{E}_{FCS},\mathbb{E}_{MCS},\mathbb{F}^{max}_{FCS},f^{max}_{MCS}\}$. 
\end{enumerate}
\subsection{Evacuation Information}
\begin{enumerate}
    \item We consider an evacuation setting in which vehicles travel from an origin node $o \in \mathbb{S}$ to a predetermined destination or shelter $d \in \mathbb{S}$. Let $\mathcal{OD}$ denote the set of all feasible origin--destination pairs, such that $(o,d) \in \mathcal{OD} \subset \mathbb{S} \times \mathbb{S}$.\\
    \item {\it Initial range}: Vehicles may have different initial driving ranges, denoted by $r^0_{od}$ for each $(o,d) \in \mathcal{OD}$. To account for this variability, we group evacuees according to their initial driving ranges, represented by the vector $\mathbf{r}^0 = [r^0_{1}, r^0_{2}, \ldots]^T \subset \mathcal{R}$, where $\mathcal{R}$ denotes the set of all possible initial ranges. \\ 
    \item The total number of od-pairs ($N_{od}$) is calculated using the set $\mathcal{OD}$ and the vector $\mathbf{r}^0$. This is done because each $(o,d) \in \mathcal{OD}$ can have vehicles with different driving ranges. Therefore, the set of all possible od-pairs is given by $m=(o,d,r^0_{od})\in\mathcal{OD}\times \mathcal{R}$. Each od-pair supplies an incoming vehicle flow denoted by $f_{in,m}$.\\
    \item {\it Loss of driving range}:  In addition, the reduction in driving range along a route must be considered. This reduction is influenced by both the chosen path $\mathbf{p}$ and the charging stations encountered or utilized along that path, denoted by $\mathbf{p}_{ch}$. We define the range evolution as
    \begin{equation}
        \mathbf{r}_{od} = [r_{od,1}, \ldots, r_{od,l}],
    \end{equation}
    where $r_{od,1} = r^0_{od}$ is the initial driving range, and $r_{od,k}$ denotes the remaining range after traversing edge $k$. A path is feasible if $\mathbf{r}_{od} \geq 0$ for all entries.

\end{enumerate}
\subsection{Evacuation problem with constant flows}
\begin{assumption}\label{assume:traffic_info}
    We assume that the transportation parameters ($\Theta_T$), charging infrastructure information ($\Theta_{ch}$), and evacuation-related data are available in advance. Furthermore, all evacuees are assumed to depart simultaneously. 
\end{assumption}
\begin{remark}\label{remark:injection_flow}
    By assuming that all evacuees depart at the same time, we indicate that the evacuation process is initiated nearly simultaneously at the origins associated with every origin-destination (od) pair, resulting in vehicle flow beginning from each origin accordingly. This assumption does not require all vehicles to enter the network at the same instant; instead, it means that the release of traffic into the network starts approximately at the same time for all od-pairs. Such an assumption aligns with the macroscopic evacuation modeling framework commonly adopted in evacuation planning studies \cite{purba2022evacuation,lim2012capacitated}. 
\end{remark}
\begin{problem}\label{prob1}
    Given Assumption~\ref{assume:traffic_info}, find an optimal path 
    $\mathbf{p}_{od} = [s_{1,od}, \ldots, s_{l,od}]$ connecting 
    $(o,d) \in \mathcal{OD}$ and a re-charging strategy 
    $\mathbf{p}_{ch,od}$ that minimize the evacuation time ($t_{evac}$), 
    while satisfying the energy feasibility constraint 
    ($\mathbf{r}_{od} \geq 0$) as well as the traffic 
    ($\mathbb{F}^{free}$) and charging 
    ($\mathbb{F}^{max}_{FCS},f^{max}_{MCS}$) flow constraints.

\end{problem}
Next, we will define the emergency evacuation problem involving multiple origin-destination pairs.
\begin{problem}\label{prob2}
    Given Assumption~\ref{assume:traffic_info}, find an optimal path 
    $\mathbf{p}_{m} = [s_{1,m}, \ldots, s_{l_m,m}]$ and a re-charging strategy 
    $\mathbf{p}_{ch,m}$ that minimize the worst-case evacuation time ($t_{evac}$), 
    while satisfying the energy feasibility constraint 
    ($\mathbf{r}_{od,m} \geq 0$) as well as the traffic 
    ($\mathbb{F}^{free}$) and charging ($\mathbb{F}^{max}_{FCS},f^{max}_{MCS}$) flow constraints. 
    The worst-case evacuation time is defined as
    $$
    t_{evac} = \max(t_{evac,m}), \quad \forall\, m \in \mathcal{OD} \times \mathcal{R},
    $$
    where $t_{evac,m}$ denotes the evacuation time along the optimal path $\mathbf{p}_{m}$.
 
\end{problem}
\begin{figure}
    \centering
    \includegraphics[width=1\linewidth]{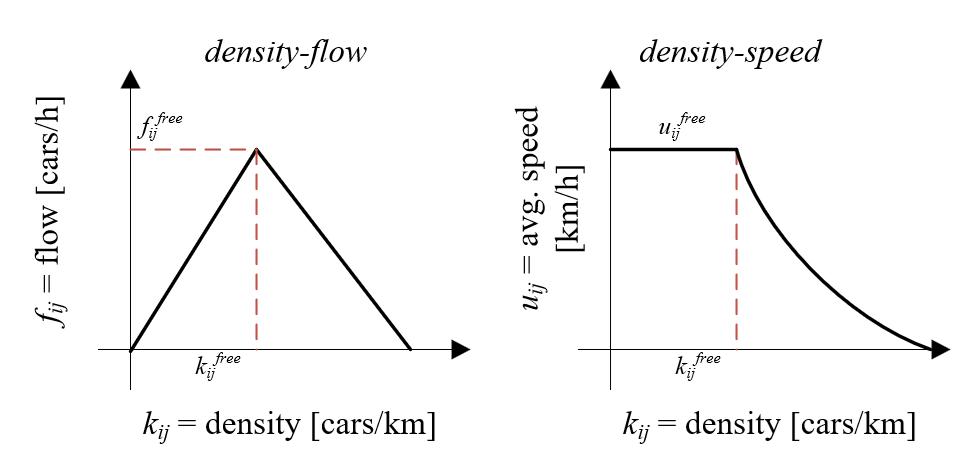}
    \caption{Illustration of the fundamental macroscopic relationship in a traffic link, assuming a triangular shape. }
    \label{fig:FundamentalTranspRelation}
\end{figure}
\begin{figure}
    \centering
    \includegraphics[width=1\linewidth]{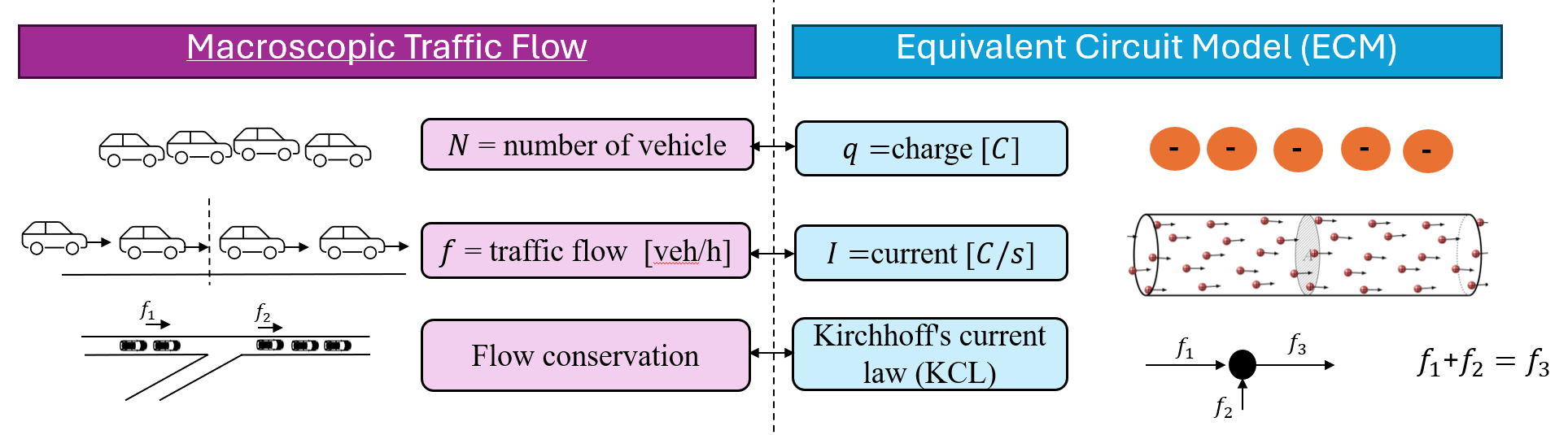}
    \caption{Analogy between traffic network and electrical circuits. }
    \label{fig:parallel_comparison}
\end{figure}
\begin{figure}
    \centering
    \includegraphics[width=0.97\linewidth]{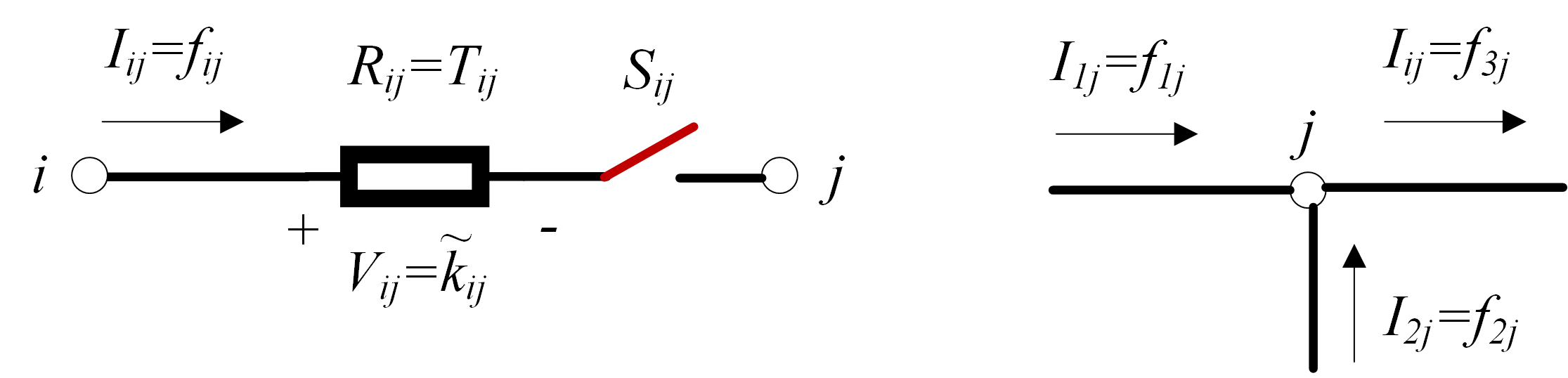}
    \caption{ECM for a transportation link (current -> transportation flow, resistance -> travel time, voltage -> vehicle concentration). }
    \label{fig:ECM-link}
\end{figure}
\section{ECM-based Evacuation Planning} \label{section:main}
This section is an extended version of Section 3 of \cite{moyalan2026optimal}, which discusses the ECM formulation for the evacuation problem. The incorporation of MCS into the ECM-based evacuation framework, the use of Kirchhoff's voltage law to define driving-range constraints, and the optimal calculation of EV charging duration are introduced for the first time in this paper. Recall the fundamental (macroscopic) relationship \cite{knoop2017introduction}  in a traffic link (Fig.~\ref{fig:FundamentalTranspRelation}):
\begin{equation}
f_{ij} = k_{ij} u^{avg}_{ij}(k_{ij}) \label{eq:flow_calc}
\end{equation}
where $k_{ij}$ [in vehicles/km] is the density of cars in the link $i \rightarrow j$, $f_{ij}$ [in vehicles/h] is the flow, and $u^{avg}_{ij}$  [in km/h] is the average speed in the link. Note that the average vehicle speed is generally dependent on the traffic density $k_{ij}$. For instance, at higher vehicle densities, the traffic flow typically decreases due to congestion effects. In this paper, however, we restrict the vehicle flow on each edge of the road network to remain within the free-flow capacity ($f_{ij}^{free}$). Under this assumption, the average speed on each link can be treated as independent of the vehicle density ($k_{ij}$). Next, the previous relationship can also be rewritten as 
\begin{equation}
f_{ij} = \left( k_{ij} d_{ij}\right) \left( u^{avg}_{ij}/d_{ij} \right) = \left( k_{ij} d_{ij}\right)/{T_{ij}}.\label{eq:OhmsLaw_transport}
\end{equation}
The term $k_{ij} d_{ij}$ denotes the total number of vehicles on the link $i\rightarrow j$ at any given time, while $T_{ij}=\left( u^{avg}_{ij}/d_{ij} \right)^{-1}$ represents the expected travel time for the link.

    \begin{figure}
    \centering
    \begin{subfigure}[t]{0.5\textwidth}
        \centering
        \includegraphics[height=1.15in]{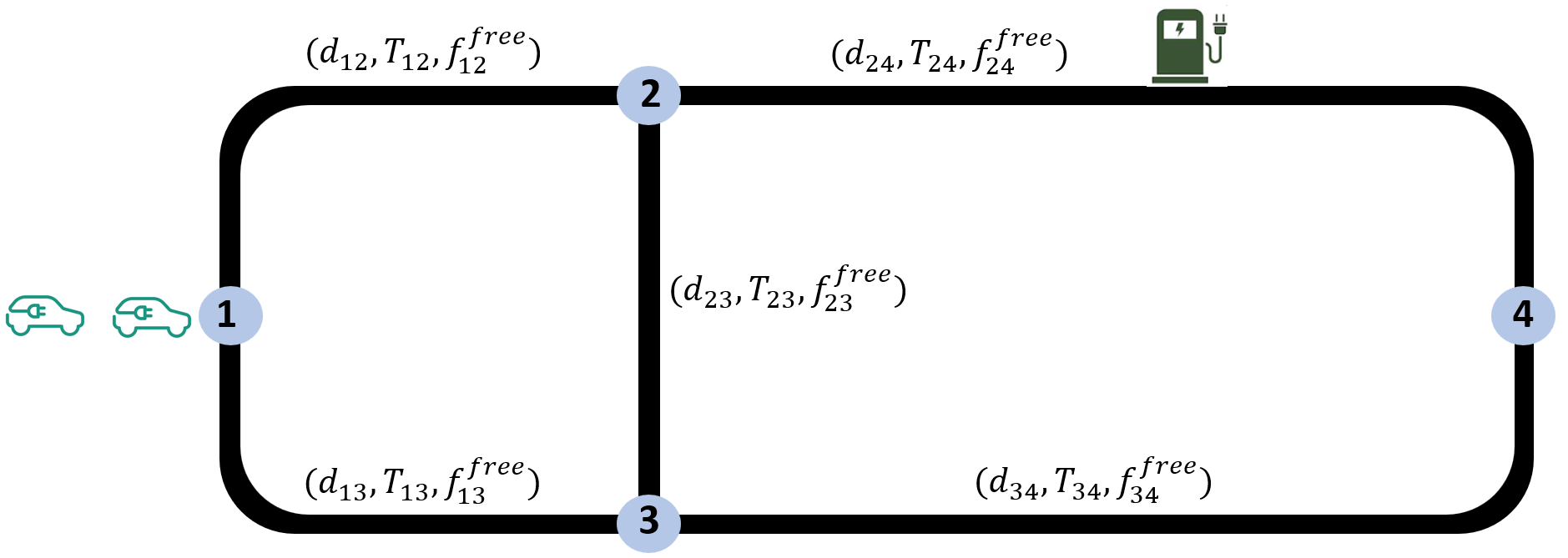}
        \caption{}
    \end{subfigure}
    
    \begin{subfigure}[t]{0.5\textwidth}
        \centering
        \includegraphics[height=1.25in]{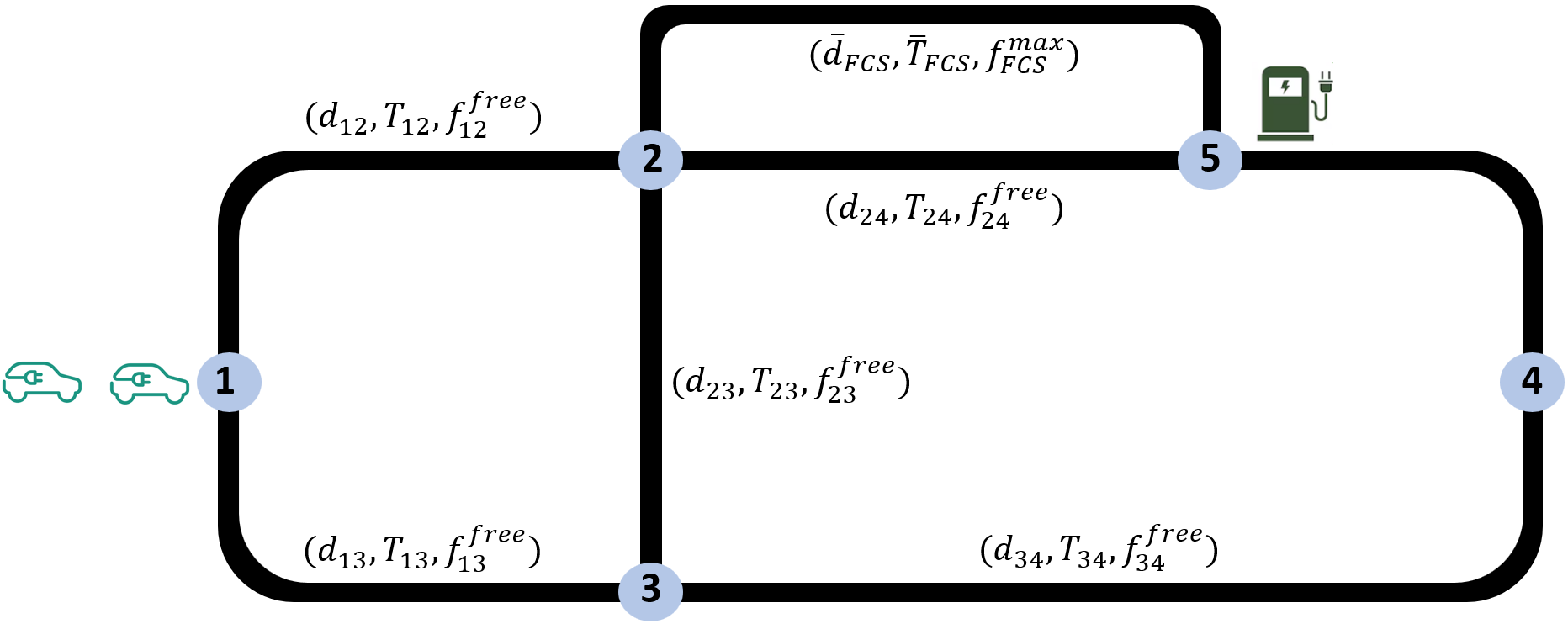}
        \caption{}
    \end{subfigure}
    \begin{subfigure}[t]{0.5\textwidth}
        \centering
        \includegraphics[height=1.3in]{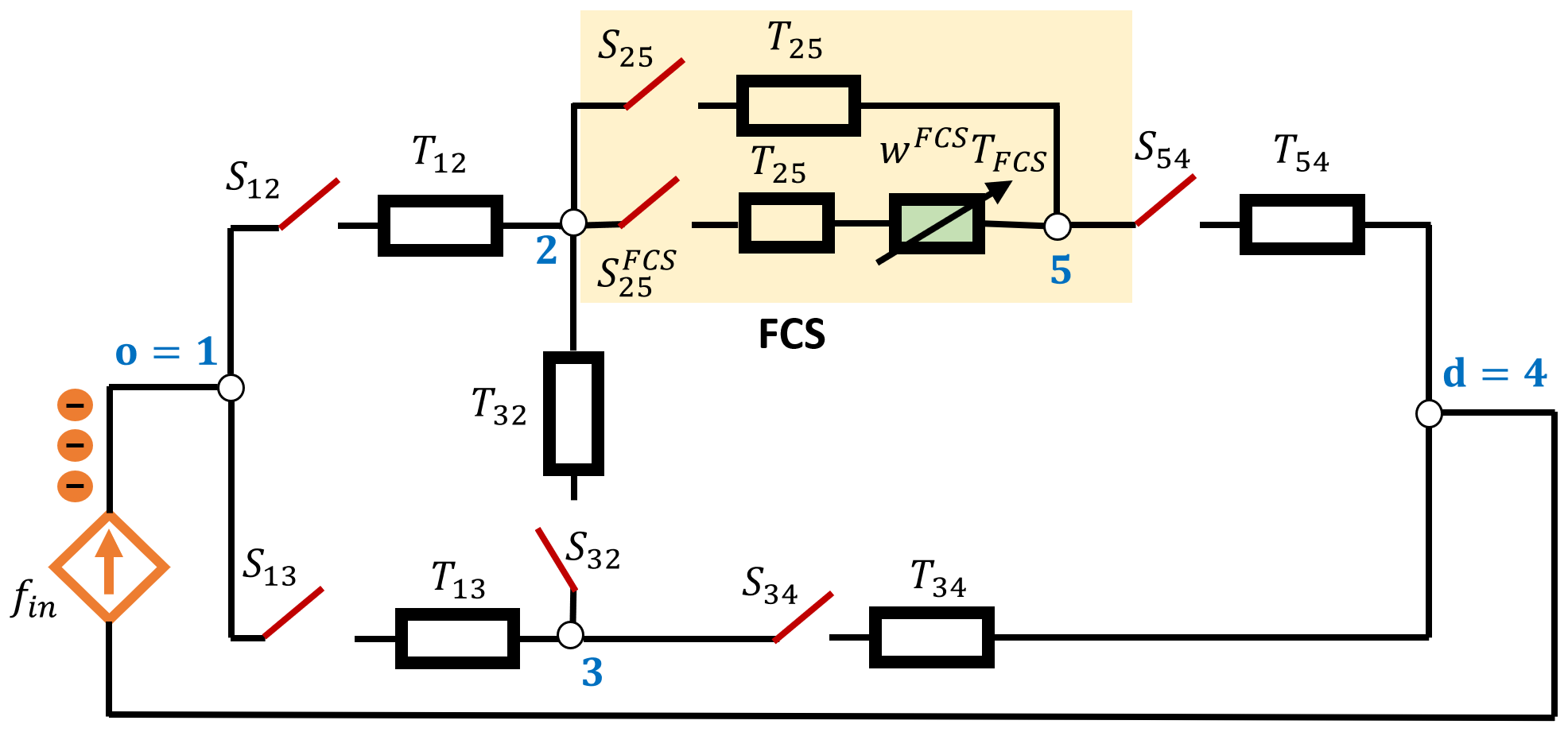}
        \caption{} 
    \end{subfigure}
    \caption{(a) Transportation network used for emergency evacuation, where node 1 represents the origin and node 4 corresponds to the destination. 
        (b) Augmented transportation network that introduces an artificial node 5 to represent the charging station. The two links connecting nodes 2 and 5 capture the alternative decision of either stopping to charge or bypassing the charging station. 
        (c) Equivalent circuit model (ECM) proposed in this work, where resistors correspond to travel time and electrical current represents vehicle flow, with charged particles depicting vehicles. The binary switch variables $S_{ij}$, obtained from the optimization process, determine which network links are activated for evacuation routing.}
    \label{fig:ECM-T-4nodes}
\end{figure}
Consider the node $j$ with three links and corresponding flows $f_{1j},f_{2j},f_{j3}$ (Fig.~\ref{fig:ECM-link}). Links $f_{1j}$ and $f_{2j}$ are "incoming" links, while link $f_{j3}$ is an "outgoing link". They must fulfill the flow balance constraint:
\begin{equation}
    f_{1j} + f_{2j} = f_{3j}. \label{eq:TrafficFlowBalance}
\end{equation}
It should be noted that the relations described above are derived under steady-state traffic conditions. The fundamental traffic equations \eqref{eq:OhmsLaw_transport}--\eqref{eq:TrafficFlowBalance} share an analogous structure with Ohm's law \footnote{Recall that $I = V/R$, where $V$ denotes voltage, $R$ represents resistance, and $I$ is the electric current.} and Kirchhoff's current law (KCL) (see Fig.~\ref{fig:parallel_comparison}), where 
\begin{itemize}
\item Traffic node $j$ is equivalent to an electric node $j$.
\item Traffic link $i \rightarrow j$ is equivalent to an electric branch $ i \rightarrow j $.
\item Vehicle flow $f_{ij}$  is equivalent to current $I_{ij}$.
\item The number of vehicles in edge $d_{ij}$ is equivalent to voltage drop between node $i$ and $j$ ($V_{ij})$. 
\item Travel time $T_{ij}$ is equivalent to a resistance $R_{ij}$.
\end{itemize}
\subsection{ECM with FCS only} \label{section:ECM_model}
Fig.~\ref{fig:ECM-link} illustrates the equivalent circuit that models the flow between nodes $i$ and $j$, reflecting the relationships described in \eqref{eq:OhmsLaw_transport}--\eqref{eq:TrafficFlowBalance}. The circuit representation also includes a switch $S_{ij}\in\{0,1\}$, which is used to control whether current (i.e., traffic flow) is permitted or blocked along the corresponding link.  

Fig.~\ref{fig:ECM-T-4nodes} shows an example transportation network consisting of four links and a single charging station located between nodes 2 and 4. To aid in the derivation of the ECM representation, an artificial node 5 is introduced. This node is connected to node 2 through two separate links and to node 4, as illustrated in Fig.~\ref{fig:ECM-T-4nodes}b. The two links characterized by the parameters $(d_{25}, T_{25}, f^{free}_{25})$ and $(\bar{d}_{FCS}, \bar{T}_{FCS}, f^{max}_{FCS})$ represent the alternative choices available when traveling from node 2 to node 4 via node 5: either bypassing the charging station or stopping to recharge. Here, $\bar{d}_{FCS} = w^{FCS}d_{FCS} - d_{25}$ represents the net distance gained at the FCS, i.e., the additional range provided by the charger (integer multiple of $d_{FCS}$) minus the distance traveled on the corresponding edge. Similarly, $\bar{T}_{FCS} = w^{FCS} T_{FCS} + T_{25}$ denotes the total time spent on that edge, consisting of an integer multiple of the charging time $T_{FCS}$ and the travel time $T_{25}$ for the edge. Here, $w^{FCS}$ takes values from the positive integer set. This formulation can be extended to accommodate any number of FCS by introducing the vectors $\mathbf{\bar{\mathbb{D}}}_{FCS}=[\bar{d}_{FCS,(i,j)}]$ and $\mathbf{\bar{\mathbb{T}}}_{FCS}=[\bar{T}_{FCS,(i,j)}]$. Fig.~\ref{fig:ECM-T-4nodes}c illustrates the corresponding ECM representation consisting of 5 nodes and 7 links. Node $o=1$ serves as the origin and is connected to a current source that injects a continuous flow of vehicles during the evacuation process. The current can then propagate through different branches, which correspond to the available road links in the transportation network. Flow along a branch $(i,j)$ occurs only when the associated switch $S_{ij}$ is activated, indicating that the corresponding path is selected for evacuation. In the case of the two links connecting nodes 2 and 5, the decision to stop at the charging station is governed by the switch $S_{25}^{FCS}$. 
\begin{assumption}\label{assume:no_splitting}
    We assume that, for each origin--destination pair defined as $m=(o,d,r^0_{od})\in\mathcal{OD}\times \mathcal{R}$, all vehicles travel along a single route without splitting into multiple paths. Furthermore, every edge in the transportation network is assumed to be unidirectional.
\end{assumption}

\textbf{Decision variables:}
In what follows, we will treat the following variables as decision variables:
\begin{itemize}
    \item $\mathbf{S}=[S_{12}, S_{13},...]^{\top} \in \mathbb{B}^{N_E} $ is a set of boolean variables that indicate if traffic is allowed to flow in a given non-charging edge. Here, $\mathbb{B}=\{0,1\}$ is a Boolean set and $N_E$ represents edges without the charging stations.
    \item $\mathbf{S}^{FCS}=[S_{12}^{FCS}, S_{13}^{FCS},...]^{\top} \in \mathbb{B}^{N_{FCS}}$ is a set of boolean variables that indicate if traffic is allowed to flow in a given charging edge.
    \item $\mathbf{w}^{FCS}=[w_{12}^{FCS}, w_{13}^{FCS},...]^{\top} \in \mathbb{I}^{N_{FCS}}$ is a set of integer variables that indicate total time spent at the FCS as a multiple of $T_{FCS}$.
    \item $\mathbf{r} = [r_1, r_2, \ldots]^{\top} \in \mathbb{R}^{N_{S}}$ denotes a set of real-valued auxiliary variables representing the available driving range of the vehicle at each node.
\end{itemize}
\textbf{Transportation flow matrix:}
We apply KCL to the nodes of the network shown in Fig.~\ref{fig:ECM-T-4nodes}c, where electrical current is interpreted as vehicle flow scaled by the corresponding link switches. For instance, applying KCL at node 1 yields
\begin{align}
    f_{12} + f_{13} = f_{in}.  \label{eq:flow_temp1}
\end{align}
However, utilizing Assumption~\ref{assume:no_splitting}, we know that
\begin{align}
    f_{12} = f_{in}S_{12},\; f_{13} = f_{in}S_{13} .\label{eq:flow_temp2}
\end{align}
Therefore, utilizing \eqref{eq:flow_temp1} and \eqref{eq:flow_temp2}, we get
\begin{align}
    f_{in} (S_{12} + S_{13}) = f_{in}.\label{eq:flow_temp}
\end{align}
Similarly, applying KCL to all nodes of Fig.~\ref{fig:ECM-T-4nodes}c allows us to obtain the following relations:
\begin{subequations}    
\begin{align}
  (\text{node}\; 1): \quad  &&  f_{in} (S_{12} + S_{13}) = f_{in}, \label{eq:node1}\\
   (\text{node}\; 2): \quad  &&   f_{in}  (S_{12}+S_{32}-S_{25}-S_{25}^{FCS}) =0, \\
   (\text{node}\; 3):    \quad      && f_{in}  (S_{13}-S_{32}-S_{34}) =0,\\
   (\text{node}\; 4):   \quad      &&f_{in}  ( S_{54}+S_{34}) = f_{in},\\
  (\text{node}\; 5):   \quad       && f_{in}  (S^{FCS}_{25} +S_{25} - S_{54} ) = 0. \label{eq:node5}
\end{align}
\end{subequations}
This can be compactly rearranged into a matrix form:
\begin{equation}
    \mathbf{A}_{KCL} \mathbf{S} + \mathbf{A}_{KCL}^{FCS} \mathbf{S}^{FCS} = \mathbf{B}_{KCL} \label{eq:network}
\end{equation}
where $\mathbf{A}_{KCL} \in \mathbb{R}^{ N_{S}   \times N_E}$, $\mathbf{A}_{KCL}^{FCS} \in \mathbb{R}^{N_S  \times N_{FCS}}$, $\mathbf{B}_{KCL} \in \mathbb{R}^{ N_{S}  \times 1}$ are matrices that depend on the configuration of the circuit/transportation network. 

\textbf{Travel time:}
The overall travel time in the network can be estimated by adding the resistances corresponding to the links utilized in the circuit. In other words,
\begin{equation}
    t_{evac} (\mathbf{S}, \mathbf{S}^{FCS},\mathbf{w}^{FCS}) = \mathbb{T}^{\top}\mathbf{S}+T_{FCS}\mathbf{w}^{FCS ^{\top}} \mathbf{S}^{FCS}.
\end{equation}
\begin{figure}
        \centering
        \includegraphics[width=1\linewidth]{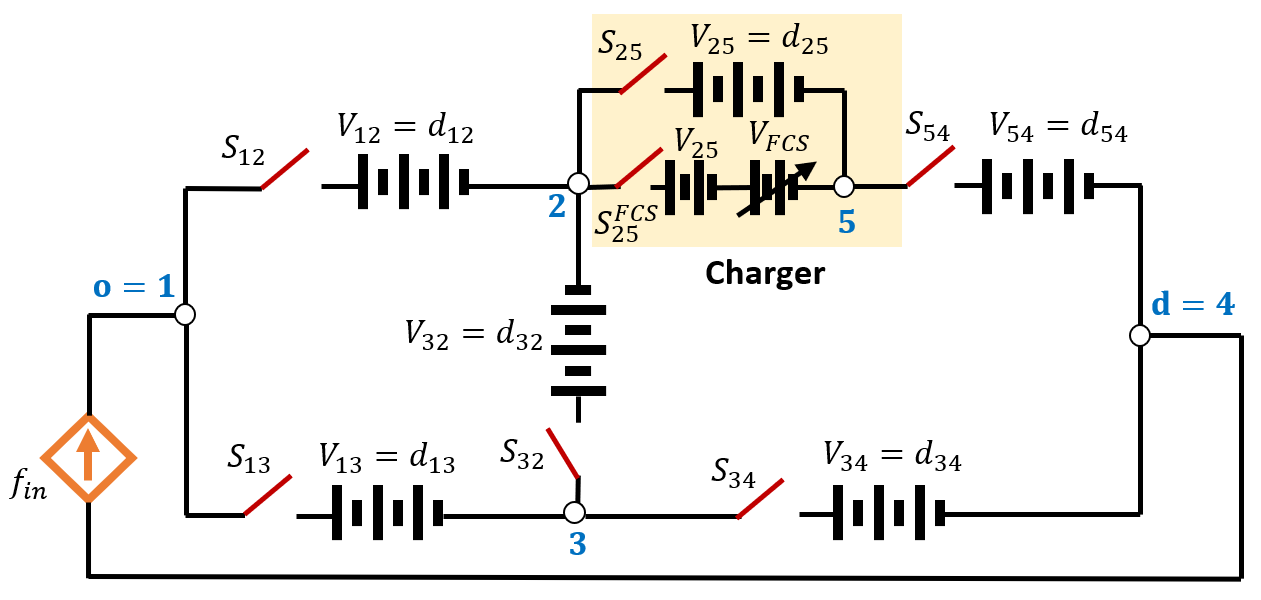}
        \caption{ECM model proposed in this study, where the voltage source represent travel distance of the corresponding edge.}
        \label{fig:voltage_circuit}
    \end{figure}  
\begin{figure}
        \centering
        \includegraphics[width=0.7\linewidth]{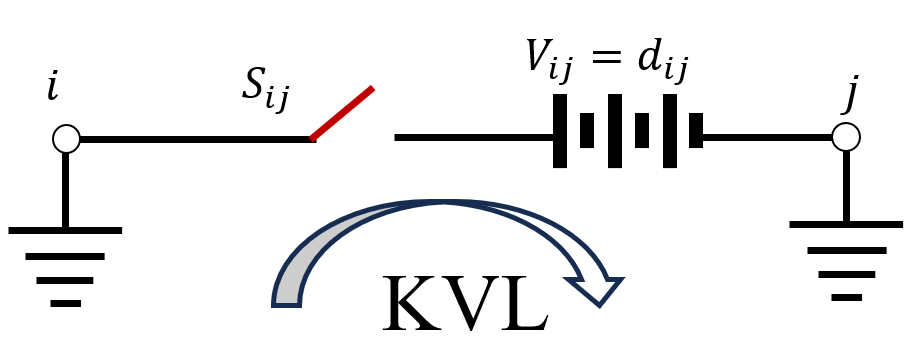}
        \caption{ECM model proposed in this study, where the voltage source represent travel distance of the corresponding edge.}
        \label{fig:kvl_loop}
    \end{figure}  
\textbf{Driving range:} The initial driving range of a vehicle ($r^0$) may be insufficient to reach the destination, necessitating a stop at a charging station. To determine feasibility, we built a second circuit. As illustrated in Fig.~\ref{fig:voltage_circuit}, each edge is modeled as a voltage source whose magnitude corresponds to the edge length. A controllable voltage source is also introduced to represent the additional range obtained at an FCS.

By applying KVL across a loop that contains each edge\footnote{Each edge only appears in one KVL loop.} (Fig.~\ref{fig:kvl_loop}) and multiplying by the corresponding binary switch variable, we obtain the driving-range constraint for the entire network, as shown below:

\begin{subequations}
    \begin{align}
    r_1 = r^0, \\
    S_{12}\left( r_1 - r_2 - d_{12}\right) = 0, \\
    S_{13}\left( r_1 - r_3 - d_{13}\right) = 0, \\
    \vdots \nonumber \\
    S_{25}\left(r_2 - r_5 - d_{25}\right) = 0, \\
    S^{FCS}_{25}\left(r_2 - r_5 - d_{25} + w^{FCS}d_{FCS}\right) = 0.
    \end{align}
\end{subequations}
where $r^0$ corresponds to the initial range of the vehicle starting at the origin (at node 1). This can be written in a general way as follows:
\begin{subequations}\label{eq:KVL}
    \begin{align}
        &\quad\quad\quad r_i = r^0,\; \text{when}\; i\in \mathcal{OD} \\
        &\quad\quad S_{ij}\left( r_i - r_j - d_{ij}\right) = 0,\\
        &\quad\quad S^{FCS}_{ij}\left( r_i - r_j - d_{ij} + w^{FCS}_{ij}d_{FCS,(ij)}\right) = 0,\\
        &\quad\quad\quad\quad\quad\quad \forall \; i,j\in \mathbb{S}.
    \end{align}
\end{subequations}

Using the ECM representation of the transportation network, the flow in each branch can be determined to minimize the upper bound on the total evacuation time. This results in the following optimization problem: 
\begin{subequations}\label{eq:FCS_optimization}
  \begin{align}
 \mathbb{P}^{FCS}:&   \min_{\mathbf{S},\mathbf{S}^{FCS}, \mathbf{w}^{FCS},\mathbf{r}} t_{evac} \left(\mathbf{S}, \mathbf{S}^{FCS},\mathbf{w}^{FCS}\right)\label{eq:FCS_cost}  \\
 & \text{s.t.}  \; \; \;  \eqref{eq:network},\eqref{eq:KVL}, \\
 & f_{in}\mathbf{S}\le\mathbb{F}^{free}, \; f_{in}\mathbf{S}^{FCS}\le\mathbb{F}^{max}_{FCS}. \label{eq:FCS_flow_capacity}
\end{align} 
\end{subequations}
 
The equation \eqref{eq:FCS_flow_capacity} ensures that the free-flow capacity of the roads and the maximum service rate of the FCS are not violated. 
\begin{remark}
     The problem given in \eqref{eq:FCS_optimization} can be solved using Integer programming, which can be tackled using modern solvers such as Gurobi \cite{gurobi}.
\end{remark}
\begin{proposition}\label{prop:unique_path}
    Under Assumption~\ref{assume:no_splitting}, given a single od-pair, the set of enabled switches obtained as the solution of \eqref{eq:FCS_optimization} defines a unique path between the origin and destination nodes in the network.
\end{proposition}
\begin{proof}
    The proof is straightforward. Utilizing Assumption~\ref{assume:no_splitting} and equations \eqref{eq:flow_temp1}-\eqref{eq:flow_temp}, we know that each edge is either selected (carries the full OD flow $f_{in}$) or not selected (carries 0) through the activation of the corresponding switch. Next, consider $\delta_{out,i}$ and $\delta_{in,i}$ to be the set of outgoing and incoming edges for the node $i$, respectively. Next, we observe that
    \begin{align*}
        \sum_{j\in \delta_{out,i}}(S_{ij}+S^{FCS}_{ij})\le 1,\quad \sum_{j\in \delta_{in,i}}(S_{ji}+S_{ji}^{FCS})\le 1.
    \end{align*}
    Therefore, the selected edges form non-branching chains \cite{botnan2024persistent}. This is defined by every node within the path having a degree of 2 (i.e., one edge in, one edge out). We can also rule out cyclic branches as well, since it would not improve the cost function given in \eqref{eq:FCS_cost}. Now, since there is only one origin (source) and one destination (sink) in the network, there must be a single directed chain starting at the origin and ending at the destination.
\end{proof}
\begin{remark}
    The results of Proposition~\ref{prop:unique_path} ensure that a unique path exists between any two nodes in the network. Consequently, in Fig.~\ref{fig:ECM-T-4nodes}c, if the evacuation path passes through nodes 2 and 5, then either $S_{25}$ or $S_{25}^{FCS}$ is activated at any given time, but never both. This configuration allows open switches to account for voltage discrepancies between links connecting the same two nodes.
\end{remark}
\begin{figure}
        \centering
        \includegraphics[width=1\linewidth]{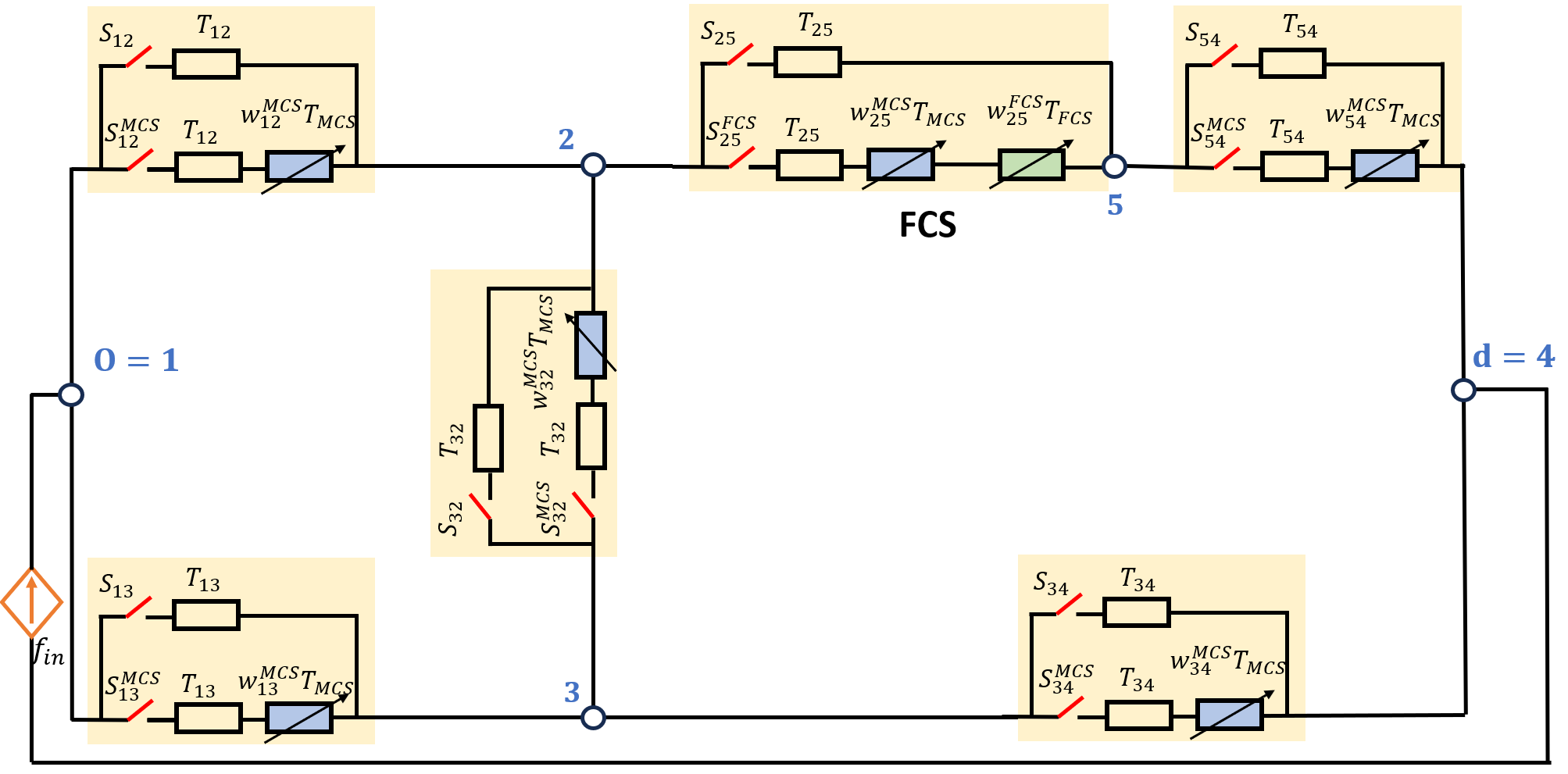}
        \caption{ECM model for the transportation network incorporating both FCS and MCS.}
        \label{fig:FCS_MCS}
    \end{figure}
\begin{figure*}[ht]
        \centering
    \begin{subfigure}[t]{0.49\textwidth}
        \centering
        \includegraphics[width=\textwidth]{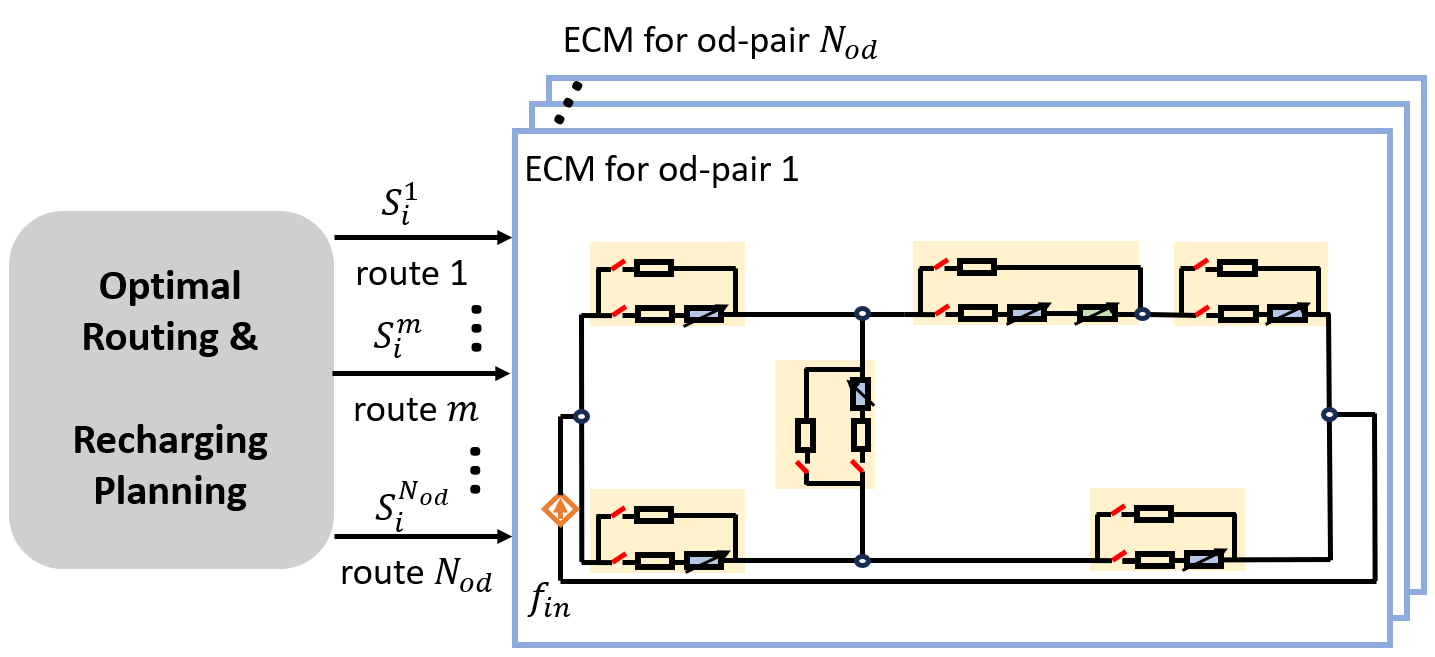}
        \caption{}
    \end{subfigure}
    \begin{subfigure}[t]{0.49\textwidth}
        \centering
        \includegraphics[width=\textwidth]{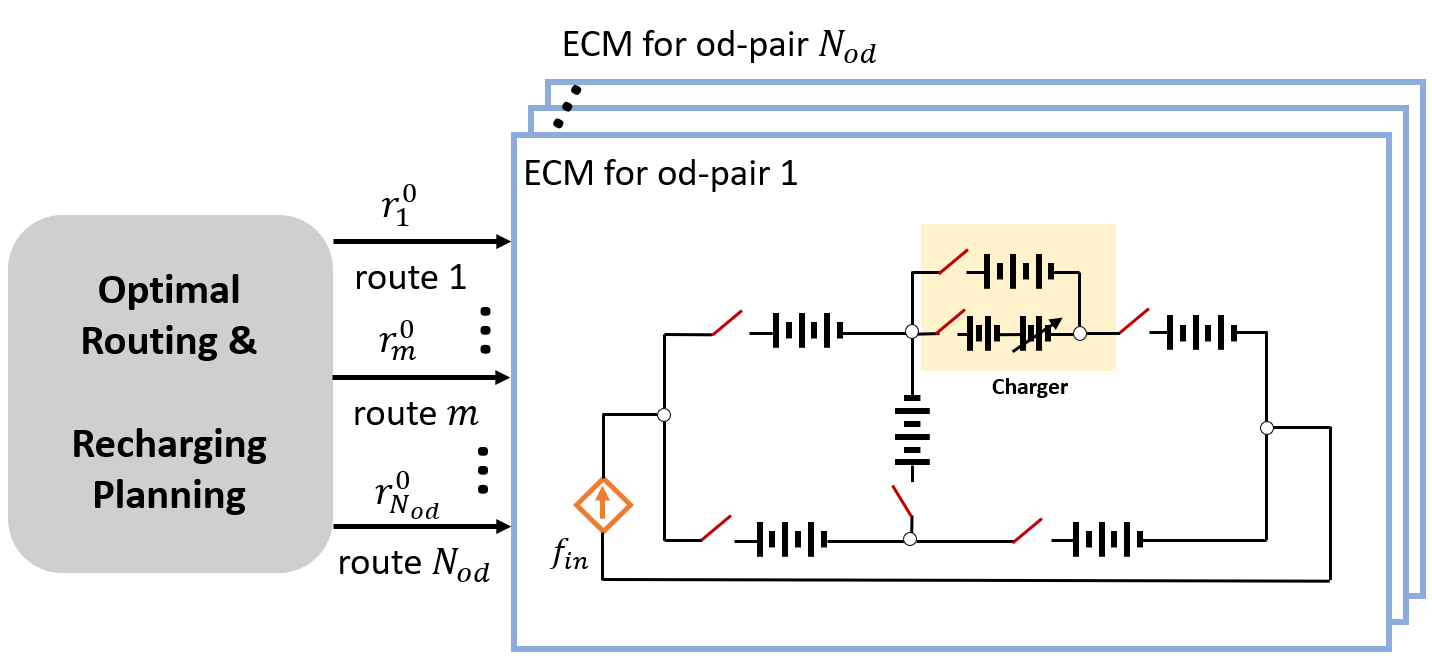}
        \caption{}
    \end{subfigure}
    \caption{(a) ECM model framework for multiple od-pairs. (b) ECM model framework with driving range constraints for multiple od-pairs.}
    \label{fig:multiple_od_ECM_model}
\end{figure*}
\subsection{ECM with FCS and MCS}
In this section, we discuss the incorporation of MCS units into the transportation network to support emergency evacuation. Similar to the methodology used for FCS, we introduce additional parallel edges to the existing network edges. The edge parameters of these new MCS-related edges, denoted by $(\bar{d}_{MCS,(i,j)}, \bar{T}_{MCS,(i,j)}, f^{max}_{MCS,(i,j)})$, encode the charging characteristics at the MCS. Specifically, $\bar{d}_{MCS,(i,j)} = w^{MCS}_{ij} d_{MCS,(i,j)} - d_{ij}$ represents the net distance gained at the MCS, i.e., the additional driving range provided by the charger (integer multiple of $d_{MCS,(i,j)}$) minus the distance traveled on the original edge. Likewise, $\bar{T}_{MCS,(i,j)} = w^{MCS}_{ij} T_{MCS} + T_{ij}$ denotes the total time spent on that edge, which includes an integer multiple of the charging time $T_{MCS}$ and the travel time $T_{ij}$ (see Fig.~\ref{fig:FCS_MCS}). The parameter $f^{max}_{MCS}$ indicates the maximum number of EVs per hour that the MCS can serve.

Each new MCS edge is associated with a binary switch $S^{MCS}_{ij}$ that represents the decision to utilize or bypass the MCS on that edge. One significant advantage of deploying MCS is the ability to position multiple units at the same location to accommodate higher vehicle inflows. Consequently, unlike FCS, the service rate of an MCS-equipped location can be scaled based on the number of available MCS units within the network. Furthermore, the mobility of MCS allows them to supplement existing FCS sites during periods of high queuing and extended waiting times.

Therefore, the additional decision variables required to incorporate MCS into the optimization are as follows:
\begin{itemize}
    \item $\mathbf{S}^{MCS}=[S_{12}^{MCS}, S_{13}^{MCS},...]^{\top} \in \mathbb{B}^{N_{MCS}}$ is a set of boolean variables that indicate if traffic is allowed to flow in a given charging edge.
    \item $\mathbf{w}^{MCS}=[w_{12}^{MCS}, w_{13}^{MCS},...]^{\top} \in \mathbb{I}^{N_{MCS}}$ is a set of integer variables that indicate total time spent at MCS as a multiple of $T_{MCS}$.
    \item $\mathbf{v}^{F-MCS}=[v_{12}^{F-MCS}, v_{13}^{F-MCS},...]^{\top} \in \mathbb{I}^{N_{F-MCS}} \subset \mathbb{I}^{N_{MCS}}$ is a set of integer variables that indicate the total number of MCS units supporting each FCS network edge.
    \item $\mathbf{v}^{I-MCS}=[v_{12}^{I-MCS}, v_{13}^{I-MCS},...]^{\top} \in \mathbb{I}^{N_{I-MCS}} \subset \mathbb{I}^{N_{MCS}}$ is a set of integer variables that indicate the total number of MCS units present at each network edge without an FCS.
    \item $\mathbf{r} = [r_1, r_2, \ldots]^{\top} \in \mathbb{R}^{N_{S}}$ denotes a set of real-valued auxiliary variables representing the available driving range of the vehicle at each node.
\end{itemize}
Now, similar to \eqref{eq:network}, the {\bf transportation flow matrix} including MCS configuration can be written in matrix form as follows:
\begin{equation}
    \mathbf{A}_{KCL} \mathbf{S} + \mathbf{A}_{KCL}^{FCS} \mathbf{S}^{FCS} + \mathbf{A}_{KCL}^{MCS} \mathbf{S}^{MCS} = \mathbf{B}_{KCL} \label{eq:MCS_network}
\end{equation}
where $\mathbf{A}_{KCL} \in \mathbb{R}^{ N_{S}   \times N_E}$, $\mathbf{A}_{KCL}^{FCS} \in \mathbb{R}^{N_S  \times N_{FCS}}$, $\mathbf{A}_{KCL}^{MCS} \in \mathbb{R}^{N_S  \times N_{MCS}}$, $\mathbf{B}_{KCL} \in \mathbb{R}^{ N_{S}  \times 1}$ are matrices that depend on the configuration of the circuit/transportation network.

Similarly, the {\bf travel time} for the new configuration can be calculated by summing up the resistances that are used in the circuit. Therefore,
\begin{align}
    &t_{evac} (\mathbf{S}, \mathbf{S}^{FCS},\mathbf{w}^{FCS},\mathbf{S}^{MCS},\mathbf{w}^{MCS}) = \nonumber \\ &\quad\quad\quad\quad\mathbb{T}^{\top}\mathbf{S}+T_{FCS}\mathbf{w}^{FCS^{\top}} \mathbf{S}^{FCS}+ \nonumber\\&\quad\quad\quad\quad\quad\quad\quad\quad T_{MCS}\mathbf{w}^{MCS^{\top}}  \mathbf{S}^{MCS}.
\end{align}
Similarly, the {\bf driving range constraints} can be formulated by applying KVL across each network edge as follows:
\begin{subequations}\label{eq:KVL_MCS}
    \begin{align}
        &\quad r_i = r^0,\; \text{when}\; i\in \mathcal{OD} \\
        & S_{ij}\left( r_i - r_j - d_{ij}\right) = 0,\\
        & S^{MCS}_{ij}\left( r_i - r_j - d_{ij} + w^{MCS}_{ij}d_{MCS,(ij)}\right) = 0,\\
        & S^{FCS}_{ij}\left( r_i - r_j - d_{ij} + w^{FCS}_{ij}d_{FCS,(ij)}+w^{MCS}_{ij}d_{MCS,(ij)}\right) \nonumber \\
        & \quad\quad\quad= 0,\\
        &\forall \; i,j\in \mathbb{S}.
    \end{align}
\end{subequations}

Let $\bar{S}:=\{\mathbf{S},\mathbf{S}^{FCS}, \mathbf{w}^{FCS},\mathbf{S}^{MCS}, \mathbf{w}^{MCS}, \mathbf{v}^{MCS}, \mathbf{r} \}$. Hence, the problem of determining the branch flows that minimize the upper bound on the evacuation time across the network can be formulated as the following optimization problem:
\begin{subequations}
    \begin{align}
 &\mathbb{P}^{FCS+MCS}:\nonumber\\
 &   \quad \min_{\bar{S}} t_{evac} \left(\mathbf{S}, \mathbf{S}^{FCS},\mathbf{w}^{FCS},\mathbf{S}^{MCS},\mathbf{w}^{MCS}\right)  \label{eq:cost}\\
 & \quad\text{s.t.}  \; \; \;  \eqref{eq:MCS_network},\eqref{eq:KVL_MCS}, \\
 & \quad f_{in}\mathbf{S}\le\mathbb{F}^{free}, \label{eq:road_capacity}\\
 & \quad f_{in}\mathbf{S}^{FCS}\le\mathbb{F}^{max}_{FCS}+f^{max}_{MCS}\mathbf{v}^{F-MCS}, \label{eq:FCS_MCS_capacity} \\
 & \quad f_{in}\mathbf{S}^{MCS}\le f^{max}_{MCS}\mathbf{v}^{I-MCS}, \label{eq:MCS_capacity} \\
 & \sum_i \left(v_i^{F-MCS} + v_i^{I-MCS} \right)\le N_{MCS}. \label{eq:MCS_number}
\end{align} 
\end{subequations}
The equation \eqref{eq:road_capacity} ensures that the free flow capacity of the roads is not violated. Similarly, \eqref{eq:FCS_MCS_capacity} ensures that the charging station's service rate is not violated. Here, $\mathbf{v}^{F-MCS}$ are integer decision variables used to increase the number of MCS units to support the FCS at a particular location. Similarly, \eqref{eq:MCS_capacity} ensures that regular network edges without an FCS can charge EVs by making use of available MCS units deployed to that branch. Here, $\mathbf{v}^{I-MCS}$ represents the number of MCS units deployed to that network edge. Finally, \eqref{eq:MCS_number} ensures the number of MCS deployed does not exceed the available MCS units. We can also aim to minimize the total number of MCS deployed by modifying \eqref{eq:cost} to the following:
\begin{align}
    \min_{\bar{S}} t_{evac} \left(\mathbf{S}, \mathbf{S}^{FCS},\mathbf{w}^{FCS},\mathbf{S}^{MCS},\mathbf{w}^{MCS}\right) \nonumber \\
    + \kappa \sum_i \left(v_i^{F-MCS} + v_i^{I-MCS} \right)
\end{align}
where $\kappa>0$ is a tuning parameter.
\begin{remark}
    The results of Proposition~\ref{prop:unique_path} can be readily extended to networks with both FCS and MCS to establish the uniqueness of the path between the origin and the destination. This is evident from the fact that network constraints given in \eqref{eq:MCS_network} ensure
    \begin{align*}
        \sum_{j\in \delta_{out,i}}(S_{ij}+S^{FCS}_{ij}+S^{MCS}_{ij})\le 1,\\
        \sum_{j\in \delta_{in,i}}(S_{ji}+S_{ji}^{FCS}+S_{ji}^{MCS})\le 1.
    \end{align*}
    
\end{remark}
\subsection{Multiple origins and destinations}
We extend the previous example to networks with multiple origins and destinations. 
Let $m=(o,d,r^0_{od}) \in \mathcal{OD} \times \mathcal{R}$ denote a feasible od-pair. We define $f_{in,m}$ as the flow injected at the origin associated with pair $m$. To address scenarios with multiple od-pairs, we employ the principle of superposition. Specifically, we examine the od-pair $m=(o,d,r^0_{od})$ while setting the inflows of all other od-pairs to zero ($f_{in,l}=0$ for $l \neq m$). In this way, the contribution of each current source is evaluated individually: 
\begin{equation}
     \mathbf{A}_{KCL} \mathbf{S}_m + \mathbf{A}_{KCL}^{FCS} \mathbf{S}^{FCS}_m + \mathbf{A}_{KCL}^{MCS} \mathbf{S}^{MCS}_m = \mathbf{B}_{KCL,m} \label{eq:network_super_pos}
\end{equation}
where $\mathbf{S}_m$, $\mathbf{S}_m^{FCS}$, and $\mathbf{S}_m^{MCS}$ denote the switch variables corresponding to non-charging routes, fixed charging station routes, and mobile charging station routes, respectively, for vehicles associated with the od-pair $m$. The vector $\mathbf{B}_{KCL,m}$ is a column vector representing the injection of a flow $f_{in,m}$ at the node corresponding to the origin $o$ of the od-pair $m \in \mathcal{OD} \times \mathcal{R}$. Fig.~\ref{fig:multiple_od_ECM_model}a presents the physical interpretation of the ECM framework for multiple od-pairs. Each od-pair has its own copy of ECM to find the unique optimal path and recharging strategy to reach the destination from the origin. The {\bf travel time} for each od-pair can be calculated as follows:
\begin{align}
    &t_{evac,m} (\mathbf{S}_m, \mathbf{S}_m^{FCS},\mathbf{w}^{FCS}_m,\mathbf{S}_m^{MCS},\mathbf{w}^{MCS}_m) = \nonumber \\ &\quad\quad\quad\quad\mathbb{T}^{\top}\mathbf{S}_m+T_{FCS}\mathbf{w}_m^{FCS^{\top}} \mathbf{S}^{FCS}_m + \nonumber\\&\quad\quad\quad\quad\quad\quad\quad\quad T_{MCS}\mathbf{w}_m^{MCS^{\top}} \mathbf{S}^{MCS}_m.
\end{align}
\begin{figure*}[ht]
        \centering
    \begin{subfigure}[t]{0.49\textwidth}
        \centering
        \includegraphics[width=\textwidth]{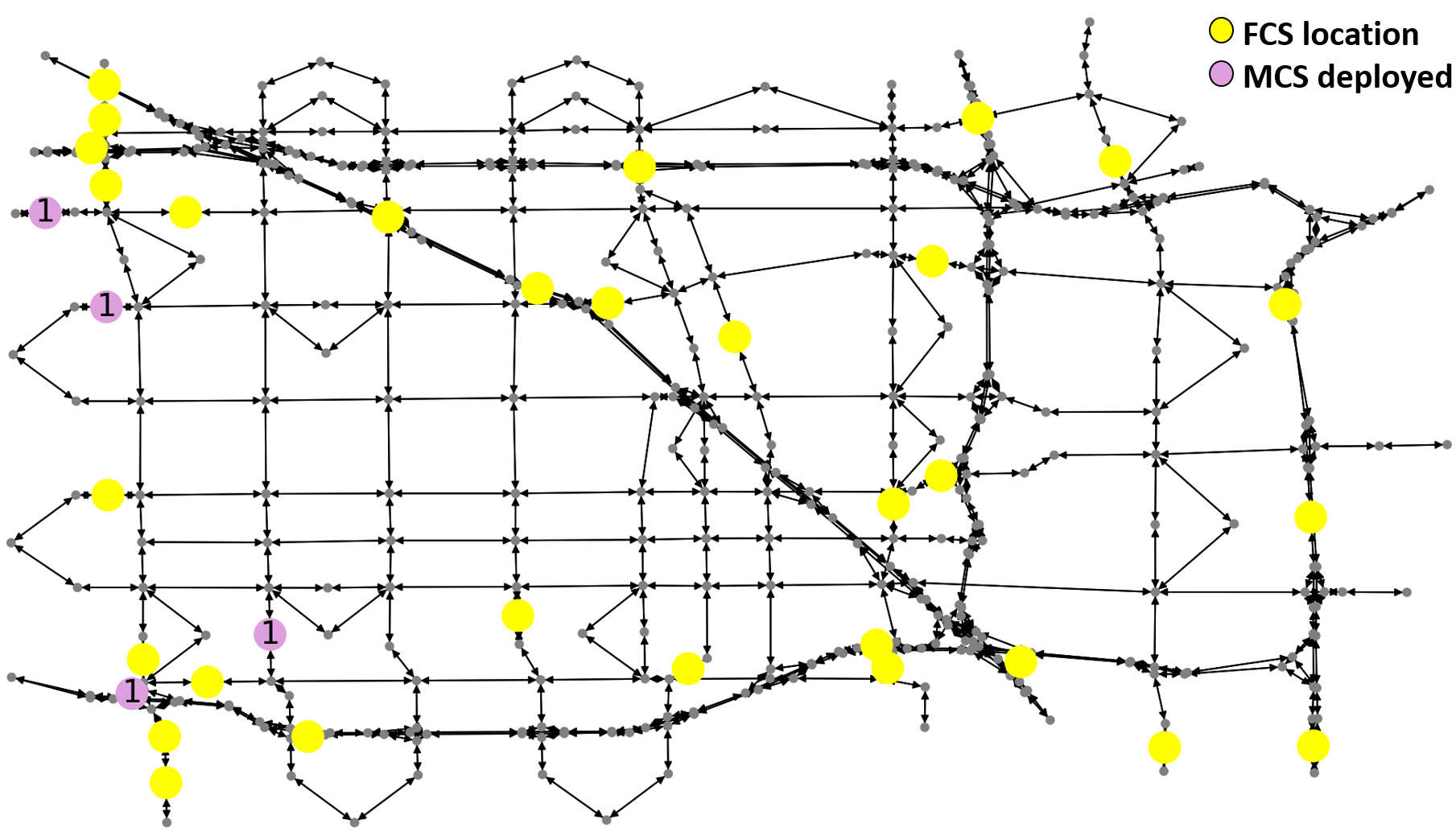}
        \caption{}
    \end{subfigure}
    \begin{subfigure}[t]{0.49\textwidth}
        \centering
        \includegraphics[width=\textwidth]{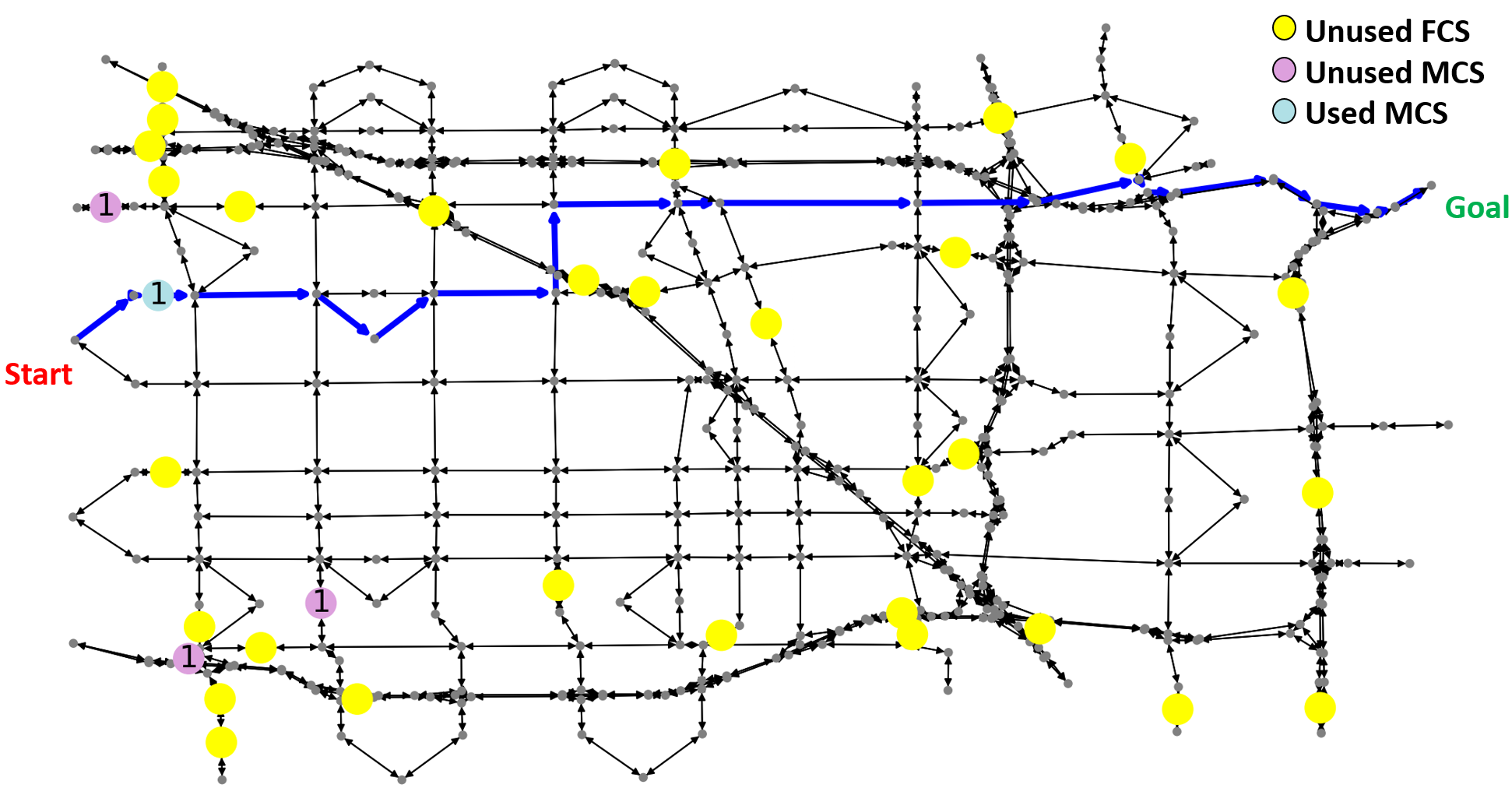}
        \caption{}
    \end{subfigure}
    \caption{(a) Anaheim network showing the locations of FCS (shown as yellow nodes) and deployed MCS (shown as pink nodes). (b) Optimal route (shown in blue) for the od-pair, $m=(22,11)$ and $\mathbb{P}^{avg}$. The vehicles of this od-pair utilize the one deployed MCS (shown in sky blue with "1" written inside to represent 1 MCS) to recharge their EVs before driving towards the goal.}
    \label{fig:anaheim}
\end{figure*}
Similarly, the {\bf driving range constraints} for each od-pair can be written as follows:
\begin{subequations}
    \begin{align}
        &\quad r_{i,m} = r^0_m,\; \text{when}\; i\in \mathcal{OD} \\
        & S_{ij,m}\left( r_{i,m} - r_{j,m} - d_{ij}\right) = 0,\\
        & S^{MCS}_{ij,m}\left( r_{i,m} - r_{j,m} - d_{ij} + w^{MCS}_{ij,m}d_{MCS,(ij)}\right) = 0,\\
        & S^{FCS}_{ij,m}\left( r_{i,m} - r_{j,m} - d_{ij} + w^{FCS}_{ij,m}d_{FCS,(ij)}\right)+ \nonumber \\
        &\quad\quad\quad S^{FCS}_{ij,m}\left(w^{MCS}_{ij,m}d_{MCS,(ij)}\right)= 0,\\
        &\forall \; i,j\in \mathbb{S}.
    \end{align}
\end{subequations}
Fig.~\ref{fig:multiple_od_ECM_model}b presents the physical interpretation of the ECM framework with driving range constraints for multiple od-pairs. The total flow on each transportation link can be obtained by aggregating the contributions from all od-pairs, i.e., $\sum_m f_{in,m}$. This combined flow must satisfy the free-flow capacity constraint associated with the non-charging link, given by
\begin{equation}
    \sum_m f_{in,m}\mathbf{S}_m \leq \mathbb{F}^{free} \label{eq:flow_capacity}
\end{equation} and the total charging capacity of FCS and MCS 
\begin{align}
    \sum_m f_{in,m}\mathbf{S}_m^{FCS} \leq \mathbb{F}^{max}_{FCS}+f^{max}_{MCS}\mathbf{v}^{F-MCS},\\
    \sum_m f_{in,m}\mathbf{S}_m^{MCS} \leq f^{max}_{MCS}\mathbf{v}^{I-MCS}. \label{eq:charger_capacity}
\end{align}
Let $\hat{S}:=\{\mathbf{S}_m,\mathbf{S}_m^{FCS}, \mathbf{w}^{FCS}_m,\mathbf{S}_m^{MCS}, \mathbf{w}^{MCS}_m, \mathbf{v}^{MCS}, \mathbf{r}_m \}$. We can now compute the optimal paths for each od-pair using the following formulation :
\begin{subequations}
    \begin{align}
         &\mathbb{P}^{FCS+MCS+m}: \min_{\bar{S},\gamma} \;\gamma  \\
          & \quad \quad \quad \quad \quad\text{s.t.}  \; \; \;  \eqref{eq:network_super_pos}-\eqref{eq:charger_capacity} \label{eq:central0}\\
          & \quad \quad \quad \quad \quad \quad \quad t_{evac,m} \le \gamma \\
          & \quad \quad \quad \quad \quad \quad \quad m \in \mathcal{OD} \times \mathcal{R} \label{eq:central2}
    \end{align}
\end{subequations}
where $t_{evac,m}$ denotes the estimated evacuation time associated with od-pair $m$, and $\gamma$ is a slack variable. The objective of this formulation is to reduce the maximum evacuation time across all od-pairs. The first group of constraints computes the flows and corresponding evacuation times for each od-pair. 
The final set of constraints applies the superposition principle to determine the total flow on every link and guarantees that this flow remains within the free-flow capacity limits. The parameters $f_{in,m}$ correspond to the input flows and can be adjusted as part of the analysis. So far, the focus has been on minimizing the worst-case evacuation time among all evacuees ($J^{max}=\gamma$). In the following, we introduce two additional performance metrics: the \textit{average evacuation time} ($J^{avg}$) and the \textit{maximum deviation in evacuation time} ($J^{\Delta}$).
\begin{align}
    J^{avg}= \frac{1}{M}\sum_m t_{evac,m}, \quad J^{\Delta} = \max_{m} |t_{evac,m}-J^{avg}|
\end{align}
where $M$ denotes the total number of od-pairs contained in the set $\mathcal{OD} \times \mathcal{R}$. The evacuation objective can then be formulated as a weighted combination of these three performance measures, resulting in the following general formulation for optimal evacuation planning:
\begin{align}
  \mathbb{P}:  &\min w_{max} J^{max} + w_{avg} \theta J^{avg} +w_{\Delta}(1-\theta) J^{\Delta}  \nonumber \\
    &\text{s.t.} \quad  \eqref{eq:central0}-\eqref{eq:central2} \nonumber     
\end{align}
where $w_{max}, w_{avg}, w_{\Delta} \geq 0$ are user-defined weighting coefficients, and $\theta \in [0,1]$ is a trade-off parameter used to regulate the penalty associated with the worst-case variation. This additional term is introduced to promote fairness in route assignment, since evacuees may be unwilling to accept routes that lead to significantly longer evacuation times. Accordingly, this work investigates three evacuation planning strategies that penalize:
\begin{itemize}
    \item $\mathbb{P}^{max}$: max. evacuation time among all od-pairs ($w_{max}=1, w_{avg}=w_{\Delta}=\theta=0$)
    \item $\mathbb{P}^{avg}$: average evacuation time ($w_{avg}=1, w_{max}=w_{\Delta}=\theta=0$)
    \item $\mathbb{P}^{avg+\Delta}$: average and max. variation in evacuation time ($w_{avg}=w_{\Delta}=1, w_{max}=0$)
\end{itemize}
\begin{remark}
    The results of Proposition~\ref{prop:unique_path} can be easily extended to scenarios containing multiple od-pairs. Here, the uniqueness of the path is established within each network copy for the individual od-pairs.  
\end{remark}
\section{Simulation Results}\label{section:results}
To evaluate the proposed evacuation strategies, we consider the transportation networks of Anaheim \cite{Anaheim_network} and Mariposa, both located in California. All simulations are conducted in Python 3.19 on a Dell workstation equipped with 32~GB of RAM and an Intel(R) i9-13900HX processor running at 2.20~GHz. The optimization problems are solved using the Gurobi 12.0 solver. 
\subsection{Anaheim network (Urban Example)}
The Anaheim transportation network consists of 416 nodes and 914 edges. In this section, we consider an emergency evacuation scenario involving 8 evacuation od-pairs. The free flow capacity of each edge, denoted by ($f^{free}_{ij}$), lies in the range of 1800 to 12600 vehicles/hour \cite{Anaheim_network}. The locations of the FCS in the Anaheim network are obtained from \cite{plugshare}. We consider only those FCS with DC fast-charging capability (see Fig.~\ref{fig:anaheim} a), which can deliver a charging rate of ($d_{FCS}/T_{FCS}$) equal to 80 km/hr to support rapid evacuation. We assume that the state of charge (SOC) of vehicles participating in the emergency evacuation corresponds to an initial driving range ($r_0$) between 10 and 80 km (i.e., $2\%-20\%$ of full battery capacity). The EV population of Orange County, California, where the Anaheim network is located, was approximately 134,000 vehicles in 2024 \cite{Anaheim_survey}. Since Anaheim accounts for roughly one-tenth of the county’s population, a proportional EV population of approximately 13,400 vehicles is assumed for the study area. To avoid the adverse effects of the disaster, we assume that the complete evacuation of EVs across the 8 od-pairs must be achieved within 4 hours. This assumption results in an input flow rate ($f_{in}$) of 420 vehicles per hour at each of the 8 od-pairs. We also use 20 MCS to aid in the evacuation procedures of the network. Each MCS has five charging ports and provides a charging rate of 200 km/h.

Fig.~\ref{fig:anaheim} illustrates the optimal evacuation routes for one of the od-pairs obtained by solving $\mathbb{P}^{avg}$. Fig.~\ref{fig:anaheim}a highlights the location of 4 MCS (shown as pink nodes) deployed by the ECM-based optimization to aid in the evacuation process.  Fig.~\ref{fig:anaheim}b presents the optimal route for one of the 8 od-pairs. In this case, the vehicles do not have sufficient SOC to reach any nearby FCS. Consequently, the ECM-based optimization deploys one MCS, shown as a sky-blue node, near the origin of the od-pair to support the evacuation by enabling on-route charging of the EVs. The optimal evacuation paths for all 8 od-pairs are shown in Fig.~\ref{fig:anaheim_full} in the Appendix.
    
\begin{figure}[ht]
    \centering
    \begin{subfigure}[t]{0.5\textwidth}
        \centering
        \includegraphics[height=1.85in]{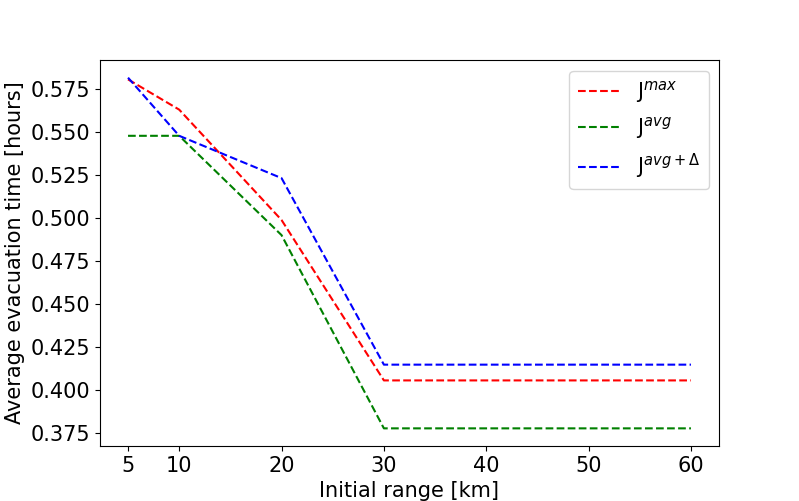}
        \caption{Comparing average $t_{evac}$ for different cost indices.}
    \end{subfigure}
    
    \begin{subfigure}[t]{0.5\textwidth}
        \centering
        \includegraphics[height=1.85in]{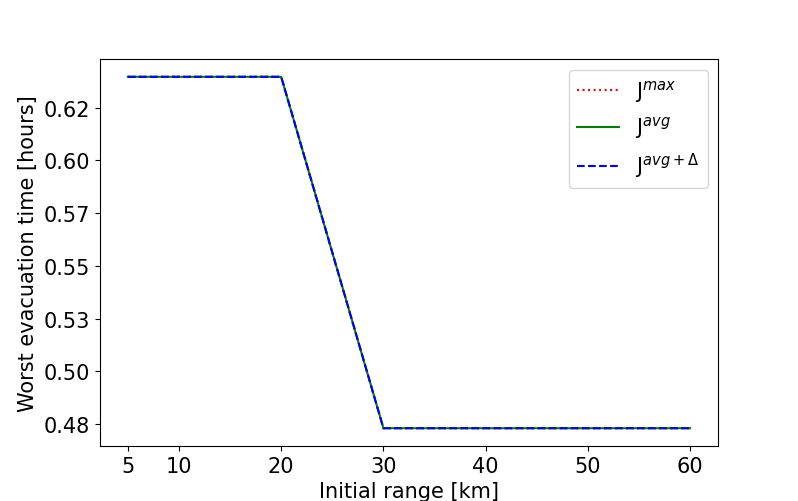}
        \caption{Comparing worst $t_{evac}$ for different cost indices.}
    \end{subfigure}
    \begin{subfigure}[t]{0.5\textwidth}
        \centering
        \includegraphics[height=1.85in]{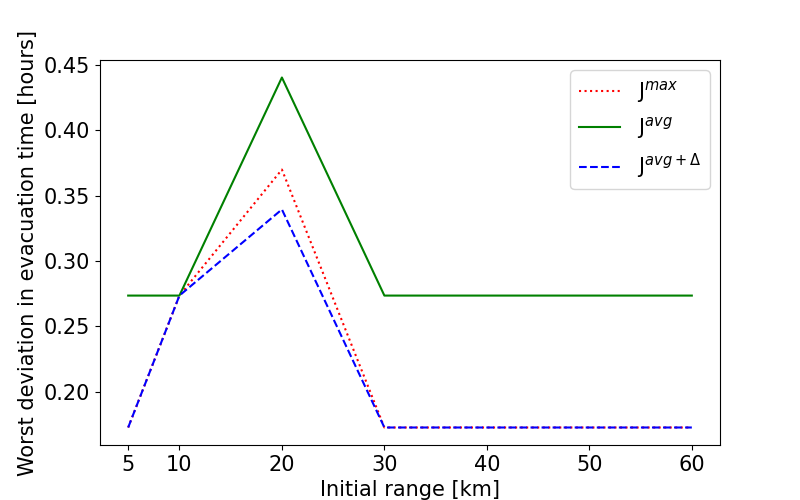}
        \caption{Comparing worst deviation of $t_{evac}$ for different cost indices.} 
    \end{subfigure}
    \caption{ Comparison of different performance indices for the Anaheim network case with respect to the initial range ($r_0$) available to EVs.}
    \label{fig:different_cost_Anaheim}
\end{figure}
Fig.~\ref{fig:different_cost_Anaheim} presents a sensitivity analysis examining the effect of the initial driving range ($r_0$) on the cost functions associated with three evacuation strategies: $J^{\max}$, $J^{\mathrm{avg}}$, and $J^{\mathrm{avg}+\Delta}$. A comparison of the average evacuation times for the different performance indices in Fig.~\ref{fig:different_cost_Anaheim}a shows that using $J^{\mathrm{avg}}$ as the cost function yields the minimum average evacuation time across all values of $r_0$. In particular, at $r_0 = 5$ km, the average evacuation time under $J^{\mathrm{avg}}$ is approximately $4.5\%$ lower than that obtained using $J^{\max}$ and $J^{\mathrm{avg}+\Delta}$. Fig.~\ref{fig:different_cost_Anaheim}b illustrates the worst-case evacuation times for the different performance indices. Notably, these values remain identical regardless of the selected cost function, indicating an inherent physical limitation of the road network in improving the worst evacuation time across all od-pairs. Finally, Fig.~\ref{fig:different_cost_Anaheim}c shows the worst deviation from the average evacuation time, denoted by $\Delta$. As expected, $\Delta$ is minimized when using $J^{\mathrm{avg}+\Delta}$. Specifically, at $r_0 = 20$ km, the value of $\Delta$ under $J^{\mathrm{avg}+\Delta}$ is approximately $8.8\%$ lower than that of $J^{\max}$ and $29\%$ lower than that of $J^{\mathrm{avg}}$. These results demonstrate the effectiveness of $J^{\mathrm{avg}+\Delta}$ in mitigating variability in evacuation times, with the improvement being particularly pronounced at intermediate values of $r_0$.

\begin{figure}[ht]
    \centering
    \begin{subfigure}[t]{0.5\textwidth}
    \centering
    \includegraphics[width=0.9\linewidth]{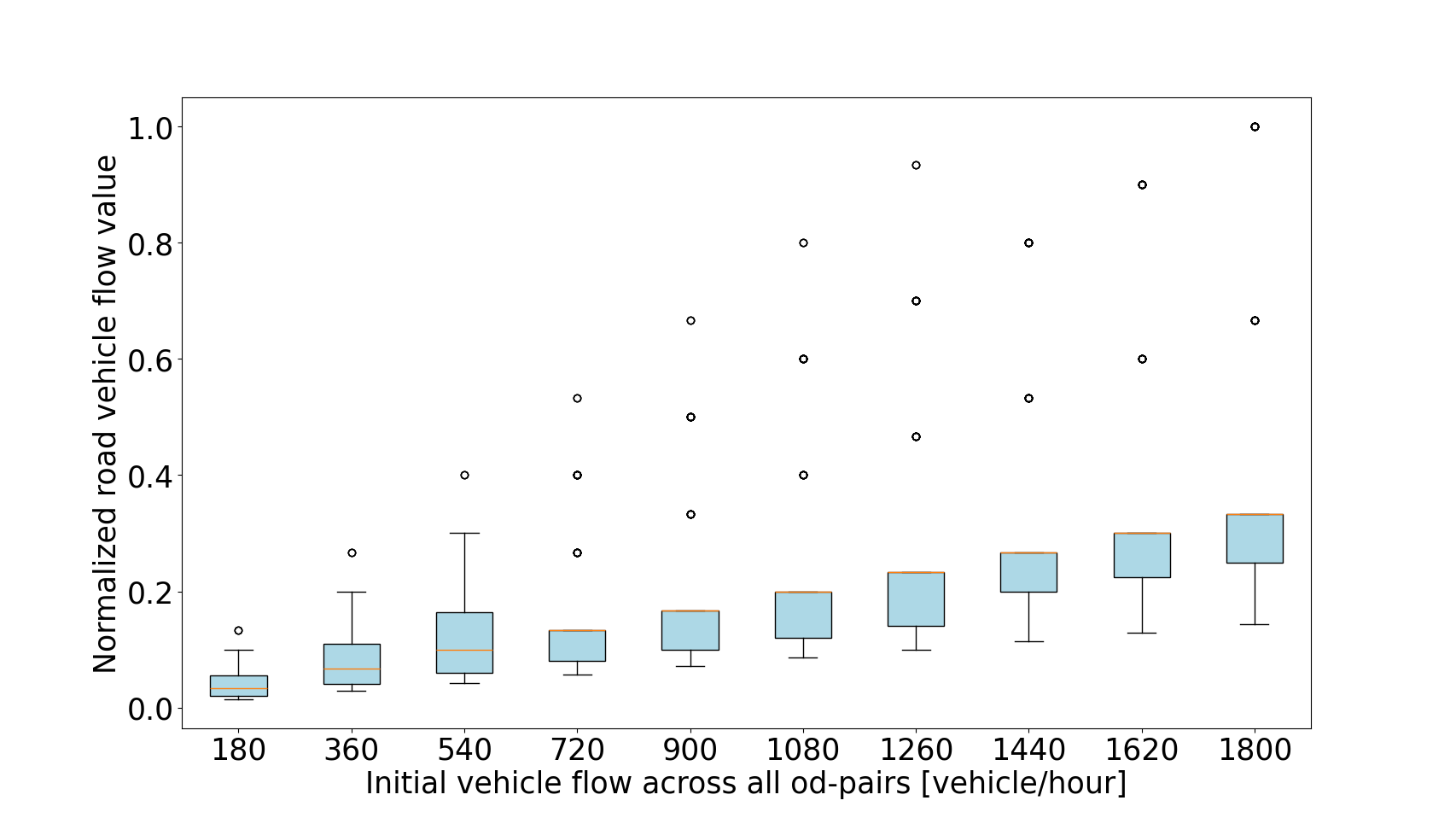}
    \caption{Normalized vehicle flow comparison for different roads of the Anaheim network with respect to the input flow ($f_{in}$). Here, $r_0 = 100$ km.}
    \end{subfigure}
    \begin{subfigure}[t]{0.5\textwidth}
        \centering
        \includegraphics[width=0.9\linewidth]{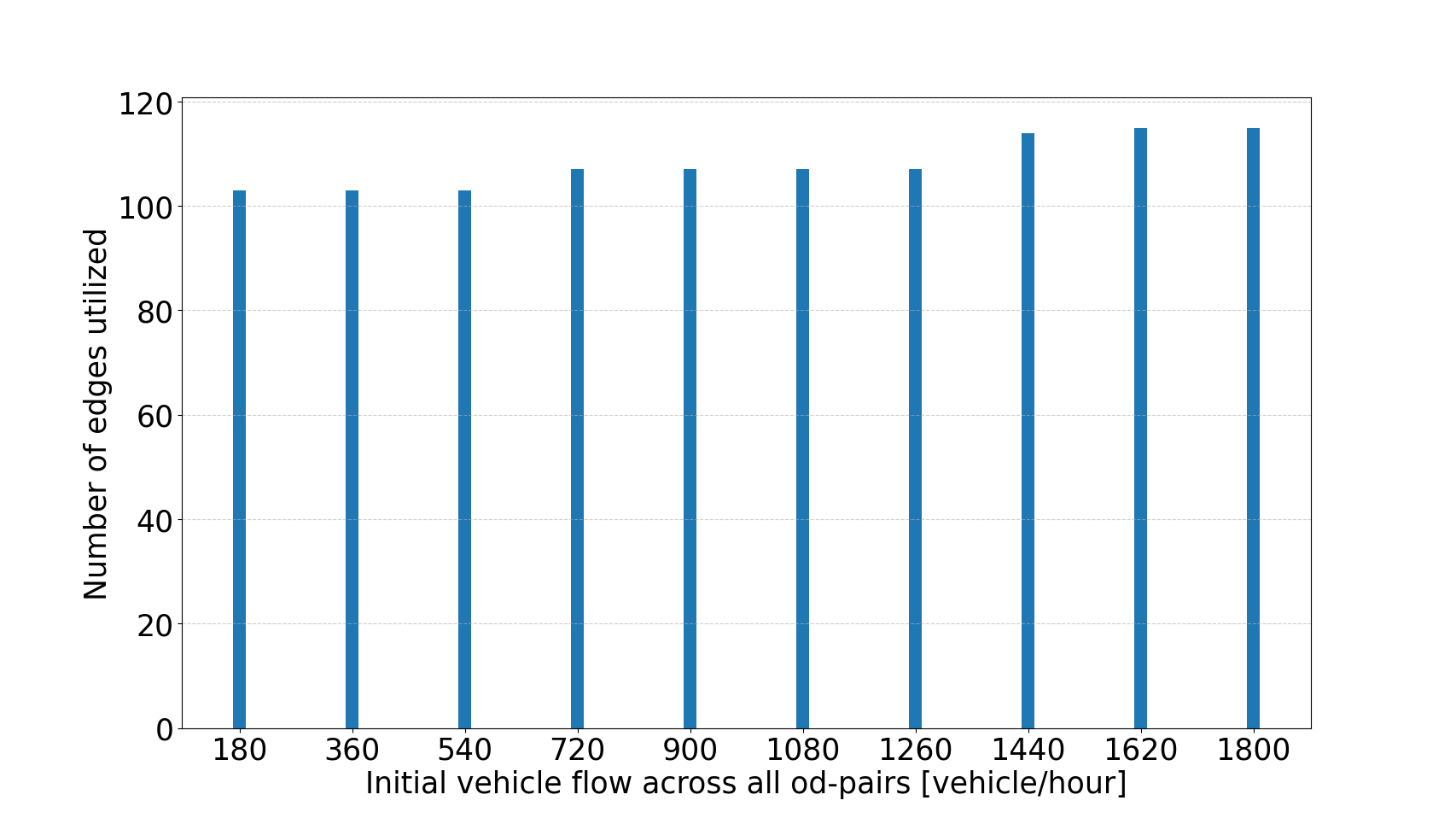}
        \caption{Number of edges used in Anaheim evacuation process for 8 od-pairs with respect to the input flow ($f_{in}$). Here, $r_0 = 100$ km. }
    \end{subfigure}
    \caption{ Sensitivity analysis for different edges for the Anaheim network for various values of input flow ($f_{in}$).}
    \label{fig:Anaheim_edge_numbers}
\end{figure}

Fig.~\ref{fig:Anaheim_edge_numbers} illustrates the impact of the solution of ECM-based optimization with respect to increasing input flow ($f_{in}$) across all od-pairs. Specifically, it showcases the distribution of road utilization and the spatial extent of congestion in the network. The boxplot in Fig.~\ref{fig:Anaheim_edge_numbers}a shows a clear upward shift in the median and range of normalized road vehicle flow as $f_{in}$ increases from 180 to 1800 vehicles/hour, indicating a systematic rise in average edge utilization. At lower demand levels, traffic is concentrated on a small subset of edges with relatively low variability, whereas higher demand scenarios exhibit wider dispersion and a growing number of outliers approaching full capacity, highlighting the emergence of localized congestion. This trend is reinforced by Fig.~\ref{fig:Anaheim_edge_numbers}b, which shows a monotonic increase in the number of edges utilized as demand grows, implying that additional routes are progressively activated to accommodate higher flows. Together, these results suggest that while the network initially absorbs increased demand through underutilized edges, higher demand levels push the system toward widespread utilization of the edges.
\begin{figure}[ht]
    \centering
    \includegraphics[width=0.9\linewidth]{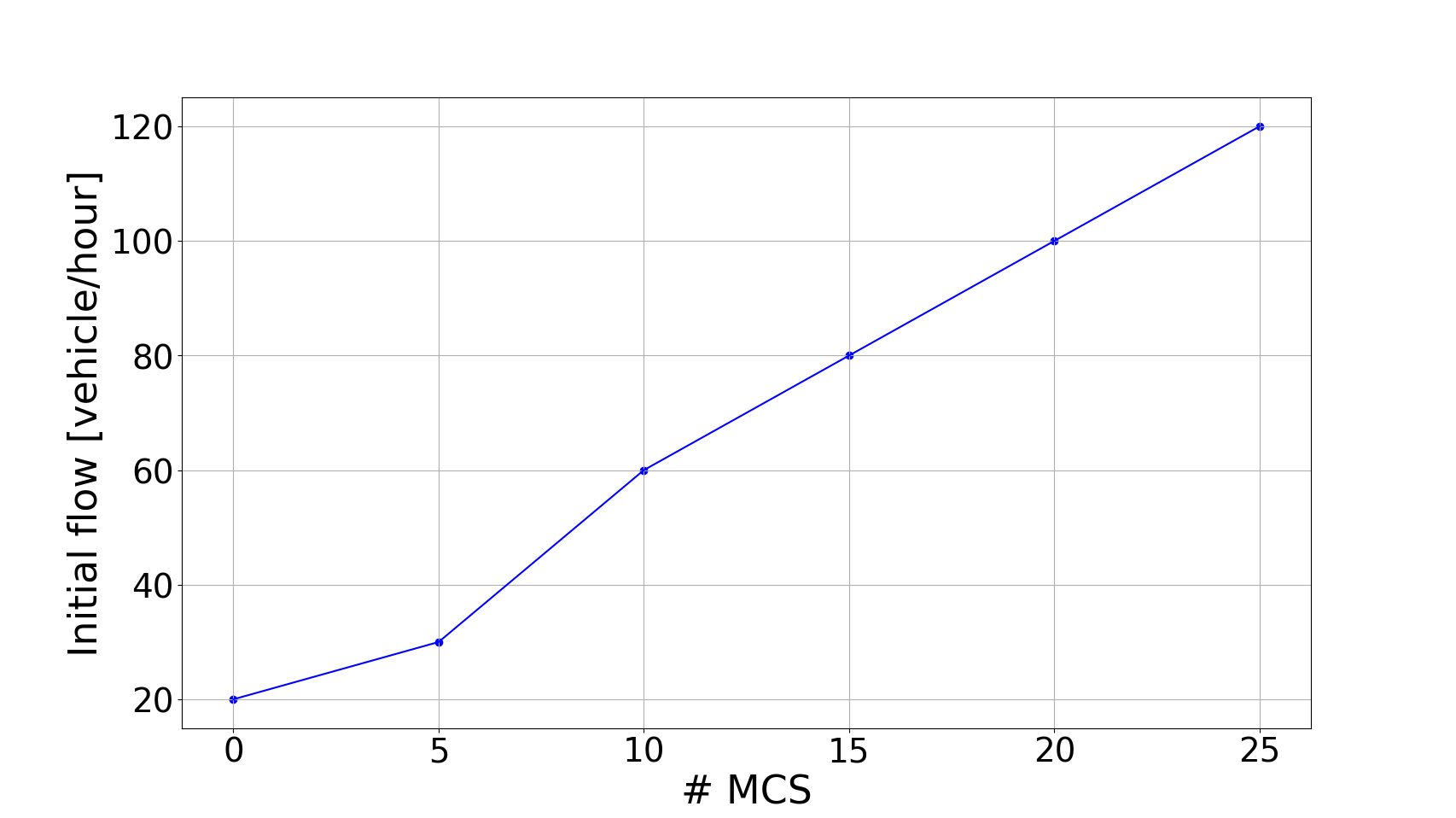}
    \caption{Effect of number of MCS on the initial flow of Anaheim network that can be evacuated without any congestion on roads and queuing on charging stations with 8 od-pairs ($r_0=10$ km). }
    \label{fig:MCS_influence_Anaheim}
\end{figure}
\begin{figure}[ht]
    \centering
    \begin{subfigure}[t]{0.5\textwidth}
        \centering
        \includegraphics[height=1.85in]{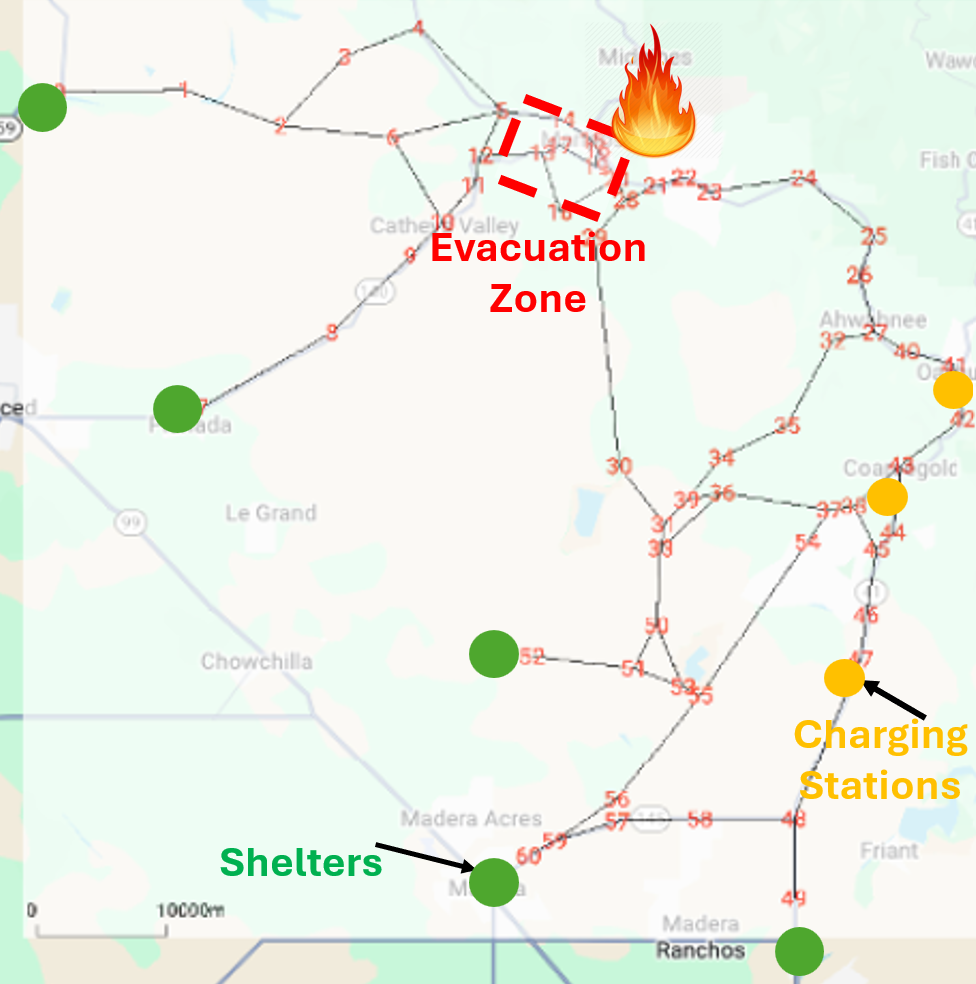}
        \caption{}
    \end{subfigure}
    
    \begin{subfigure}[t]{0.5\textwidth}
        \centering
        \includegraphics[height=1.85in]{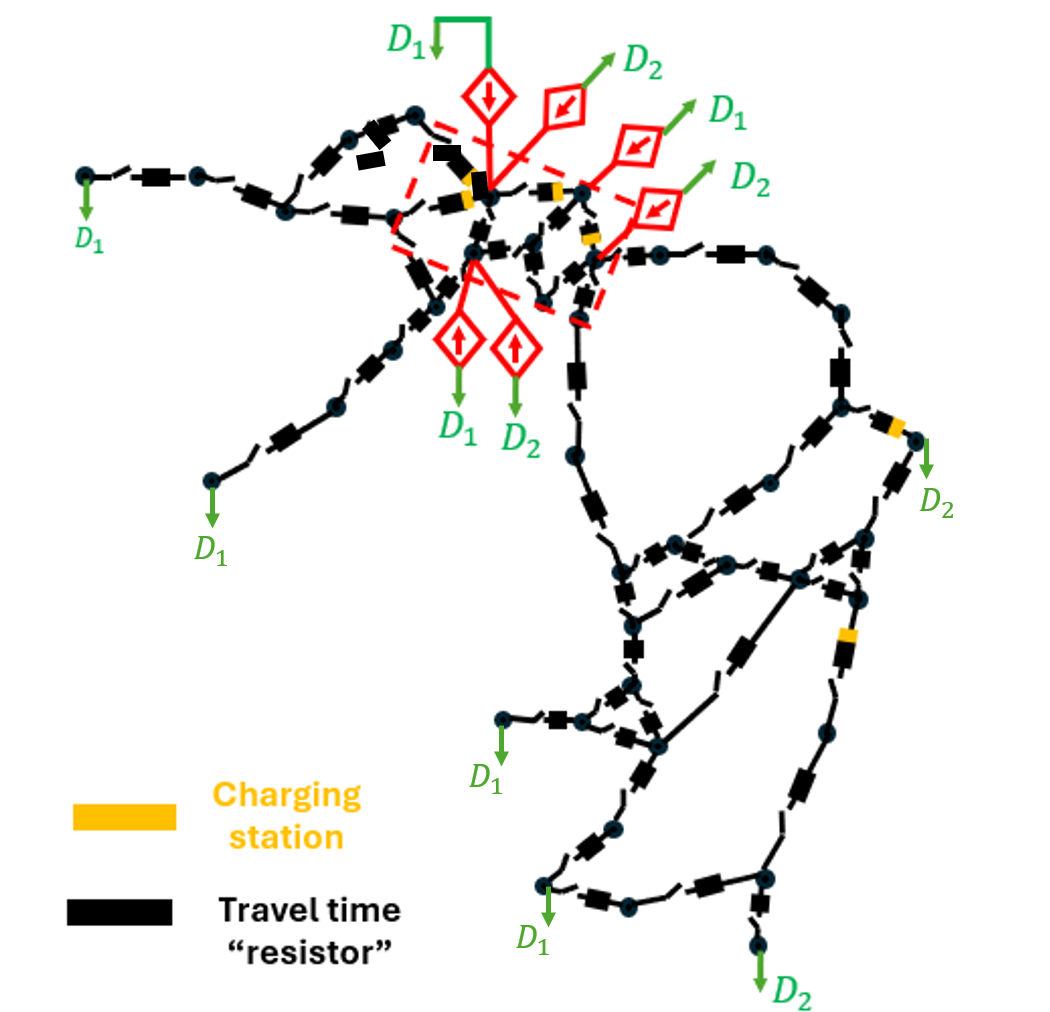}
        \caption{}
    \end{subfigure}
    \caption{(a) Illustration of an evacuation request and transportation network. (b) ECM of the evacuation problem.}
    \label{fig:mariposa_map}
\end{figure}

Fig.~\ref{fig:MCS_influence_Anaheim} illustrates the minimum number of MCS required to complete the evacuation for varying input flow rates ($f_{in}$) across the eight od-pairs of the Anaheim network, assuming an initial driving range of 10 km for all EVs. The results show that for $f_{in}\leq 20$ vehicles/hour, the existing FCS are sufficient to meet charging demands without inducing congestion or queuing within the network. However, when $f_{in}>20$ vehicles/hour, the minimum number of required MCS increases markedly to prevent congestion and excessive waiting times. This behavior arises from the fixed locations and limited service capacities of FCS. In contrast, the flexibility and location-independent charging capability of MCS enable them to dynamically alleviate localized bottlenecks, highlighting their critical role in high-demand evacuation scenarios.  


\begin{figure}[ht]
    \centering
    \includegraphics[width=1\linewidth]{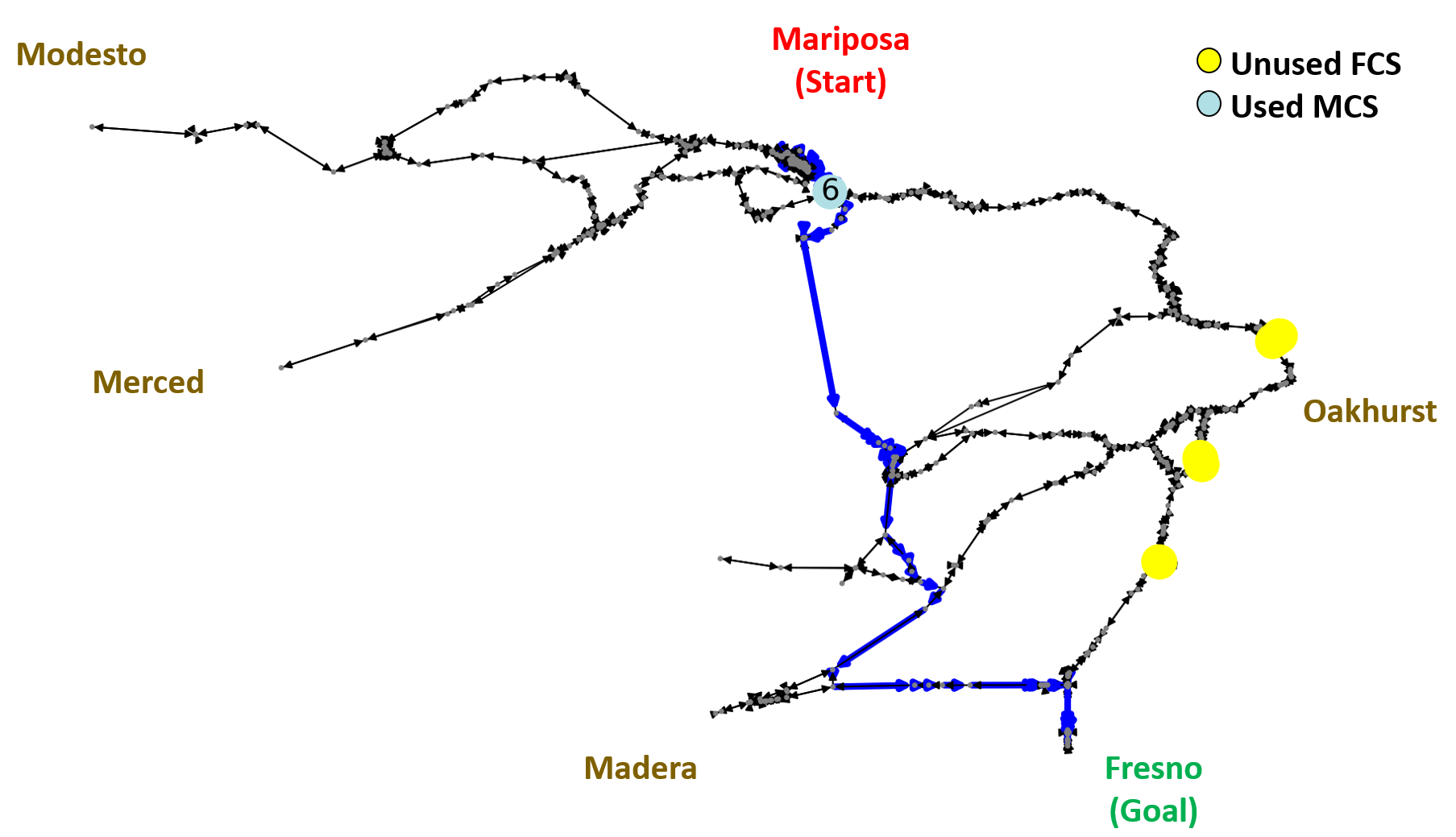}
    \caption{Mariposa network showing the optimal route (shown with blue) for the od-pair from Mariposa downtown towards Fresno using the optimization $\mathbb{P}^{avg}$. Here, the yellow nodes represent all the fast charging FCS locations on the network. The vehicles of this od-pair utilize the two deployed MCS (shown in sky blue with "6" written inside to represent 6 MCS) to recharge their EVs before driving towards the goal.}
    \label{fig:Mariposa_plot}
\end{figure}
    
\begin{figure}[ht]
    \centering
    \begin{subfigure}[t]{0.5\textwidth}
        \centering
        \includegraphics[height=1.85in]{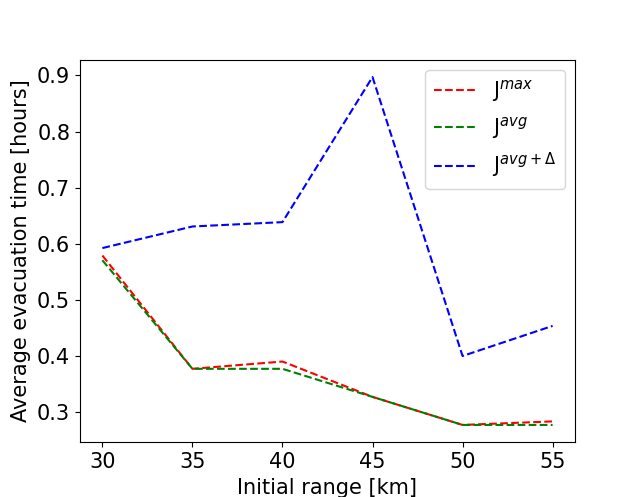}
        \caption{Comparing average $t_{evac}$ for different cost indices.}
    \end{subfigure}
    
    \begin{subfigure}[t]{0.5\textwidth}
        \centering
        \includegraphics[height=1.85in]{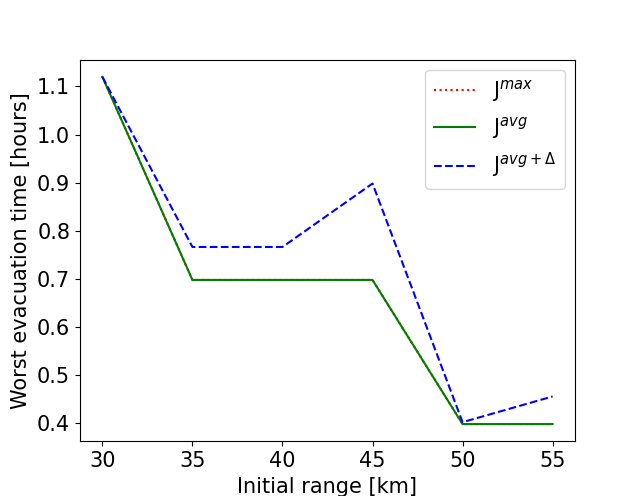}
        \caption{Comparing worst $t_{evac}$ for different cost indices.}
    \end{subfigure}
    \begin{subfigure}[t]{0.5\textwidth}
        \centering
        \includegraphics[height=1.85in]{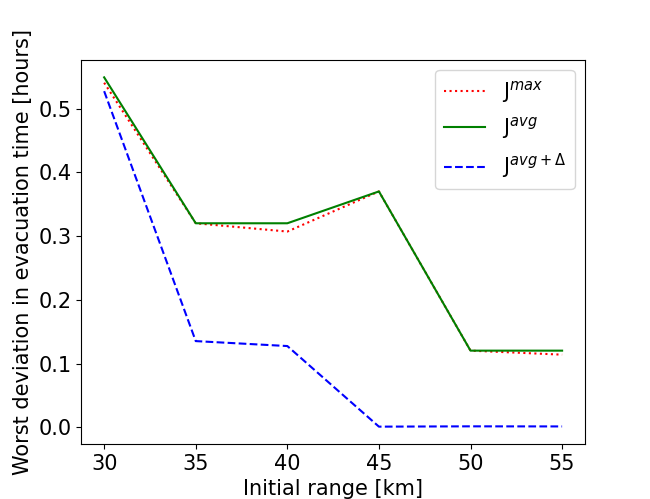}
        \caption{Comparing worst deviation of $t_{evac}$ for different cost indices.} 
    \end{subfigure}
    \caption{ Comparison of different performance indices for the Mariposa network case with respect to the initial range ($r_0$) available to EVs.}
    \label{fig:different_cost_mariposa}
\end{figure}
\begin{figure}[ht]
    \centering
    \begin{subfigure}[t]{0.5\textwidth}
        \centering
        \includegraphics[height=1.85in]{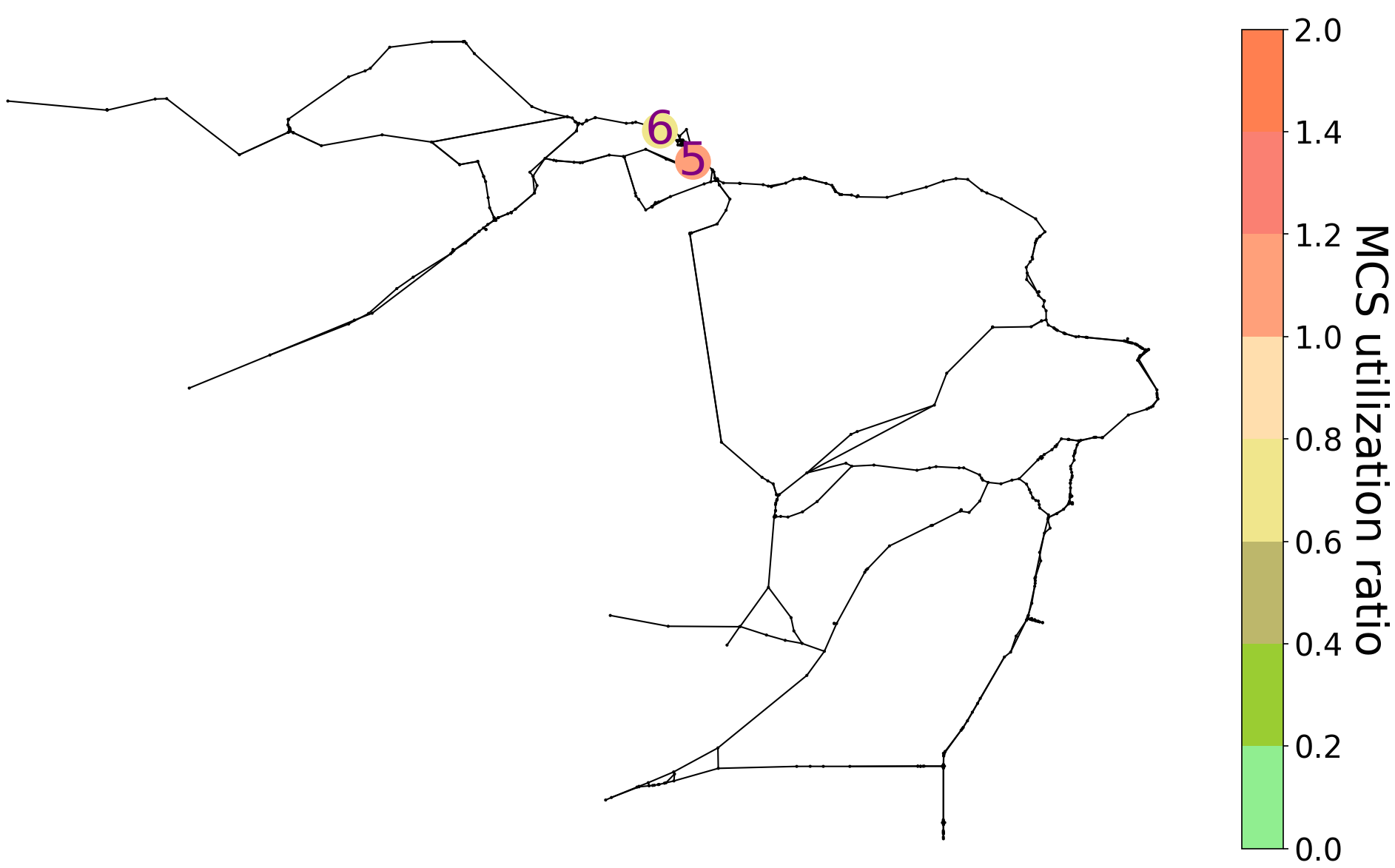}
        \caption{}
    \end{subfigure}
    
    \begin{subfigure}[t]{0.5\textwidth}
        \centering
        \includegraphics[height=1.85in]{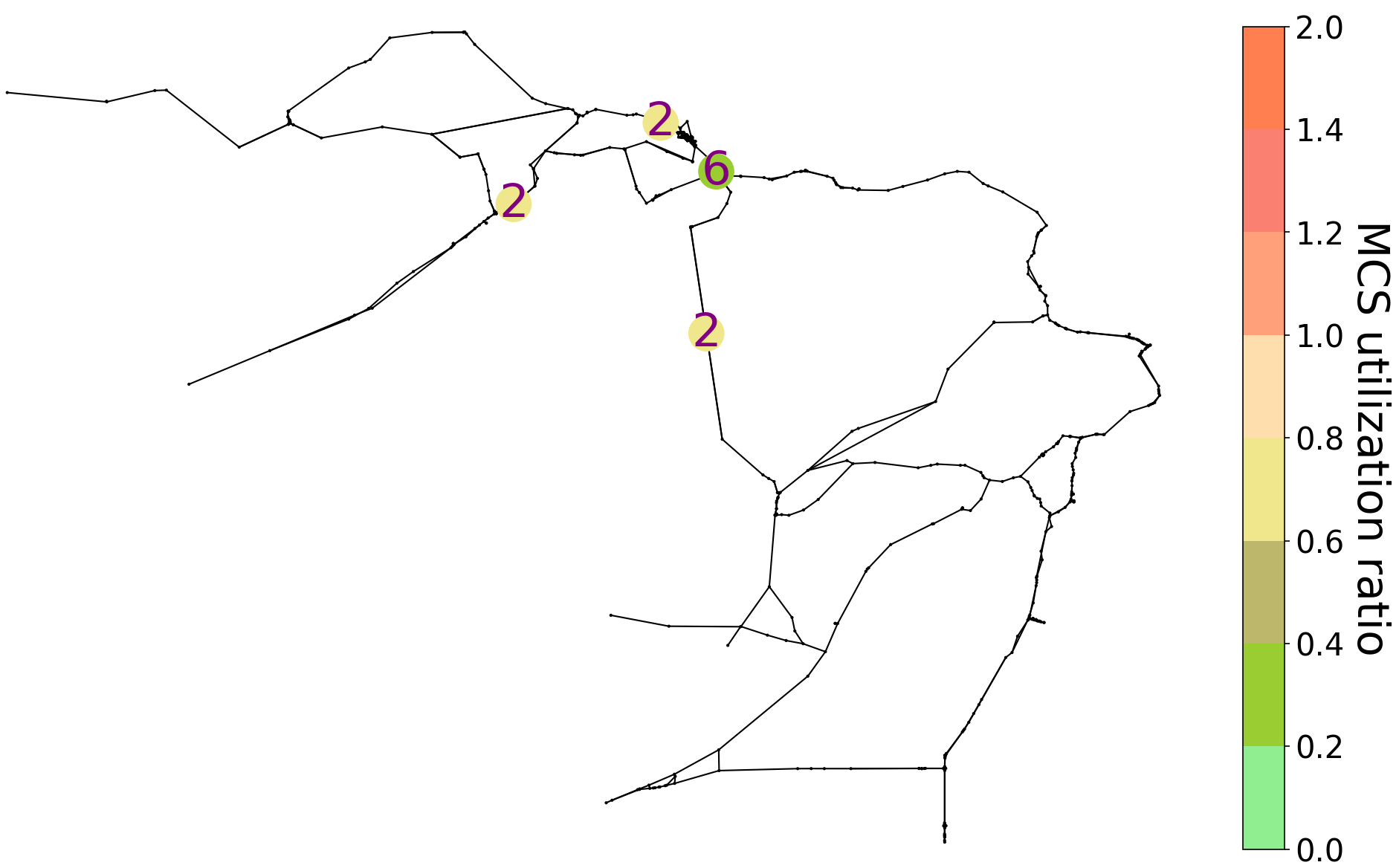}
        \caption{}
    \end{subfigure}
    \caption{Deployment of MCS using (a) baseline algorithm and (b) ECM-based optimization algorithm. Here, the color of the MCS node represents the MCS utilization ratio.}
    \label{fig:mariposa_baseline_comparison}
\end{figure}
\subsection{Mariposa (Rural Example)}
The Mariposa transportation network consists of 632 nodes and 1065 edges (see Fig.~\ref{fig:mariposa_map}). In this section, we consider an emergency evacuation scenario involving 7 evacuation od-pairs. The free flow capacity of each edge, denoted by ($f^{free}_{ij}$), lies in the range of 400 to 1600 vehicles/hour, which is determined by the corresponding road types (e.g., local street vs. major arterial). The locations of the FCS in the Mariposa network are obtained from \cite{plugshare}. We consider only those FCS with DC fast-charging capability (see Fig.~\ref{fig:Mariposa_plot}), which can deliver a charging rate of ($d_{FCS}/T_{FCS}$) equal to 150 km/hr to support rapid evacuation. To reduce exposure to the approaching disaster, all FCS located in downtown Mariposa (near the origins) are assumed to be deactivated. Similar to the Anaheim network, we assume that the state of charge (SOC) of vehicles participating in the emergency evacuation corresponds to an initial driving range ($r_0$) between 10 and 80 km (i.e., $2\%-20\%$ of full battery capacity). The number of vehicles to be evacuated is approximately 2100: 1500 for residents and 600 for visitors staying overnight in the city of Mariposa \cite{Mariposa_survey,Mariposa_visitor_survey}.  We assume that roughly 20$\%$ of this population owns an EV. Therefore, we need to plan the evacuation strategy for around 420 vehicles. Therefore, to evacuate within 1 hour of the disaster alert, the input flow ($f_{in}$) at all 7 od-pairs is assumed to be 60 vehicles/hour. We also use 20 MCS to aid in the evacuation procedures of the network. Each MCS has 5 charging ports and can provide a charging rate of 200 km/hr.

\begin{figure}[ht]
    \centering
    \includegraphics[width=0.9\linewidth]{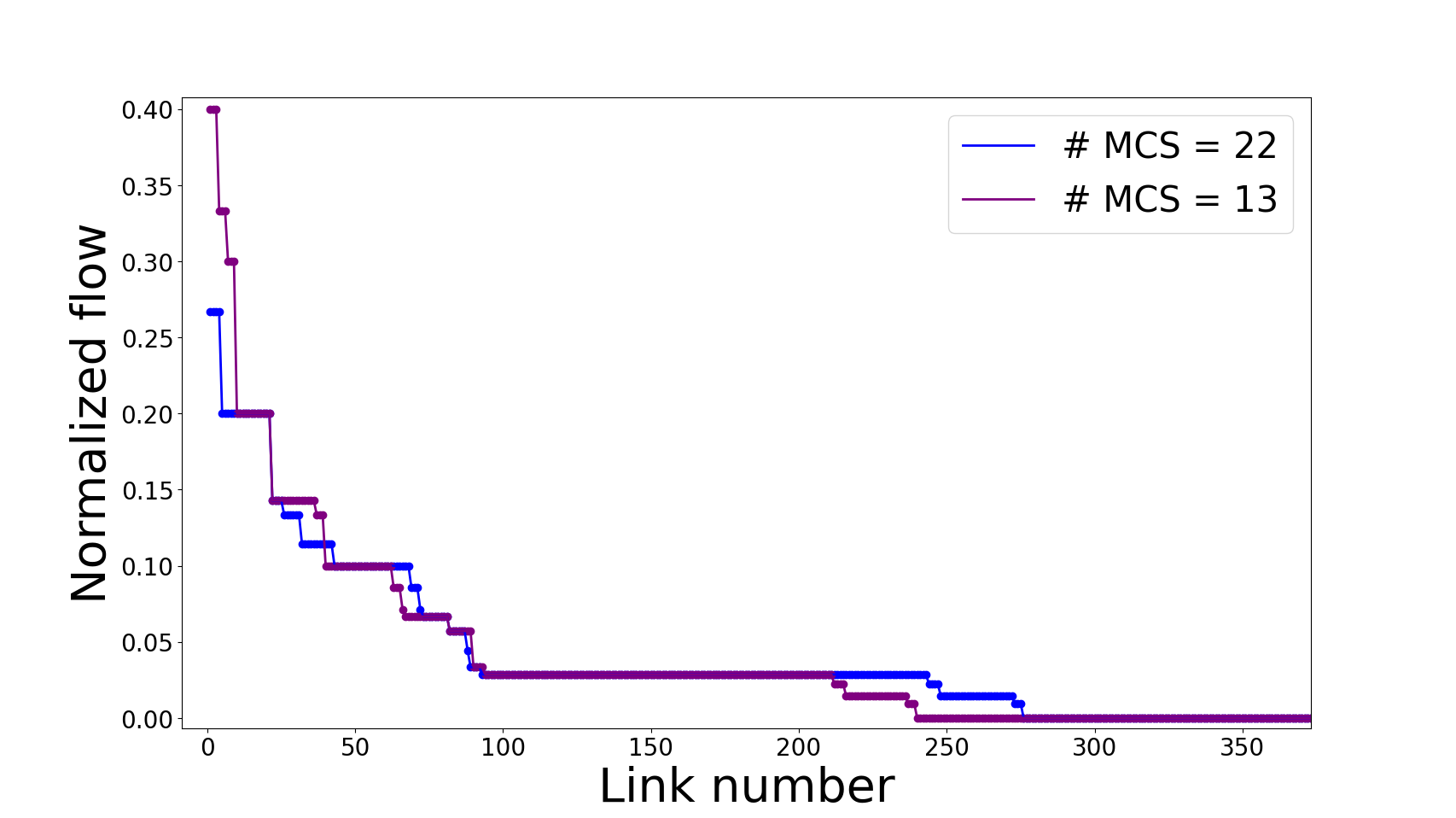}
    \caption{Effect of the number of MCS on the normalised vehicle for each edge in the Mariposa network. }
    \label{fig:MCS_dependency}
\end{figure}

Fig.~\ref{fig:Mariposa_plot} depicts the optimal evacuation route (highlighted in navy blue) for one od-pair of the Mariposa network obtained by solving $\mathbb{P}^{avg}$. In this scenario, the evacuating EVs do not have sufficient SOC to reach any nearby FCS, which are shown as yellow nodes. As a result, the ECM-based optimization strategically deploys 6 MCS, represented by sky-blue nodes, to facilitate evacuation through on-route charging. This example underscores the critical importance of the location-independent nature of MCS, as all available FCS are positioned too far from the origin to provide effective charging support to the evacuating vehicles.

Fig.~\ref{fig:different_cost_mariposa} presents a sensitivity analysis for the Mariposa test case, examining the effect of the initial driving range ($r_0$) on the cost functions of three evacuation strategies: $J^{\max}$, $J^{\mathrm{avg}}$, and $J^{\mathrm{avg}+\Delta}$. As shown in Fig.~\ref{fig:different_cost_mariposa}a, the average evacuation time is minimized when the ECM-based optimization employs $J^{\mathrm{avg}}$ as the cost function. In particular, at $r_0 = 40$ km, the average evacuation time under $J^{\mathrm{avg}}$ is approximately $5.2\%$ lower than that obtained using $J^{\max}$ and $68\%$ lower than that obtained using $J^{\mathrm{avg}+\Delta}$. Fig.~\ref{fig:different_cost_mariposa}b illustrates the worst-case evacuation times for the different cost functions, where the minimum values are achieved using both $J^{\max}$ and $J^{\mathrm{avg}}$. The proximity between the average and worst-case evacuation times under these two cost functions can be attributed to inherent limitations of the Mariposa network, which restrict the extent to which worst-case performance can be improved. In contrast, across all figures, the values associated with $J^{\mathrm{avg}+\Delta}$ are comparatively high, except in Fig.~\ref{fig:different_cost_mariposa}c, where $J^{\mathrm{avg}+\Delta}$ yields the lowest values by explicitly minimizing the deviation from the average evacuation time. This observation further underscores the importance of selecting cost functions based on application-specific objectives and performance priorities. 

In Fig.~\ref{fig:mariposa_baseline_comparison}, we compare the MCS deployment obtained from the proposed ECM-based optimization with that from a baseline evacuation algorithm. The baseline algorithm deploys MCS to the nearest candidate charging location relative to the origin of each od-pair. If the nearest location reaches its MCS capacity limit, the algorithm terminates further deployment without considering alternative, more distant sites. Consequently, the baseline approach follows a greedy strategy, which leads to significant queuing of EVs at nearby charging locations. Fig.~\ref{fig:mariposa_baseline_comparison}a and Fig.~\ref{fig:mariposa_baseline_comparison}b illustrate the spatial distribution of MCS deployments on the Mariposa network for the baseline and ECM-based optimization methods, respectively. The color of each MCS node represents its utilization ratio, defined as the ratio between the incoming vehicle charging demand rate and the MCS service rate. Under the baseline algorithm, all MCS are deployed at only two locations, causing the available parking capacity at these sites to saturate rapidly and resulting in insufficient charging service at one location. This is reflected in the utilization ratio exceeding one, indicating congestion and queue buildup. In contrast, the ECM-based optimization explicitly accounts for the initial driving ranges of EVs at the origin nodes and deploys MCS in a more distributed manner to mitigate both road congestion and charging queues. This is evidenced by the utilization ratios at all MCS locations remaining below one, indicating that the charging demand is adequately served without excessive queuing.     

Fig.~\ref{fig:MCS_dependency} illustrates the dependence of normalized vehicle flow, sorted in descending order, across individual edges of the Mariposa network as a function of the number of MCS deployed during the evacuation. The results show that when a larger number of MCS (e.g., 22) is utilized, the vehicle flow on each edge is relatively lower compared to scenarios with fewer MCS (e.g., 13). This behavior arises because a higher availability of MCS allows them to be distributed across multiple locations and edges, providing EV drivers with greater routing flexibility when traveling toward the evacuation zone. In contrast, when fewer MCS are deployed, charging options become limited, causing a larger fraction of EVs to converge onto the same edges to access the available MCS for charging, thereby increasing traffic concentration and edge utilization.

\subsection{High-fidelity Validation}
To validate the ECM optimization results at a granular level, a micro-simulation model for the Mariposa network was created using the commercial simulation software Aimsun \cite{Aimsun_info}, which is shown in Fig.~\ref{fig:Mariposa_aimsun_network}. Aimsun has implemented the Microscopic Free-flow aCceleration (MFC) Model and the Battery Consumption Model to enable the simulation of EVs (passenger cars only) at a network scale. As illustrated in Fig.~\ref{fig:EV_and_charging_station_aimsun}, each EV has its battery capacity determined by its vehicle type, and its current state of charge (SOC) and total energy consumed are tracked and updated by Aimsun at every simulation time step. Furthermore, a new module was developed using Aimsun's Application Programming Interface (API) to enable the simulation of real-world charging behavior. As shown in Fig.~\ref{fig:EV_and_charging_station_aimsun}, charging stations are modeled as centroids. Both fixed and mobile charging stations can be deployed at the same location (i.e., centroid), with different configurations of charging ports and charger powers (e.g., DCFC and L2). An EV with a low SOC can be routed to its nearest available charging station or to a designated charging station determined by an optimization algorithm (e.g., the proposed algorithm in this study). A charging station will assign available chargers to EVs, queue EVs when all chargers are in use, monitor the remaining battery of mobile charging stations, and disable them when they run out of battery. The Aimsun mobility simulator can be used to replicate realistic traffic conditions, including road congestion, stop signs, traffic signals, and speed reductions during turning maneuvers. Table~\ref{table:evac_time_validated} presents the evacuation times for the seven od-pairs in the Mariposa evacuation scenario obtained using the proposed ECM-based optimization framework and compares them with the corresponding simulation results generated in Aimsun by implementing the same evacuation routes. The differences in evacuation times vary across the od-pairs. For od-pairs 3, 4, 6, and 7, which correspond to shorter travel distances (see Fig.~\ref{fig:Mariposa_full}), the error remains below 6.67$\%$. In contrast, for od-pairs 1, 2, and 5, which involve longer travel distances, the error increases to approximately 13.84$\%$ or larger. This trend is expected because longer routes encounter a greater number of traffic control elements and more frequent speed variations due to road conditions, which are captured in the microscopic simulation but simplified in the optimization model.

Table~\ref{table:Aimsun_comparison_baseline} further compares the evacuation times obtained using a baseline algorithm and the ECM-based optimization after implementing both strategies in Aimsun. The results clearly demonstrate improved evacuation performance for the ECM-based approach, which leverages information about the EVs’ initial SOC to deploy MCS more efficiently. This targeted deployment reduces roadway congestion and charging station queuing, thereby leading to shorter evacuation times.

\begin{figure}[ht]
    \centering
    \includegraphics[width=0.75\linewidth]{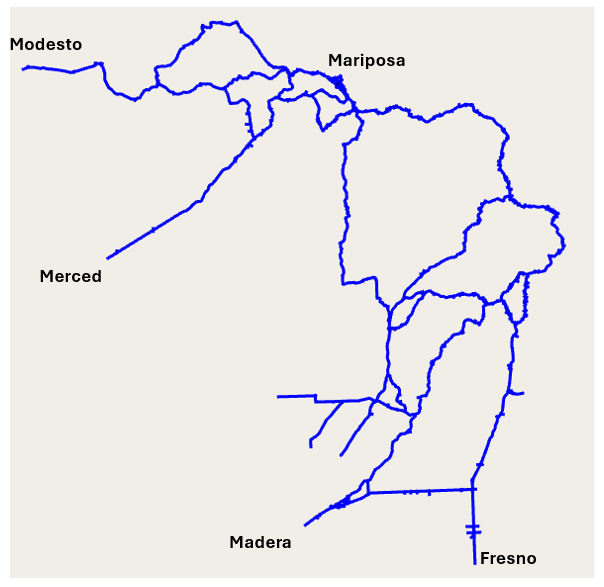}
    \caption{The Mariposa network in Aimsun.}
    \label{fig:Mariposa_aimsun_network}
\end{figure}

\begin{figure}[ht]
    \centering
    \includegraphics[width=0.9 \linewidth]{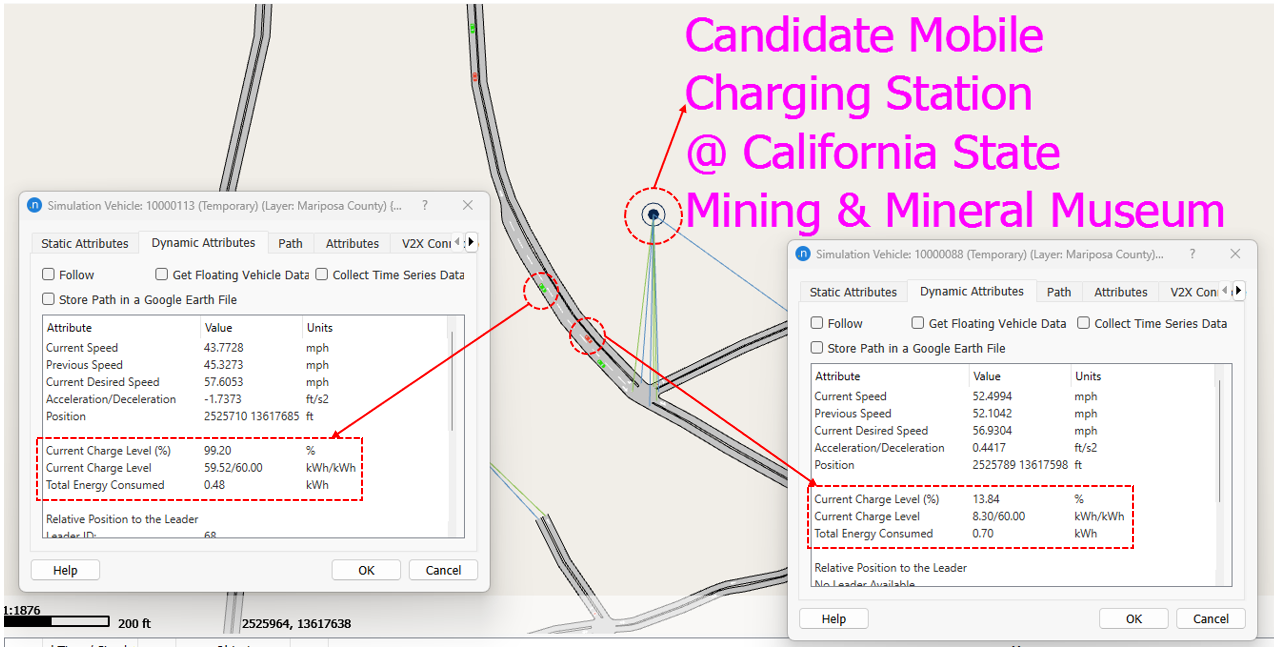}
    \caption{Electric vehicles and charging stations in Aimsun.}
    \label{fig:EV_and_charging_station_aimsun}
\end{figure}
\begin{table}[!htp]
\centering
\caption{Validation of evacuation time obtained from ECM-based optimization in Aimsun}
\label{table:evac_time_validated}
\begin{tabular}{c c c c}
\toprule
$\#$ od-pair & Flow-based model & Microscopic (Aimsun) & Error [\%]\\
\midrule
1    & 0.97 hrs & 1.15 hrs & 15.65\\
2    & 1.12 hrs & 1.3 hrs & 13.84\\
3  & 0.76 hrs & 0.79 hrs & 3.79\\
4  & 0.59 hrs & 0.61 hrs & 3.27\\
5  & 0.73 hrs & 0.94 hrs & 22.34\\
6  & 0.64 hrs & 0.68 hrs & 5.88\\
7  & 0.7 hrs & 0.75 hrs & 6.67\\
\bottomrule
\end{tabular}
\end{table}
\begin{table}[!htp]
\centering
\caption{Aimsun evacuation time comparison between baseline and ECM Algorithm}
\label{table:Aimsun_comparison_baseline}
\begin{tabular}{c c c c}
\toprule
$\#$ od-pair & Baseline & ECM-optimization & Improvement [\%]\\
\midrule
1    & 1.17 hrs & 1.15 hrs & 1.7\\
2    & 1.34 hrs & 1.3 hrs & 2.98\\
3  & 1.31 hrs & 0.79 hrs & 39.69\\
4  & 1.08 hrs & 0.61 hrs & 43.51\\
5  & 1.47 hrs & 0.94 hrs & 36.05\\
6  & 1.17 hrs & 0.68 hrs & 41.88\\
7  & 0.82 hrs & 0.75 hrs & 8.53\\
\bottomrule
\end{tabular}
\end{table}
\section{Conclusion}
This paper presented a novel ECM–based framework for EVs emergency evacuation that jointly addresses routing, charging, and infrastructure constraints under disaster scenarios. By exploiting analogies between transportation networks and electrical circuits, the proposed approach models traffic flow, travel time, and EV driving range using Kirchhoff’s laws, enabling a unified and physically interpretable optimization formulation. Unlike conventional evacuation models, the framework explicitly captures energy feasibility through voltage-based range modeling and incorporates both FCS and MCS to mitigate limited charging accessibility and infrastructure disruptions.

The resulting integer programming framework determines optimal evacuation routes, charging durations, and the placement and number of MCS units while respecting road capacity and charging service constraints. The extension to multiple origin–destination pairs using the principle of superposition enables scalable evacuation planning and supports fairness-aware objectives, including worst-case, average, and variance-based evacuation times. Simulation results on large-scale transportation networks in California demonstrate that the proposed method significantly improves evacuation efficiency and robustness, particularly in low initial-state-of-charge scenarios, where traditional infrastructure alone is insufficient.

Future work will focus on incorporating dynamic traffic congestion models, stochastic disaster impacts, real-time MCS reallocation, and tighter coupling with power system constraints to further enhance the resilience and practicality of EV-based evacuation planning.

\bibliographystyle{IEEEtran}
\bibliography{reference}
\section*{Appendix}
\newpage
\begin{figure*}[ht]
    \centering
    \begin{subfigure}[t]{0.49\textwidth}
        \centering
        \includegraphics[width=\textwidth]{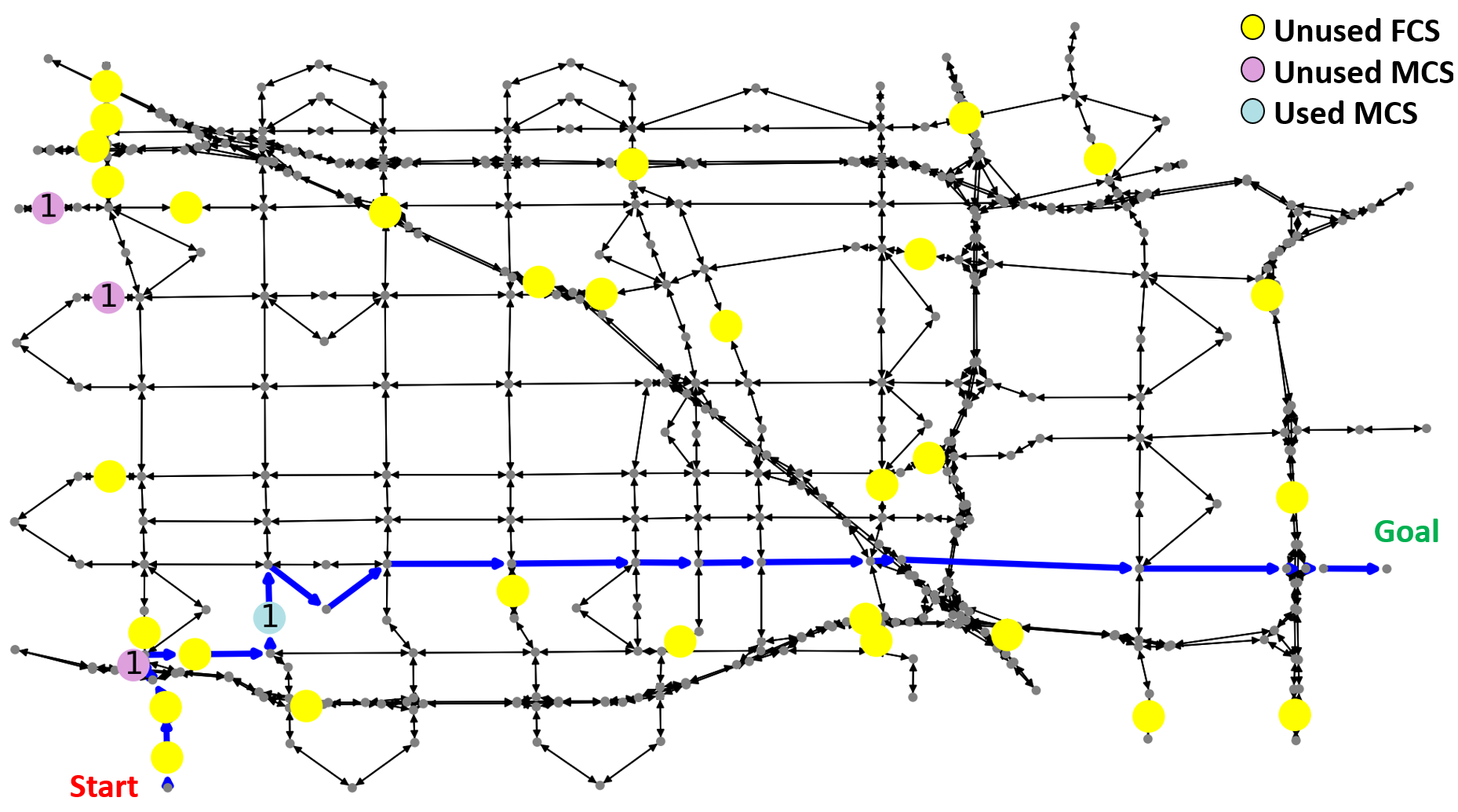}
        \caption{}
    \end{subfigure}
        \centering
    \begin{subfigure}[t]{0.49\textwidth}
        \centering
        \includegraphics[width=\textwidth]{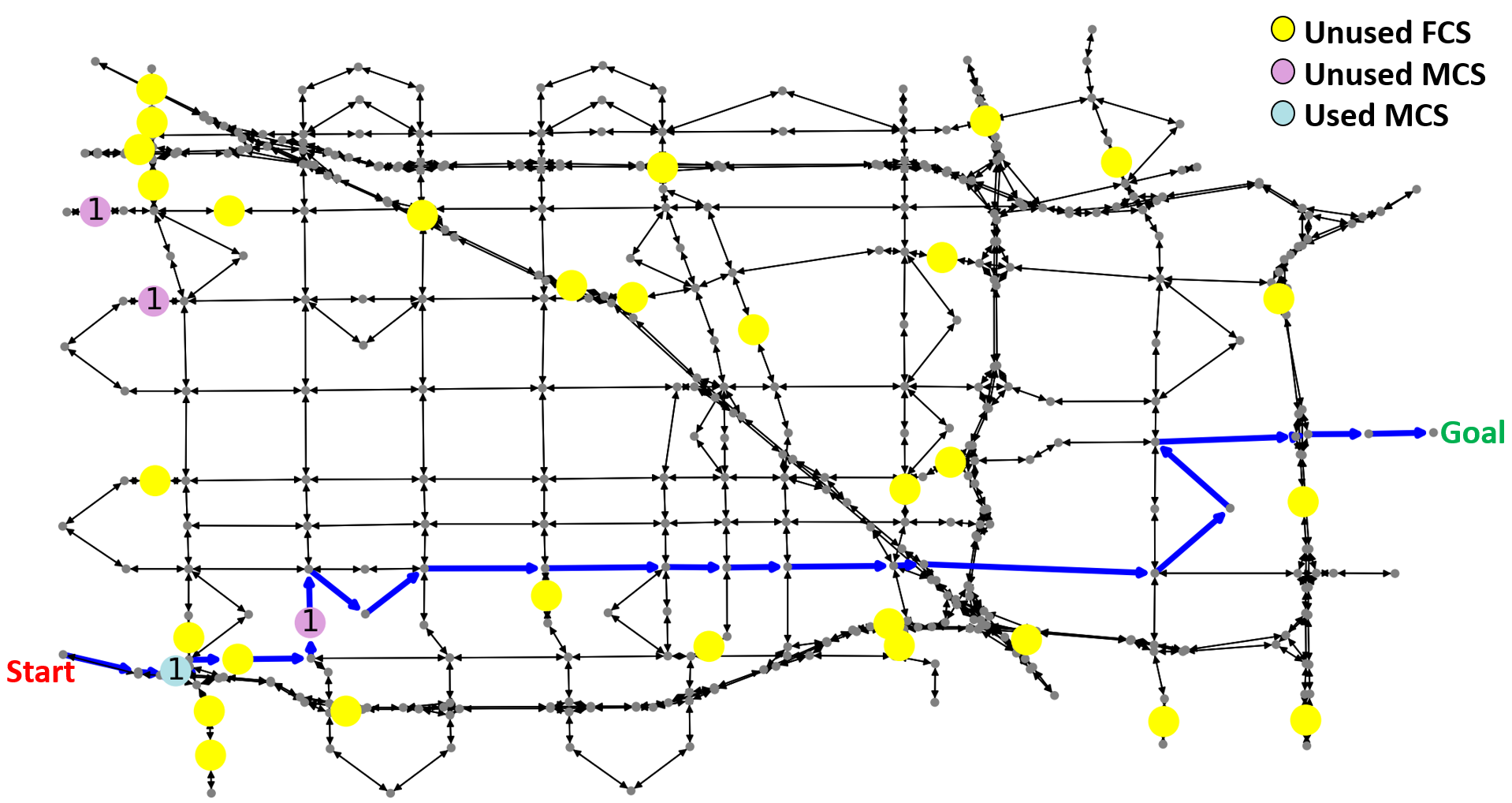}
        \caption{}
    \end{subfigure}
    \begin{subfigure}[t]{0.49\textwidth}
        \centering
        \includegraphics[width=\textwidth]{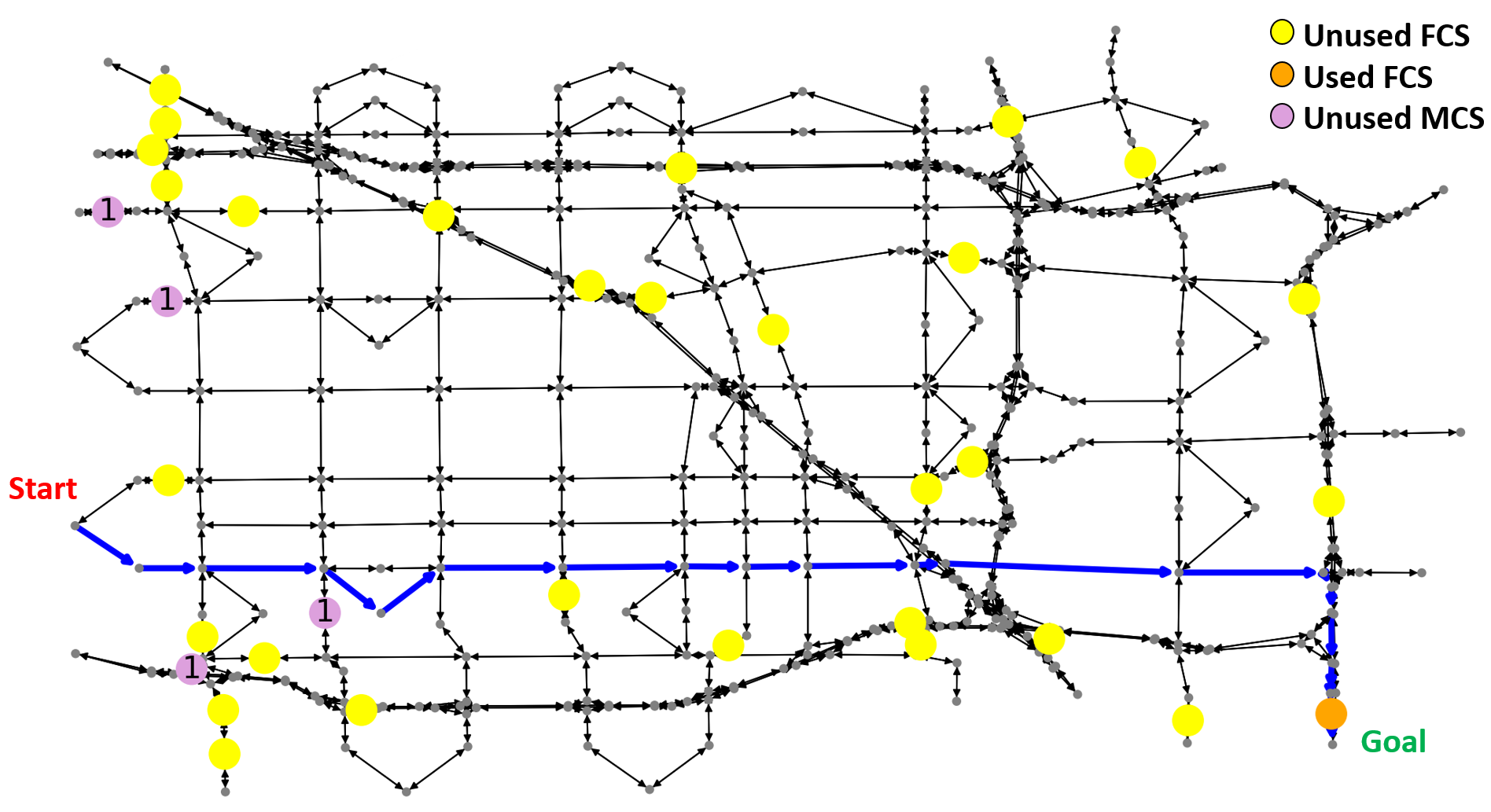}
        \caption{}
    \end{subfigure}
        \begin{subfigure}[t]{0.49\textwidth}
        \centering
        \includegraphics[width=\textwidth]{Figures/Anaheim_sol_4.png}
        \caption{}
    \end{subfigure}
        \begin{subfigure}[t]{0.49\textwidth}
        \centering
        \includegraphics[width=\textwidth]{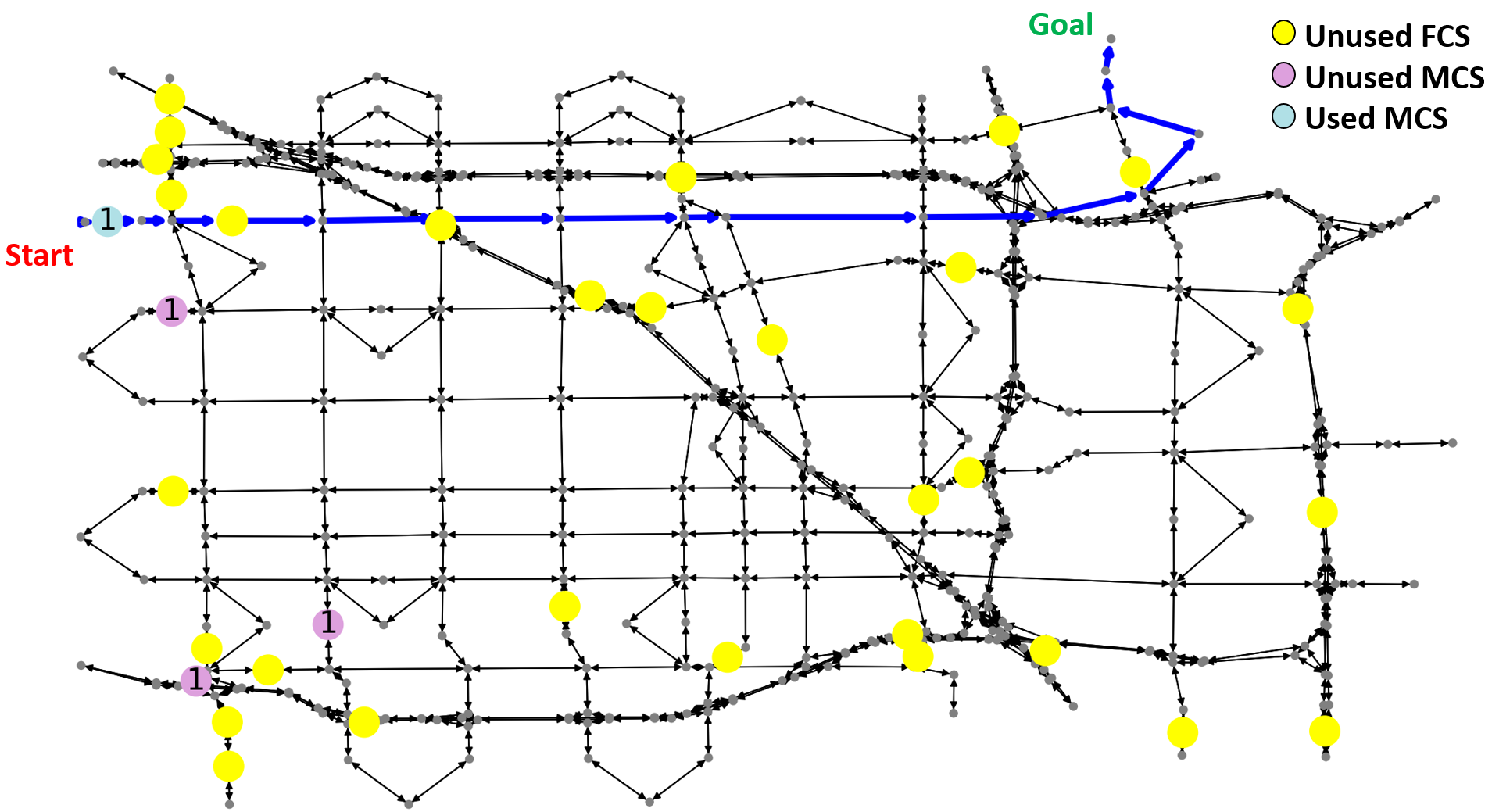}
        \caption{}
    \end{subfigure}
        \begin{subfigure}[t]{0.49\textwidth}
        \centering
        \includegraphics[width=\textwidth]{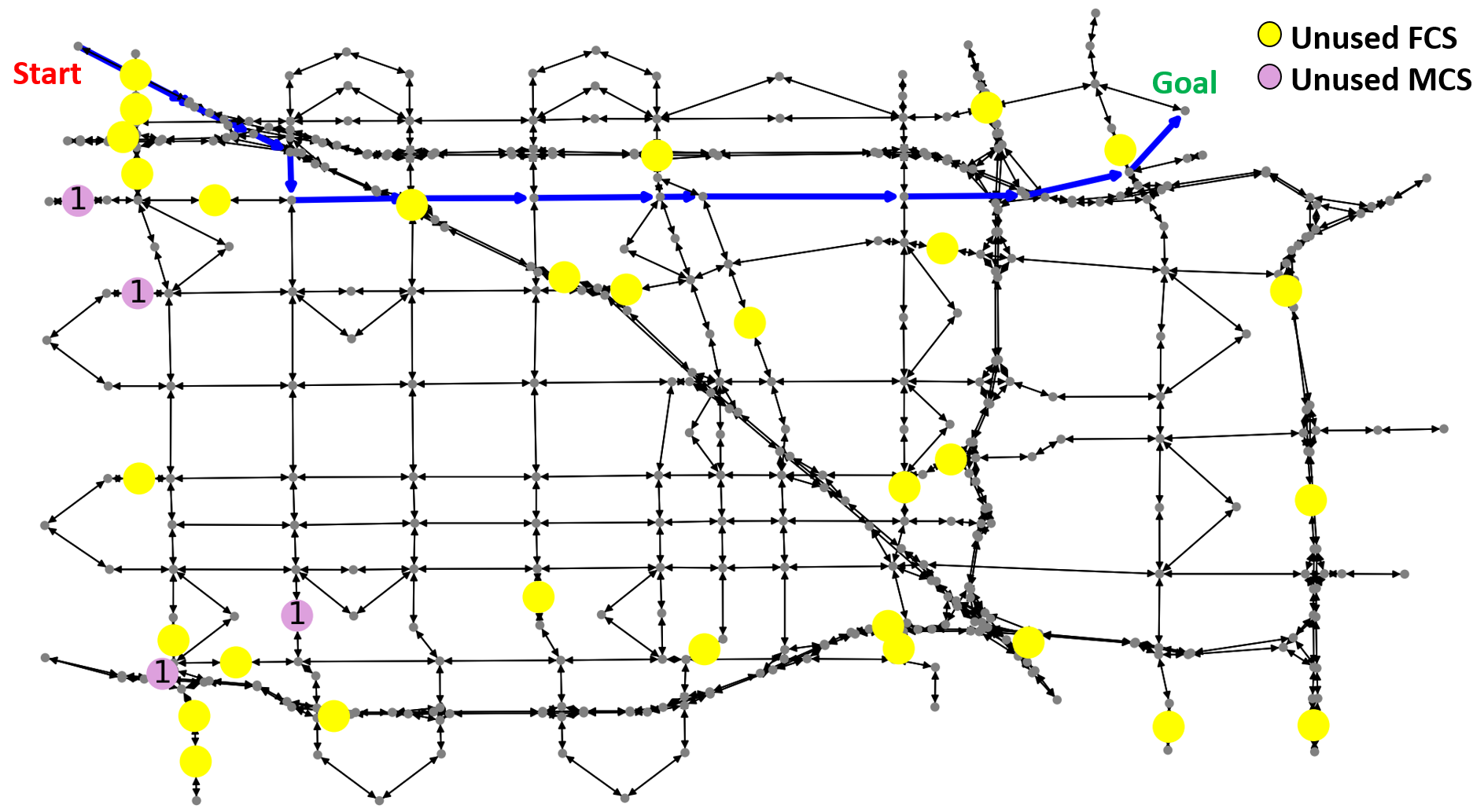}
        \caption{}
    \end{subfigure}
        \begin{subfigure}[t]{0.49\textwidth}
        \centering
        \includegraphics[width=\textwidth]{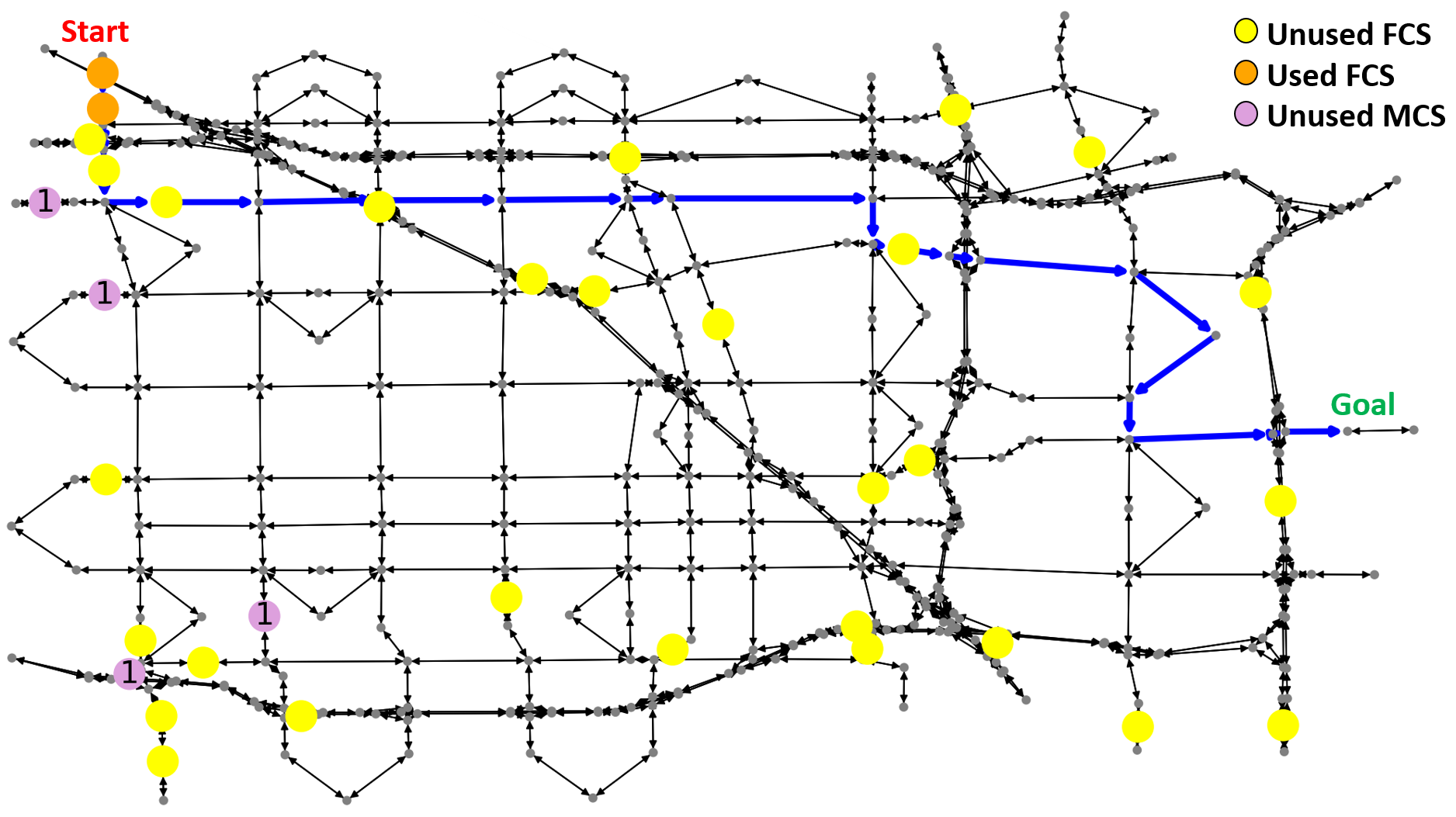}
        \caption{}
    \end{subfigure}
        \begin{subfigure}[t]{0.49\textwidth}
        \centering
        \includegraphics[width=\textwidth]{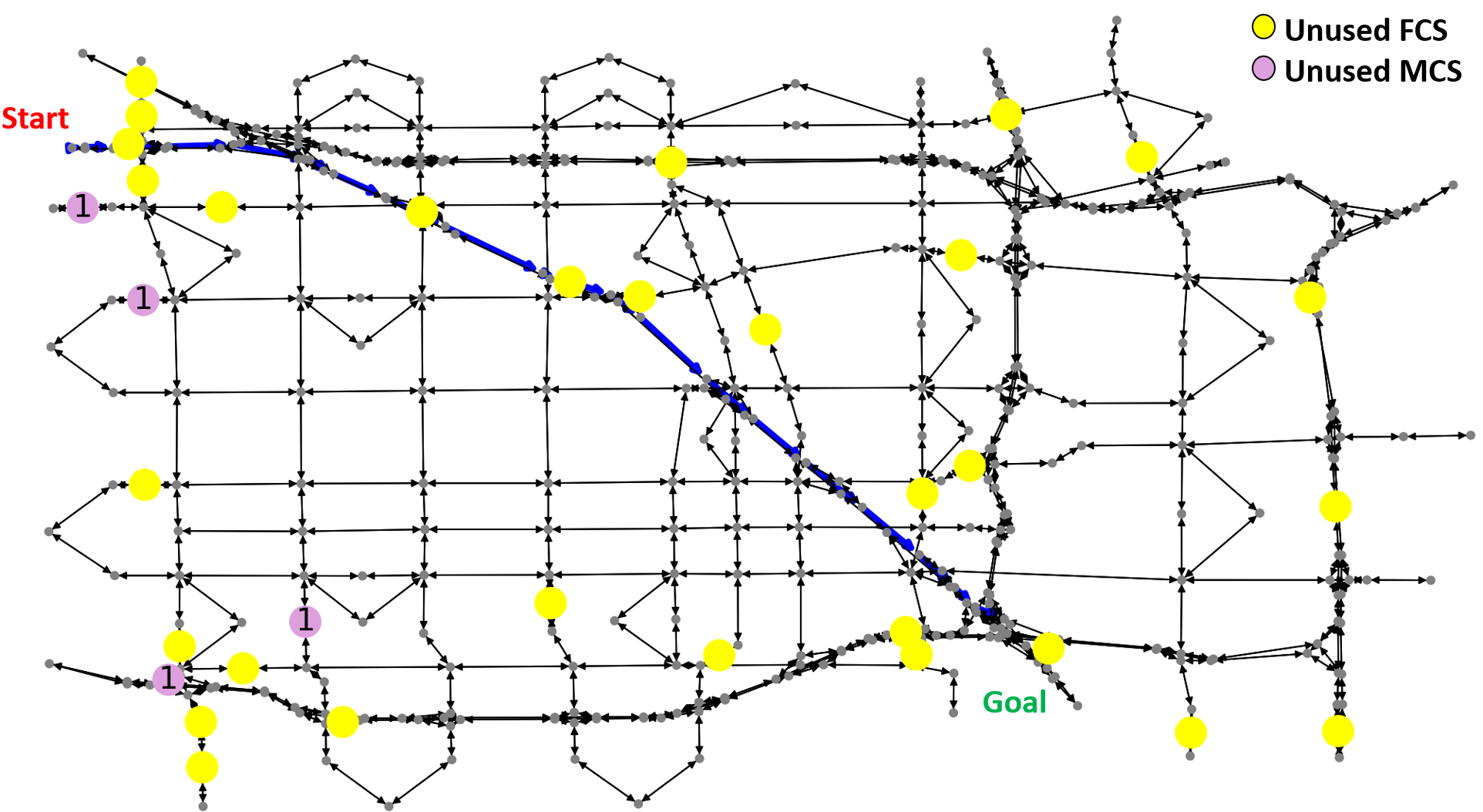}
        \caption{}
    \end{subfigure}
    \caption{Optimal evacuation paths of the 8 od-pairs of the Anaheim network obtained as the solution of ECM-based optimization. }
    \label{fig:anaheim_full}
\end{figure*}
\begin{figure*}[ht]
    \centering
    \begin{subfigure}[t]{0.49\textwidth}
        \centering
        \includegraphics[width=\textwidth]{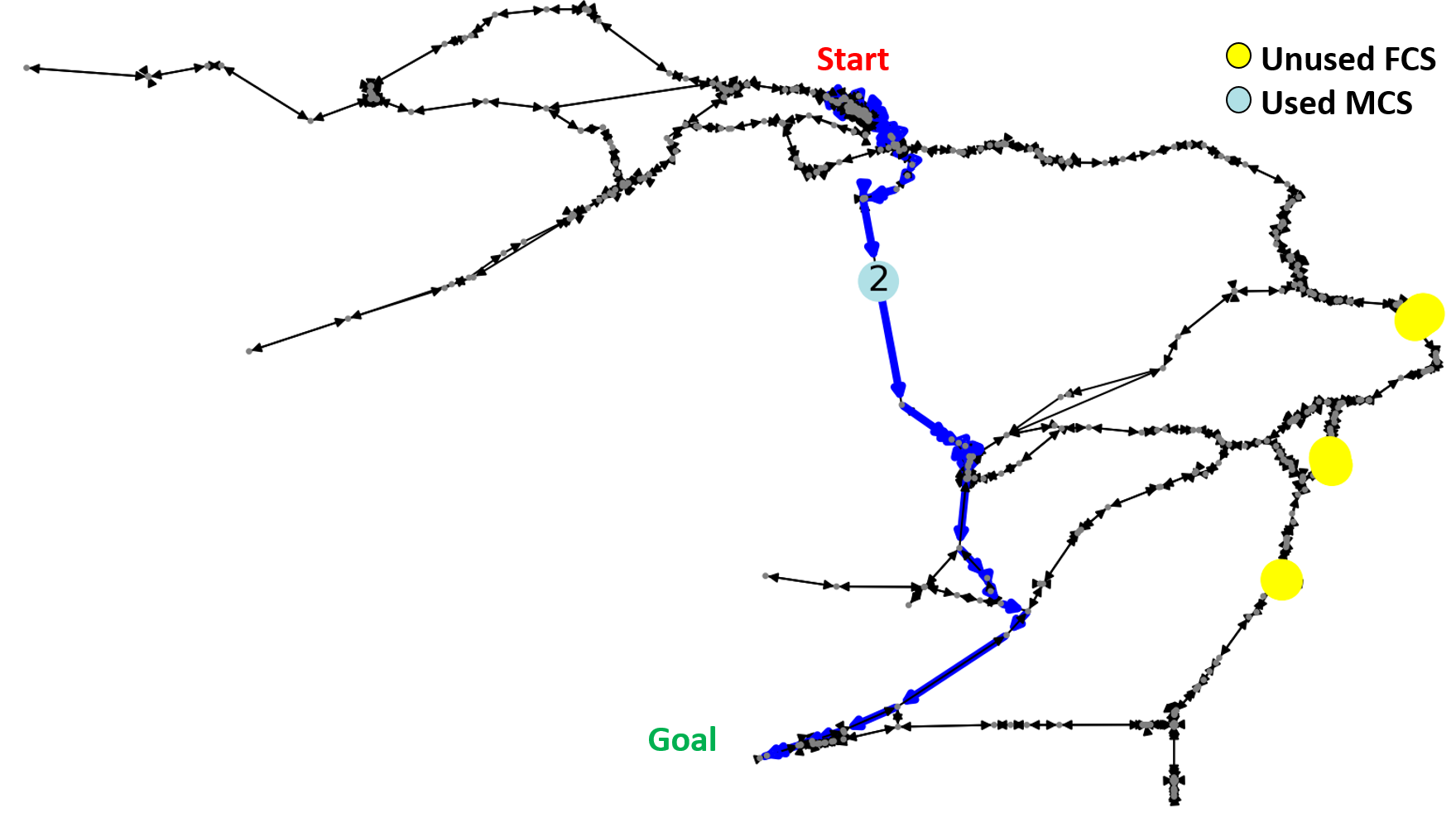}
        \caption{}
    \end{subfigure}
        \centering
    \begin{subfigure}[t]{0.49\textwidth}
        \centering
        \includegraphics[width=\textwidth]{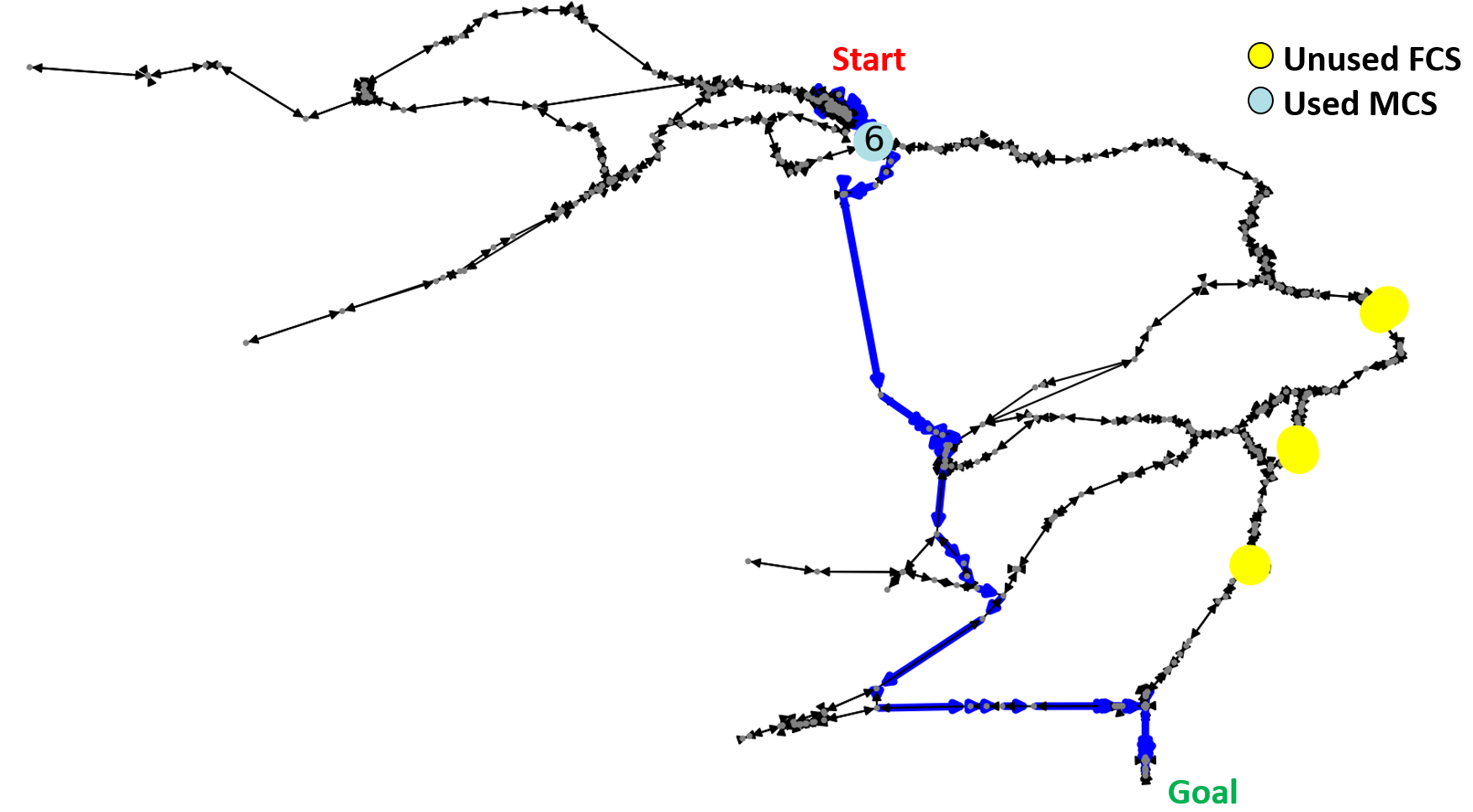}
        \caption{}
    \end{subfigure}
    \begin{subfigure}[t]{0.49\textwidth}
        \centering
        \includegraphics[width=\textwidth]{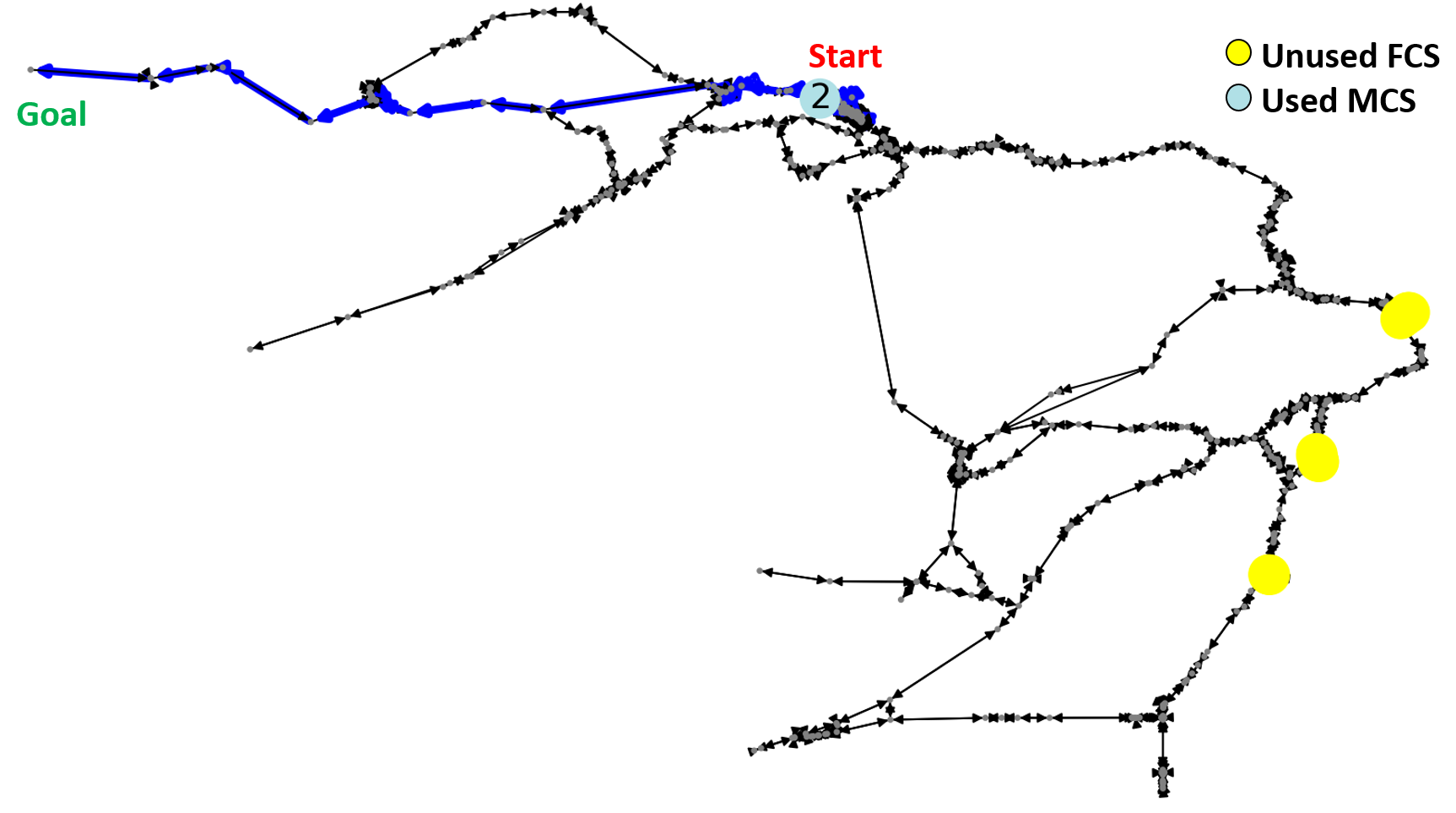}
        \caption{}
    \end{subfigure}
        \begin{subfigure}[t]{0.49\textwidth}
        \centering
        \includegraphics[width=\textwidth]{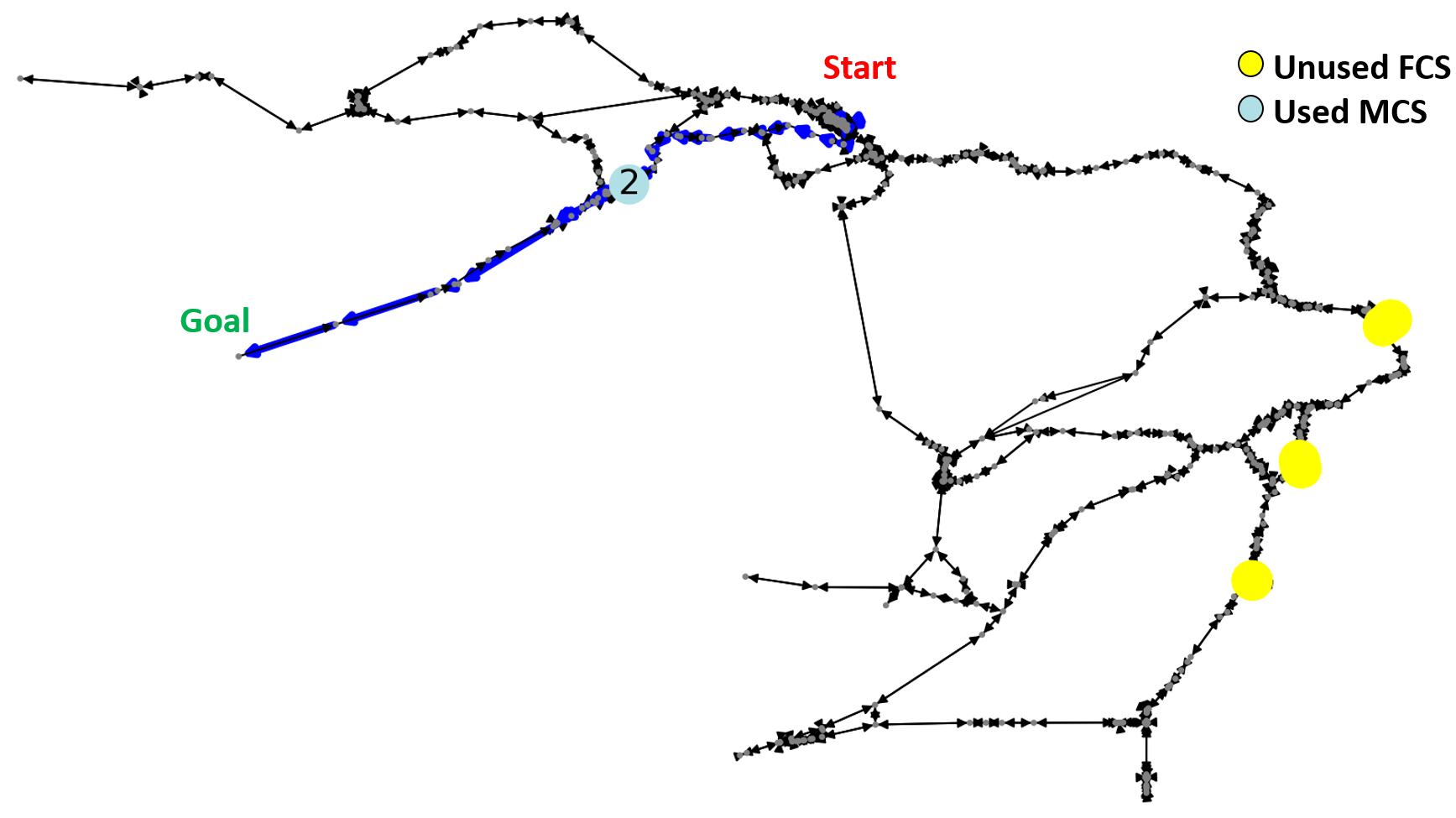}
        \caption{}
    \end{subfigure}
        \begin{subfigure}[t]{0.49\textwidth}
        \centering
        \includegraphics[width=\textwidth]{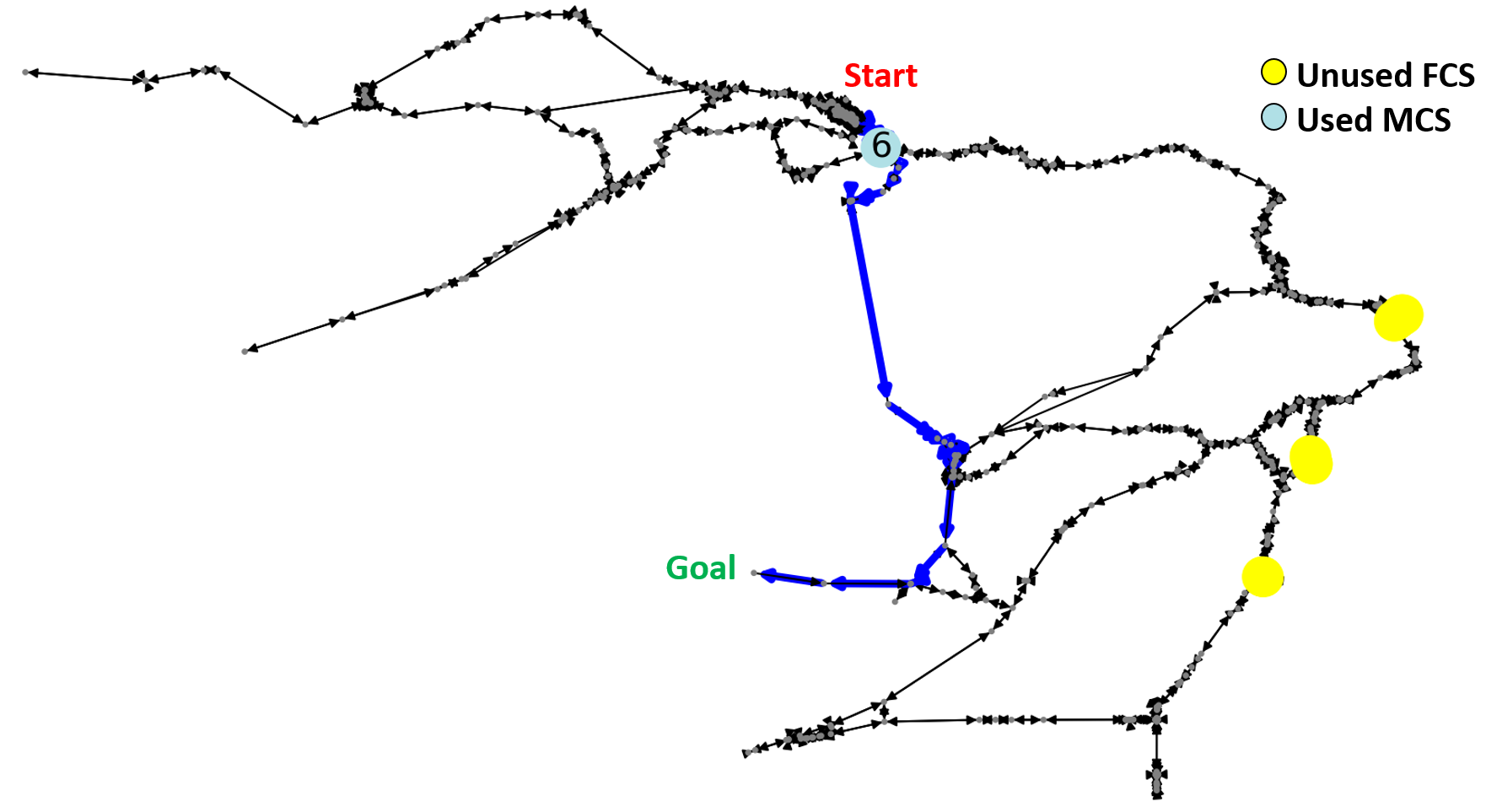}
        \caption{}
    \end{subfigure}
        \begin{subfigure}[t]{0.49\textwidth}
        \centering
        \includegraphics[width=\textwidth]{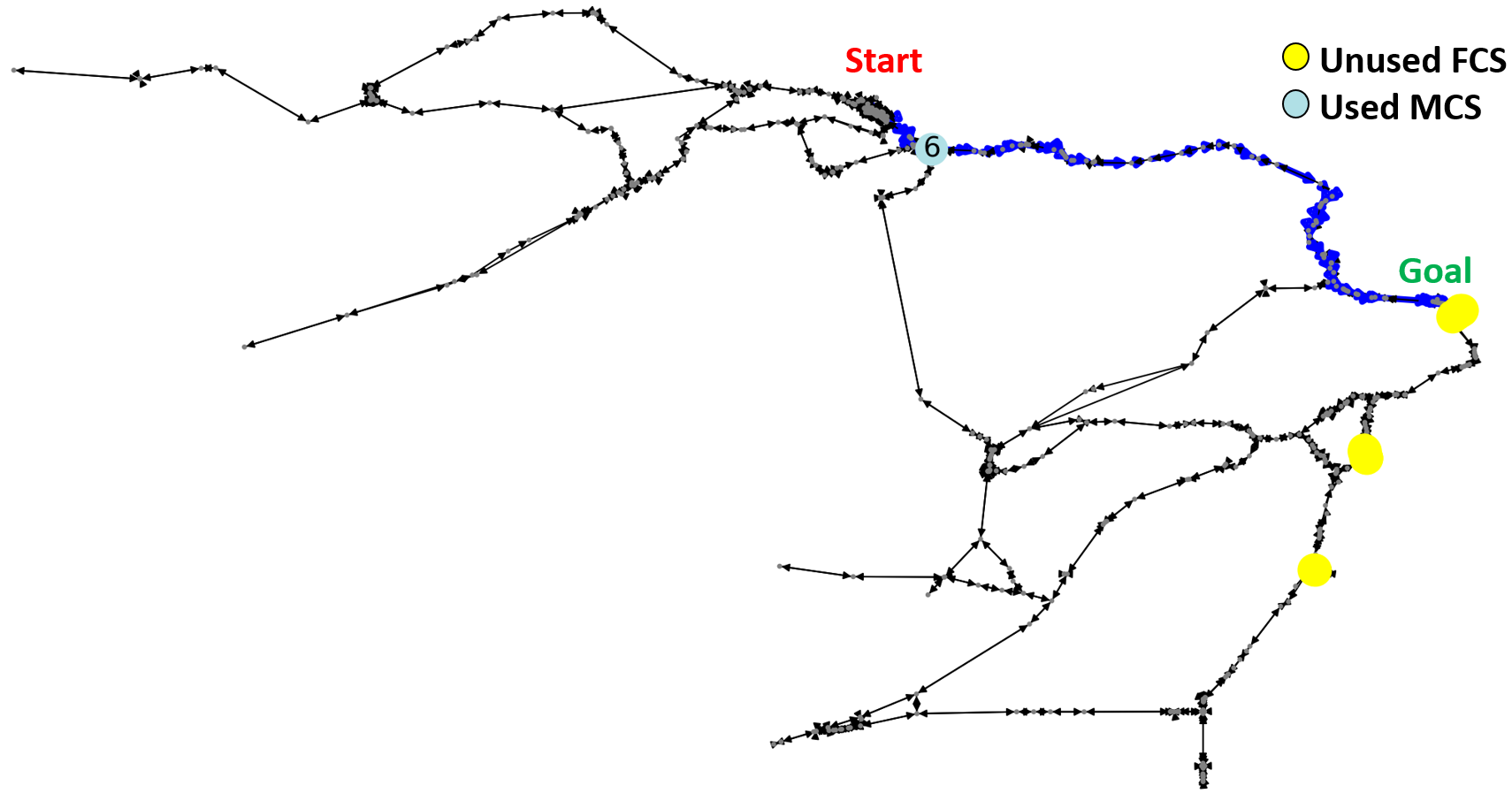}
        \caption{}
    \end{subfigure}
        \begin{subfigure}[t]{0.49\textwidth}
        \centering
        \includegraphics[width=\textwidth]{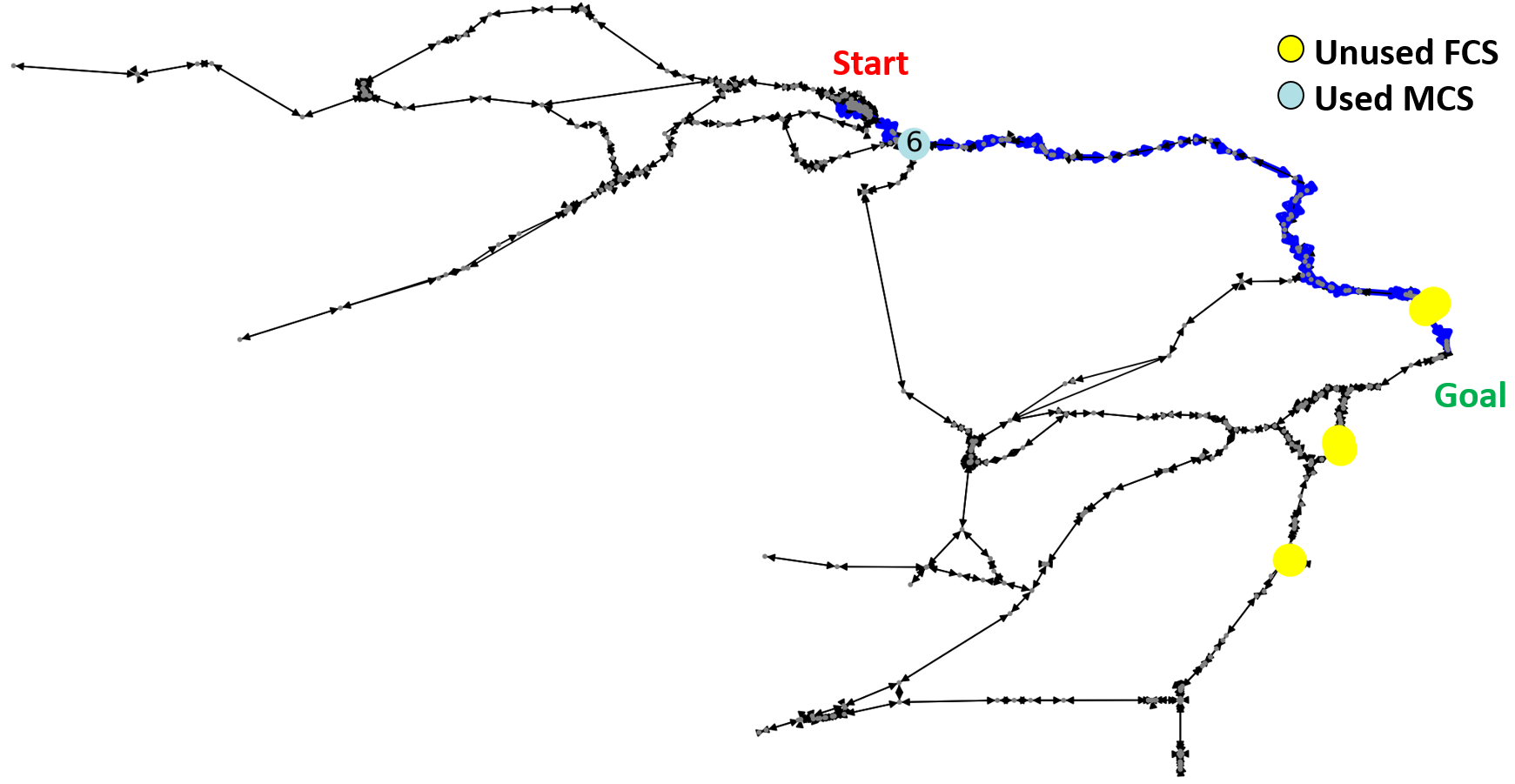}
        \caption{}
    \end{subfigure}

    \caption{Optimal evacuation paths of the 7 od-pairs of the Mariposa network obtained as the solution of ECM-based optimization. }
    \label{fig:Mariposa_full}
\end{figure*}
\end{document}